\documentclass[10pt, sigconf]{acmart}

\usepackage{booktabs} % For formal tables
\usepackage{algorithm, algpseudocode, color,amsmath, graphicx,subcaption,listings,xspace, url, multirow, tabularx, balance, bm, mathtools, flushend}
\usepackage[inline]{enumitem}
\usepackage[normalem]{ulem}
\DeclarePairedDelimiter{\ceil}{\lceil}{\rceil}

\newif\ifpaper
\paperfalse   %%FULL VERSION
\ifpaper
\else
\renewcommand\footnotetextcopyrightpermission[1]{} % removes footnote with conference information in first column
\fi

\copyrightyear{2019} 
\acmYear{2019} 
\setcopyright{acmcopyright}
\acmConference[SIGMOD '19]{2019 International Conference on Management of Data}{June 30-July 5, 2019}{Amsterdam, Netherlands}
\acmBooktitle{2019 International Conference on Management of Data (SIGMOD '19), June 30-July 5, 2019, Amsterdam, Netherlands}
\acmPrice{15.00}
\acmDOI{10.1145/3299869.3300092}
\acmISBN{978-1-4503-5643-5/19/06}
\newcommand{\todo}[1]{{\color{red} TODO: #1}}
\newcommand{\eat}[1]{}
\newcommand{\am}[1]{{\color{blue} \emph{[[AM: #1]]}}}
\newcommand{\change}[1]{{\color{black} #1}}
\newcommand{\conf}[1]{{\color{black} #1}}
\newcommand{\full}[1]{{\color{black} #1}}

\newcommand{\census}{{\tt Adult}\xspace}
\newcommand{\location}{{\tt NYTaxi}\xspace}

\newcommand{\tabledomain}{\mathcal{D}}
\newcommand{\attr}{A}
\newcommand{\attrlist}{attr(R)}
\newcommand{\rowdomain}{dom(R)}
\newcommand{\workload}{{\bf W}}
\newcommand{\wq}{{\bf w}}
\newcommand{\data}{{\bf x}}
\newcommand{\transform}{\mathcal{T}}
\newcommand{\strategy}{{\bf A}}
\newcommand{\noisev}{{\bm \eta}}
\newcommand{\ans}{{\bf y}}

\newcommand{\trans}{\mathbb{T}}
\newcommand{\simf}{\mathbb{S}}
\newcommand{\simt}{\Theta}

\newcommand{\cleanermodel}{\mathcal{C}}

\newcommand{\system}{$A$PEx\xspace}

\newcommand{\lcc}{LCC\xspace}

\newcommand{\wcq}{WCQ\xspace}
\newcommand{\icq}{ICQ\xspace}
\newcommand{\tcq}{TCQ\xspace}

\newcommand{\transFunc}{\textsc{translate}}
\newcommand{\runFunc}{\textsc{run}}

\newcommand{\relaxprivacy}{\textsc{RelaxPrivacy}}
\newcommand{\upperbound}{\epsilon^{u}}
\newcommand{\lowerbound}{\epsilon^{l}}
\newcommand{\squishlist}{
   \begin{list}{$\bullet$}
    {
      \setlength{\itemsep}{0pt}
      \setlength{\parsep}{3pt}
      \setlength{\topsep}{3pt}
      \setlength{\partopsep}{0pt}
      \setlength{\leftmargin}{1.5em}
      \setlength{\labelwidth}{1em}
      \setlength{\labelsep}{0.5em} } }

\newcommand{\squishend}{
    \end{list}  }

\newcommand{\squishenum}{
   
   \begin{list}{scount}{ \usecounter{scount}}
    {
      \setlength{\itemsep}{0pt}
      \setlength{\parsep}{3pt}
      \setlength{\topsep}{3pt}
      \setlength{\partopsep}{0pt}
      \setlength{\leftmargin}{1.5em}
      \setlength{\labelwidth}{1em}
      \setlength{\labelsep}{0.5em} } }

\newcommand{\stitle}[1]{\smallskip \noindent{\bf #1}}

\begin{document}
\title{\system: Accuracy-Aware Differentially \\ Private Data Exploration}

\author{Chang Ge}
\affiliation{%
  \institution{University of Waterloo}
}
\email{c4ge@uwaterloo.ca}

\author{Xi He}
\affiliation{%
  \institution{University of Waterloo}
}
\email{xi.he@uwaterloo.ca}

\author{Ihab F. Ilyas}
\affiliation{%
  \institution{University of Waterloo}
}
\email{ilyas@uwaterloo.ca}

\author{Ashwin Machanavajjhala}
\affiliation{%
  \institution{Duke University}
}
\email{ashwin@cs.duke.edu}

% The default list of authors is too long for headers.
\renewcommand{\shortauthors}{Ge, He, Ilyas and Machanavajjhala}

%!TEX root=./main.tex
\begin{abstract}
Organizations are increasingly interested in allowing external data scientists to explore their sensitive datasets. Due to the popularity of differential privacy, data owners want the data exploration to ensure provable privacy guarantees. However, current systems for answering queries with differential privacy place an inordinate burden on the data analysts to understand differential privacy, manage their privacy budget, and even implement new algorithms for noisy query answering. Moreover, current systems do not provide any guarantees to the data analyst on the quality they care about, namely accuracy of query answers.

We present \system, a novel system that allows data analysts to pose adaptively chosen sequences of queries along with required \textit{accuracy bounds}. By translating queries and accuracy bounds into differentially private algorithms with the least privacy loss, \system returns query answers to the data analyst that meet the accuracy bounds, and proves to the data owner that the entire data exploration process is differentially private. Our comprehensive experimental study on real datasets demonstrates that \system can answer a variety of queries accurately with moderate to small privacy loss, and can support data exploration for entity resolution with high accuracy under reasonable privacy settings.
\end{abstract}


\maketitle

%!TEX root=./main.tex

\section{Introduction}
\label{sec:intro}

Data exploration has been gaining more importance as large and complex
datasets are collected by organizations to enable data scientists to
augment and to enrich their analysis. These datasets are  usually a mix
of public data and private sensitive data, and exploring them  involves operations such
as profiling, summarization, building histograms on various
attributes, and top-k queries to better understand the underlying semantics.
Besides the large body of work on interactive exploration on large data repositories~\cite{DBLP:conf/sigmod/2018hilda},
interacting with large datasets can be an important step towards many mundane jobs.
For example, in entity matching tools (such as Magellan~\cite{DBLP:journals/pvldb/KondaDCDABLPZNP16},
and Data Tamer~\cite{DBLP:conf/cidr/StonebrakerBIBCZPX13}), a sequence of exploration queries are needed to approximate the data distribution of various fields to pick the right
matching model parameters,  features, etc. In data integration systems,
exploring the actual data instance of the sources plays a crucial
role in matching the different schemas.

In our interaction with many large enterprises, exploring private data
sets to enable large-scale data integration and analytics projects is often
a challenging task; while incorporating these data sources in the
analysis carries a high value, data owners often do not trust the analysts and hence, require privacy guarantees. For instance, Facebook recently announced that they would allow academics to analyze their data to ``study the role of social media in elections and democracy", and one of the key concerns outlined was maintaining the privacy of their customers' data \cite{king:blogpost, king:newmodel}. 

Differential privacy~\cite{Dwork06differentialprivacy, Dwork:2014:AFD:2693052.2693053} has emerged as a popular privacy notion since (1) it is a persuasive mathematical guarantee that individual records are hidden even when aggregate statistics are revealed, (2) it ensures privacy even in the presence of side information about the data, and (3) it allows one to bound the information leakage by a total \textit{privacy budget} across multiple data releases.  Differential privacy has seen adoption in a number of real products
at the US Census Bureau~\cite{haney17:census,machanavajjhala08onthemap,Vilhuber17Proceedings},
Google~\cite{Erlingsson14Rappor}, Apple~\cite{Greenberg16Apples} and Uber~\cite{DBLP:journals/corr/JohnsonNS17}. Hence, it is natural for data owners to aspire for systems that permit differentially private data exploration.

%\stitle{The Challenge:}
\subsection{The Challenge and The Problem}
While there exist general purpose differentially private query answering systems, they are not really meant to support interactive querying, and they fall short in two key respects. First, these systems place an inordinate burden  on the data analyst to understand differential privacy and differentially private algorithms. For instance, PINQ \cite{McSherry:2009:PIQ:1559845.1559850} and  wPINQ \cite{Proserpio:2014:CDS:2732296.2732300} allow users to write differentially private programs and ensure that every program expressed satisfies differential privacy. However, to achieve high accuracy, %the analyst typically would need to reason about how the system adds noise to the queries and which queries can be answered accurately based on the privacy literature.
the analyst has to be familiar with the privacy literature to understand how the system adds noise and to identify if the desired accuracy can be achieved in the first place.
$\epsilon$ktelo \cite{ektelo} has high level operators that can be composed to create accurate differentially private programs to answer counting queries. However, the analyst still needs to know how to optimally apportion privacy budgets across different operators. FLEX \cite{DBLP:journals/corr/JohnsonNS17} allows users to answer one SQL query under differential privacy, but has the same issue of apportioning privacy budget across a sequence of queries. Second, and somewhat ironically, these systems do not provide any guarantees to the data analyst on the quality they really care about, namely \textit{accuracy} of query answers. In fact, most of  these systems take a privacy level ($\epsilon$) as input and make sure that differential privacy holds, but leave accuracy unconstrained.

We aim to design a system that allows data analysts to explore a sensitive dataset $D$ held by a data owner by posing a sequence of declaratively specified queries that can capture typical data exploration workflows. The system aims at achieving the following dual goals: (1)~since the data are sensitive, the data owner would like the system to provably bound the information disclosed about any one record in $D$ to the analyst; and (2)~since privacy preserving mechanisms introduce error, the data analyst must be able to specify an accuracy bound on each query.

Hence, our goal is to design a system that can:
\squishlist
\item Support declaratively specified aggregate queries that capture a wide variety of data exploration tasks. 
\item Allow analysts to specify accuracy bounds on queries.
\item Translate an analyst's query into a differentially private mechanism with minimal privacy loss $\epsilon$ such that it can answer the query while meeting the accuracy bound.
\item Prove that for any interactively specified sequence of queries, the analyst's view of the entire data exploration process satisfies $B$-differential privacy, where $B$ is a owner specified privacy budget.
\squishend

Our problem is similar in spirit to Ligett et al.~\cite{accuracyfirst:nips17}, which also considers analysts who specify accuracy constraints.  While their work focus on finding the smallest privacy cost for a given differentially private mechanism and accuracy bound, our focus is on a more general problem: for a given query, find a mechanism and a minimal privacy cost to achieve the given accuracy bound.  We highlight the main technical differences in Appendix~\ref{sec:related}.

\subsection{Contributions and Solution Overview}

We propose \system, an accuracy-aware privacy engine for sensitive data exploration. \system solves the aforementioned challenges as follows: 
\squishlist
\item A data analyst can interact with the private data through \system using declaratively specified aggregate queries. Our query language supports three types of aggregate queries: (1) \textit{workload counting queries} that capture the large class of linear counting queries (e.g., histograms and CDFs) which are a staple of statistical analysis, (2) \textit{iceberg queries}, which capture HAVING queries in SQL and frequent pattern queries, and (3) \textit{top-k queries}. These queries form the building blocks of several data exploration workflows. To demonstrate their applicability in real scenarios, in Section~\ref{sec:casestudy}, we  express two important data cleaning tasks, namely blocking and pair-wise matching, using sequences of queries from our language.

\item In our language, each query is associated with intuitive accuracy bounds that permit \system to use differentially private mechanisms that introduce noise while meeting the accuracy bound.
\item For each query in a sequence, \system employs an \textit{accuracy translator} that finds  a privacy level $\epsilon$ and a differentially private mechanism that answers the query while meeting the specified accuracy bound. For the same privacy level, the mechanism that answers a query with the least error  depends on the query and dataset. Hence, \system implements a suite of differentially private mechanisms for each query type, and given an accuracy bound chooses the mechanism that incurs the least privacy loss based on the input query and dataset. 
\item \system uses a \textit{privacy analyzer} to decide whether or not to answer a query such that the privacy loss to the analyst is always bounded by a budget $B$. The privacy analysis is novel since (1) the privacy loss of each mechanism is chosen based on the query's accuracy requirement, and (2) some mechanisms have a data dependent privacy loss.   
\item In a comprehensive empirical evaluation on real datasets with query and application benchmarks, we demonstrate that (1) \system chooses a differentially private mechanism with the least privacy loss that answers an input query under a specified accuracy bound, and (2) allows data analysts to accurately explore data while ensuring provable guarantee of privacy to data owners. 
\squishend

%\noindent{\bf 
\subsection{Organization}
Section~\ref{sec:preliminaries} introduces preliminaries and background. Section~\ref{sec:queries} describes our query language and accuracy measures. \system architecture is summarized in Section~\ref{sec:system-overview} and details of the accuracy translator and the privacy analysis are described in Sections~\ref{sec:translation} and \ref{sec:analyzer}, respectively.  \system is comprehensively evaluated on real data using exploration queries in Section~\ref{sec:evaluation} and on an entity resolution case study in Section~\ref{sec:casestudy}. We conclude with future work in Section~\ref{sec:discussion}.

%%%%%%%%%% OLD INTRO
\eat{
%%%data exploration is important in discovery based applications
Increasingly large and complex datasets are collected by organizations
and have become an valuable resource in discovery-based applications today,
e.g., for location-based services, healthcare, genomics, and marketing.
Unlike working with traditional databases,
the data analyst do not always know what queries to ask on these data.
This typically requires the data analyst to ask a sequence of exploration queries
on the data to build a mental model or knowledge space
which can be used in further analysis.
This interactive process is known as \emph{data exploration}.
To support this process, a number of useful commercial tools
(e.g. Trifacta~\cite{DBLP:conf/chi/KandelPHH11} and Tableau~\cite{polaris02})
have been developed followed with extensive research work which focus on
(i) improving the effectiveness of data exploration
by identifying interesting data items or relevant queries~\cite{DBLP:conf/cidr/SellamK13,Fan:2011:ISQ:2004686.2005644}
or (ii) optimizing the performance of data exploration in answering queries more efficiently
~\cite{Alagiannis:2014:HHA:2588555.2610502,Halim:2012:SDC:2168651.2168652,
Cormode:2012:SMD:2344400.2344401,Agarwal:2014:KYW:2588555.2593667,
Agarwal:2013:BQB:2465351.2465355,aqpplus16,revisitaqp17, Tauheed:2012:SPL:2350229.2350267, Kalinin:2014:IDE:2588555.2593666}.
%Among the optimization techniques, approximate query processing~\cite{Cormode:2012:SMD:2344400.2344401,Agarwal:2014:KYW:2588555.2593667,
%Agarwal:2013:BQB:2465351.2465355,aqpplus16,revisitaqp17} trade-off performance with query accuracy, where the analyst is willing to explore the data with tolerable noisy answers.

However, these techniques are not designed for exploration on \emph{sensitive} data
(e.g. intelligence data collected by agencies or health records in a hospital).
Consider a setting where a medical center hires an external cleaning engineer
to clean a collection of medical records with missing values and duplicates.
On one hand, the cleaner needs to first explore the dataset to decide the right configuration for the cleaning tools.
On the other hand, the medical center needs to ensure the sensitive information of the patients
is protected from the cleaner during the data exploration.
This constraint either forces the analyst to give up the opportunity of exploring the data
by treating the dataset as a black box,
or forces the data owner to sacrifice the privacy of individuals and hence their trust
by leaking true answers to the analyst even if the computation is done on encrypted records.
Inspired by approximate query processing technique
~\cite{Cormode:2012:SMD:2344400.2344401,Agarwal:2014:KYW:2588555.2593667, Agarwal:2013:BQB:2465351.2465355,aqpplus16,revisitaqp17}
where noisy answers with reasonably error bounds have been shown effective for data exploration 
a model that allows exploration on sensitive data with provably bounded accuracy and privacy guarantee
is desirable to solve this problem.

Differential privacy~\cite{DBLP:conf/tcc/DworkMNS06} as a state-of-the-art privacy guarantee
is an appealing solution for the collaboration between external data analysts and data owners on sensitive data exploration,
for this guarantee provably bounds the information leakage about any single record in the database
and permits noisy answers to aggregate queries.
Unfortunately, there are very few and limited applications.
The challenges in using differential privacy for data exploration are two-fold.
First, existing approaches using differential privacy for data exploration
focus on one specific domain~\cite{Johnson:2013:PDE:2487575.2487687, visdpt16}
and thus cannot be generalized to other domains.
Second, though there are general differential privacy tools for SQL databases
\cite{McSherry:2009:PIQ:1559845.1559850,Proserpio:2014:CDS:2732296.2732300,DBLP:journals/corr/JohnsonNS17},
neither the data owners nor the data analysts are privacy experts
who know how to add noise to the queries and how to correlate query accuracy with data privacy
using these tools.
%%%how to solve this problem?

\begin{figure}[t]
	\centering
        \includegraphics[width=0.5\textwidth]{./figure/problemsetting}
	\caption{\system overview}\label{fig:problemsetting}
\end{figure}

Therefore, to solve above challenges, we propose \system, an accuracy-aware privacy engine for data exploration,
which allows external analysts to explore sensitive datasets while ensuring differential privacy.
As shown in Figure~\ref{fig:problemsetting},
the data analyst explores the dataset
by adaptively posing exploratory queries to \system.
Rather than thinking in terms of privacy bounds,
the data analysts are only required to specify an \emph{accuracy} requirement on each query.
As the name suggests, \system \emph{translates} these queries with accuracy requirements to
mechanisms that satisfy differential privacy.
\system allows the data owner to specify a bound
on the  privacy leakage about any single record to the data analysts,
and ensures that every sequence of queries it answers has a total privacy leakage below this bound.
In this way, the data owner can easily collaborate with external data analysts
to achieve guarantees on both the exploration quality and the privacy leakage.

The key technical contributions of this paper are :
\squishlist
\item[1)] %{\bf Accuracy-Aware Privacy Engine:}
We design and build the first accuracy-aware privacy engine
that allows data analyst to explore private data in a natural way
by specifying only accuracy requirements.
This approach bridges the gap between privacy concerns of data owner
and accuracy expectation of the data analyst,
and opens new and interesting research problems in private data exploration.

\item[2)] %{\bf Exploration Queries:}
We propose a class of \emph{exploration queries}
that a data analyst can use to privately explore the sensitive dataset.
This class of queries are flexible to support a number of exploratory tasks.
We formulate the semantics of accuracy requirement for each of these query types.
and show that they can match the accuracy expectation
and meet the privacy constraint in the evaluation.

%with respective tolerance formulation on their answers
%to query the aggregate statistics about underlying data
%and show that these primitives can enable an iterative data cleaning process
%by providing the cleaning engineer with the necessary information to tune configuration parameters
%while quantifying the amount of privacy leakage during this interaction.

\item[3)] %{\bf Translating Accuracy to Privacy:}
We design \emph{translating mechanisms} for the exploration queries,
where their errors vary among the queries and the data.
\system helps data analyst choose an optimal mechanism
that can achieve the accuracy requirement
with the least possible privacy cost.
This is in contrast to almost all prior work in differentially private algorithm design
that assumes a constraint on privacy and minimizes error.

\item[4)] %\emph{Privacy Analyzer:}
We prove that any sequence of interactions with the data analyst
satisfies differential privacy with a bounded privacy budget,
even if the accuracy requirement (and hence the privacy budget) per query
is adaptively chosen.

\item[5)] %\emph{Evaluation and Case Study:}
We conduct extensive experiments on two sets of query benchmarks
and an application benchmark (entity resolution) to illustrate
the trade-offs between privacy and accuracy and
the quality of exploration task under a given privacy constraint
using \system.
\squishend
}

\eat{
However, many datasets for data exploration are highly sensitive.
Existing approaches for private datasets
either give up the opportunity of exploring the data by treating the dataset as a black box~\cite{},
or leak true answers of given queries such as computation over encrypted dataset~\cite{}
or secure computations over distributed datasets~\cite{},
which provides no privacy guarantees for individual records in the dataset.
Moreover, either the data owners or the data analysts are privacy experts who can easily quantify
the privacy loss over the data exploration process.

On the other hand, many useful datasets contain sensitive information about individuals, such as electronic medical records.
In this problem, we consider two human roles:
the {\it data owner} who is responsible for ensuring the privacy of individuals who participate into the dataset;
and the {\it data analyst} who is familiar with the data analytics tools and understands the implications of the parameters.
Developing systems that can support private data exploration are important
so that data analysts can verify their hypotheses on the private data while
quantifying the privacy loss about individuals throughout data exploratory process.
This remains an open problem as prevalent approaches for private datasets
either giving up the opportunity of exploring the data by treating the dataset as a black box,
or leaking true answers of given queries
such as computation over encrypted dataset or secure computations over distributed datasets,
which provides no privacy guarantees for individual records in the dataset.
}

%!TEX root=./main.tex
\section{Preliminaries and Background}\label{sec:preliminaries}
We consider the sensitive dataset in the form of a single-table relational schema $R(\attr_1,\attr_2,\ldots,\attr_d)$,
where $attr(R)$ denotes the set of attributes of $R$. Each attribute $\attr_i$ has a domain $dom(A_i)$.
The full domain of $R$ is $\rowdomain=dom(A_1)\times\cdots \times dom(A_d)$, containing all possible tuples conforming to $R$.
An instance $D$ of relation $R$ is  a multiset whose elements are tuples in $\rowdomain$.
We let the domain of the instances be $\tabledomain$.
Extending our algorithms to schemas with multiple tables is an interesting avenue for future work.
%$D\in \mathcal{D}$, with schema $\attrlist = \{\attr_1, \attr_2, \ldots, \attr_d\}$.
%Each row $r \in D$ is drawn from a domain consisting of $d$ attributes,
%$\rowdomain = \attr_1\times \cdots \attr_d$.

%\subsection{Differential Privacy}\label{sec:dp_def}
We use differential privacy as our measure of privacy. An algorithm that takes as input a table $D$ satisfies differential privacy~\cite{Dwork:2014:AFD:2693052.2693053, DBLP:conf/tcc/DworkMNS06}
if its output does not significantly change
by adding or removing a single tuple in its input.
\begin{definition}[$\epsilon$-Differential Privacy \cite{Dwork:2014:AFD:2693052.2693053}]\label{def:dp}
A randomized mechanism $M: \mathcal{D}\rightarrow \mathcal{O}$ satisfies $\epsilon$-differential privacy if
\begin{equation}
Pr[M(D) \in O] \leq e^{\epsilon} Pr[M(D')\in O]
\end{equation}
for any set of outputs $O\subseteq \mathcal{O}$,
and any pair of \emph{neighboring} databases $D,D'$
such that %$D$ and $D'$ differ by adding or removing a record; i.e., 
$|D\backslash D'\cup D'\backslash D| = 1$.
\end{definition}
Smaller values of $\epsilon$ result in stronger privacy guarantees as $D$ and $D'$ are harder to distinguish using the output. Composition and postprocessing theorems of differential privacy are described in Appendix~\ref{sec:seq} and will be used to bound privacy across multiple data releases in Section~\ref{sec:analyzer}.

%We did not choose semantically secure or property preserving encryption, as the former does not allow the data analyst to learn anything about the data, and the latter is susceptible to attacks using side information \cite{naveed2015:inference}.

%No prior work have shown how to answer all these queries with bounded errors in a unified framework for data exploration.

%!TEX root=./main.tex
\section{Queries and Accuracy}\label{sec:queries}
In this section, we describe our query language for expressing aggregate queries and the associated accuracy measures. Throughout the paper, we assume that the schema and the full domain of attributes are public.

\subsection{Exploration Queries}\label{sec:explorequeries}
\system supports a rich class of aggregate queries
that can be expressed in a SQL-like declarative format.
%{\small
\begin{lstlisting}[mathescape=true]
BIN $D$ ON $f(\cdot)$ WHERE $W=\{\phi_1,\ldots,\phi_L\}$
[HAVING $f(\cdot)>c$]
[ORDER BY $f(\cdot)$ LIMIT $k$];
\end{lstlisting}
%}

%This class of exploration queries has
%a massive data exploration use cases~\cite{polaris02,aqpplus16,revisitaqp17,
%	Agarwal:2013:BQB:2465351.2465355}.
%We also show in our evaluation
%that the exploration on real world datasets (e.g. census and location data)
%and even more complex exploration tasks for entity resolution
%can be supported using this class of queries.
%The format of these queries also allow the data analyst
%to explore data in a SQL-style declarative format
%and is extendable to more complex queries in the future.

Each query in our language is associated with a \textit{workload} of predicates $W=\{\phi_1,\ldots,\phi_L\}$. Based on $W$, the tuples in a table $D$ are divided into bins. Each bin $b_i$ contains all the tuples in $D$ that satisfy the corresponding predicate $\phi_i:\rowdomain \rightarrow \{0,1\}$, i.e., $b_i= \{r\in D| \phi(r) = 1\}$. As we will see later, bins need not be disjoint. Moreover, a query has an aggregation function $f:\rowdomain^* \rightarrow \mathbb{R}$, which returns a numeric answer $f(b_i)$ for each bin $b_i$. The output of this query without the optional clauses (in square brackets) is a list of counts $f(b_i)$ for bin $b_i$.

Each query can be specialized using one of two optional clauses: the HAVING clause returns a list of bin identifiers $b_i$ for which $f(b_i) > c$; and the ORDER BY ... LIMIT clause returns the $k$ bins that have the largest values for $f(b_i)$.
%\change{
%\sout{Throughout this paper, we assume $f(\cdot)$ is the COUNT function and omit extensions to other aggregates like AVG, SUM, QUANTILE due to space constraints. }
%}

Throughout this paper, we focus on COUNT as the aggregate function and discuss other aggregates like AVG, SUM, QUANTILE in Appendix~\ref{app:otherqueries}.

\stitle{Workload Counting Query (\wcq).}
\begin{lstlisting}[mathescape=true]
BIN $D$ ON $f(\cdot)$ WHERE $W=\{\phi_1,\ldots,\phi_L\};$
\end{lstlisting}
Workload counting queries capture the large class of \textit{linear counting queries}, which are the bread and butter of statistical analysis and have been the focus of majority of the work in differential privacy \cite{Li:2015:MMO:2846574.2846647, Hay:2016:PED:2882903.2882931}. Standard SELECT...GROUP BY queries in SQL are expressible using \wcq.
For instance, consider a table $D$ having
an attribute $State$ with domain $\{$AL, AK, \ldots, WI, WY$\}$
and an attribute $Age$ with domain $[0,\infty)$.
Then, a query that returns the number of people with age above 50 for each state 
%\begin{lstlisting}[mathescape=true]
%SELECT $State$, COUNT(*) FROM $D$
%    WHERE $Age>50$ GROUP BY $State$;
%\end{lstlisting}
can be expressed using \wcq as:
\begin{lstlisting}[mathescape=true]
BIN $D$ ON COUNT(*) 
WHERE $W=\{{\scriptstyle Age>50\wedge State=\text{AL},\ldots,Age>50\wedge State=\text{WY}}\}$;
\end{lstlisting}

%Note that this format also allows GROUP BY an attribute with continuous domain.
%For example, we can specify $W=\{0<Age\leq50,  50<Age\}$ to query
%for the population with age below 50 and above 50 respectively.
%We can also easily specify a different discretization of the continuous domain,
%e.g.,  $W=\{0 < Age \leq 21,  21<Age\}$.
%This expression allows more flexibility and
%also helps control privacy loss which we will explain later.

Other common queries captured by \wcq include:
(1) {\em histogram queries:} the workload $W$ partitions $D$ into $|W_h|$ disjoint bins,
e.g. $W_h=\{0 < Age \leq 10,  10<Age\leq 20,\ldots, 90<Age\}$ the resulting \wcq returns counts for each bin;
(2) {\em cumulative histograms:}  we can define a workload $W_p$ that places tuples in $D$ into a set of inclusive bins $b_1\subseteq b_2 \cdots \subseteq b_L$,
e.g. $W_p=\{Age \leq 10,  Age\leq 20,\ldots, Age\leq 90\}$. The resulting query outputs a set of cumulative counts. We call such a $W_p$ a \textit{prefix} workload.

\stitle{Iceberg Counting Query (\icq).}
%The corresponding expression for iceberg counting queries is the base query template with the HAVING option turned on:

\begin{lstlisting}[mathescape=true]
BIN $D$ ON COUNT(*) WHERE $W=\{\phi_1,\ldots,\phi_L\}$
HAVING COUNT(*)$>c$;
\end{lstlisting}

An iceberg query returns bin identifiers if the aggregate value for that bin is greater than a given threshold $c$.
For instance, a query which returns the states in the US with a population of at least 5 million
%\begin{lstlisting}[mathescape=true]
%SELECT $State$ FROM $D$ GROUP BY $State$ 
%	HAVING COUNT(*) > 5 million;
%\end{lstlisting}
can be expressed as: 
\begin{lstlisting}[mathescape=true]
BIN $D$ ON COUNT(*) WHERE $W=\{\mbox{\small State=AL,...,State=WY}\}$
HAVING COUNT(*)>$\mbox{\small 5 million}$;
\end{lstlisting}

Note that since the answer to the query is a subset of the predicates in $W$ (i.e. a subset of bin identifiers  but not the aggregate values for these bins, an \icq is not a linear counting query. 

\eat{
	For example, which states in U.S. have a population more than 5 million?
	Without considering privacy, iceberg queries can be computed directly
	from the answers to workload queries.
	However, hiding the aggregates for the outputted bins
	make a difference in the privacy loss.
	Hence, we consider a special class of iceberg queries in this work,
	which return only the predicates of the bins with aggregate values passing the given threshold
	but not the aggregate values over these bins.
	Hence, the output of this class of iceberg queries
	is a subset of the bin predicates $W'\subseteq W$.
	If a bin of tuples in $D$ satisfy $\phi_i$ has an aggregate value greater than $c$,
	then $\phi_i$ will be in the query output.}

\stitle{Top-$k$ Counting Query (\tcq).}
%The corresponding expression for top-$k$ queries is the base query template with the ORDER BY $f(\cdot)$ LIMIT $k$ option turned on:

\begin{lstlisting}[mathescape=true]
BIN $D$ ON COUNT(*) WHERE $W=\{\phi_1,\ldots,\phi_L\}$
ORDER BY COUNT(*) LIMIT $k$;
\end{lstlisting}

\tcq first sorts the bins based on their aggregate values (in descending order) and returns the top $k$ bins identifiers (and not the aggregate values).
For example, a query which returns the three US states with highest population can be expressed as: 
\begin{lstlisting}[mathescape=true]
BIN $D$ ON COUNT(*) WHERE $W=\{\mbox{\small State=AL,...,State=WY}\}$
ORDER BY COUNT(*) LIMIT $k$;
\end{lstlisting}

%!TEX root=./main.tex
\subsection{Accuracy Measure}\label{sec:accuracy}
To ensure differential privacy, the answers to the exploration queries are typically noisy. To allow the data analyst to explore data with bounded error, we extend our queries to incorporate an accuracy requirement. The syntax for accuracy is inspired by that in BlinkDB~\cite{Agarwal:2013:BQB:2465351.2465355}:

\begin{lstlisting}[mathescape=true]
BIN $D$ ON $f(\cdot)$ WHERE $W=\{\phi_1,\ldots,\phi_L\}$
[HAVING $f(\cdot)>c$]
[ORDER BY $f(\cdot)$ LIMIT $k$]
ERROR $\alpha$ CONFIDENCE $1-\beta$;
\end{lstlisting}

%The accuracy requirement for the queries are defined as follows.
%\begin{definition}[$(\alpha,\beta)$-$q$.type accuracy]\label{def:accuracy}
%Given a query, $q: \tabledomain \rightarrow O$,
%and a distance function over the query output,
%$d: O\times O\rightarrow \mathbb{R}^{+}$.
%We say a mechanism $M:\tabledomain \rightarrow O$
%satisfies $(\alpha,\beta)$-q.type accuracy for
%$\alpha\in \mathbb{R}^+$ and $\beta\in [0,1)$
%if for $D\in \tabledomain$
%\begin{equation}
%\Pr[dist_{q.type}(M(D),q(D)) \geq \alpha] \leq \beta.
%\end{equation}
%\end{definition}
%
%We fix $\beta$ to be a very small value ($0.0005$ in our experiments) throughout the paper.

We next define the semantics of the accuracy requirement for each of our query types.  
The accuracy requirement for a \wcq $q_W$ is defined as a bound on the maximum error across queries in the workload $W$.

\begin{definition}[$(\alpha,\beta)$-\wcq accuracy]
	Given a workload counting query $q_W:\tabledomain \rightarrow \mathbb{R}^L$,
	where $W=\{\phi_1,\ldots,\phi_L\}$.
	Let $M:\tabledomain \rightarrow \mathbb{R}^L$ be a mechanism
	that outputs a vector of answers $y$ on $D$.
	Then, $M$ satisfies $(\alpha,\beta)$-$W$ accuracy,
	if  $\forall D\in \tabledomain$, %for all $\phi_i\in W$
	\begin{eqnarray}
	\Pr[\|y - q_W(D)\|_{\infty} \geq \alpha] \leq \beta,
	\end{eqnarray}
	where $\|y - q_W(D)\|_{\infty} = \max_j |y[i]-c_{\phi_i}(D)|$.
\end{definition}

\begin{figure}[t]
	\centering
	\includegraphics[width=\columnwidth]{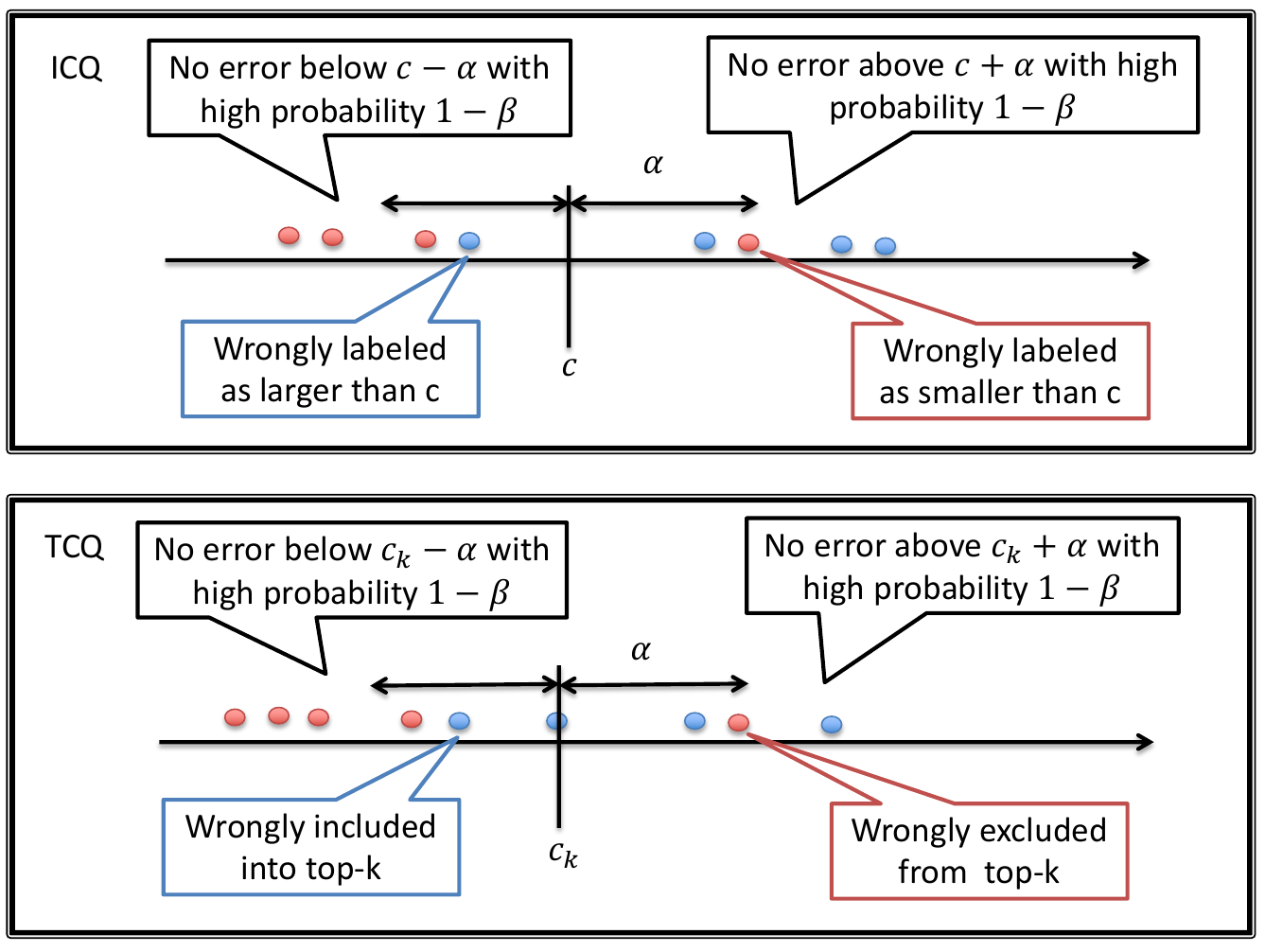}
	\caption{Accuracy Requirement for \icq and \tcq}\label{fig:icq_tcq_error}
\end{figure}

The output of iceberg counting queries \icq and top-$k$ counting queries \tcq
are not numeric, but a subset of the given workload predicates.
Their accuracy measures are different from \wcq,
and depend on their corresponding workload counting query $q_{W}$.

\begin{definition}[$(\alpha,\beta)$-\icq accuracy]
	Given an iceberg counting query $q_{W,>c}: \tabledomain \rightarrow O$,
	where $W=\{\phi_1,\ldots,\phi_L\}$, and $O$ is a power set of $W$.
	Let $M:\tabledomain \rightarrow O$ be a mechanism that outputs a subset of $W$.
	Then, $M$ satisfies $(\alpha,\beta)$-\icq accuracy for $q_{W,>c}$
	if for $D$,
	%for all $\phi\in W$,
	\begin{eqnarray}
	%\Pr[\phi \in M(D) ~|~ c_{\phi}(D) < c-\alpha] \leq \beta \\
	%\Pr[\phi \notin M(D) ~|~ c_{\phi}(D) > c+\alpha] \leq \beta
	\Pr[|\{\phi \in M(D) ~|~ c_{\phi}(D) < c-\alpha\}|>0] \leq \beta \\
	\Pr[|\{\phi \in (W-M(D)) ~|~ c_{\phi}(D) > c+\alpha\}|>0] \leq \beta
	\end{eqnarray}
\end{definition}

A mechanism for \icq can make two kinds of errors:
label predicates with true counts greater than $c$ as $<c$ (red dots in Figure~\ref{fig:icq_tcq_error})
and label predicates with true counts less than $c$ as $>c$ (blue dots in Figure~\ref{fig:icq_tcq_error}).
We say a mechanism satisfies $(\alpha,\beta)$-\icq accuracy if
with high probability, all the predicates with true counts greater than $c+\alpha$
are correctly labeled as $>c$
and all the predicates with true counts less than $c-\alpha$
are correctly labeled as $<c$.
The mechanism may make arbitrary mistakes within the range $[c-\alpha,c+\alpha]$.

\begin{definition}[$(\alpha,\beta)$-\tcq accuracy]\label{def:tcqaccuracy}
	Given a top-$k$ counting query $q_{W,k}: \tabledomain \rightarrow O$,
	where $W=\{\phi_1,\ldots,\phi_L\}$, and $O$ is a power set of $W$.
	Let  $M:\tabledomain \rightarrow O$ be a mechanism that outputs a subset of $W$.
	Then, $M$ satisfies $(\alpha,\beta)$-\tcq accuracy
	if for $D\in \tabledomain$,
	\begin{eqnarray}
	\Pr[ |\{\phi \in M(D) ~|~ c_{\phi}(D) < c_k-\alpha \}| > 0] \leq \beta \\
	\Pr[ |\{\phi \in ( \Phi -M(D)) ~|~ c_{\phi}(D) > c_k + \alpha \}| > 0 ] \leq \beta
	\end{eqnarray}
	where $c_k$ is the $k^{th}$ largest counting value among all the bins,
	and $\Phi$ is the true top-$k$ bins.
\end{definition}
The intuition behind Definition~\ref{def:tcqaccuracy} is similar to that of \icq and is explained in Figure~\ref{fig:icq_tcq_error}: predicates with count greater than $c_k + \alpha$ are included and predicates with count less than $c_k -\alpha$ do not enter the top-$k$ with high probability.  
%Similar to \icq, the mechanism for \tcq can also makes two types of mistakes:
%exclude top-$k$ predicates from the output (red in Figure~\ref{fig:icq_tcq_error}
%and include a non top-$k$ predicates in the output (blue in Figure~\ref{fig:icq_tcq_error}).
%We say a mechanism satisfies $(\alpha,\beta)$-\tcq accuracy
%if with high probability all the predicates with true counts greater than $c_k+\alpha$
%are correctly included in the top-$k$ output
%and all the predicates with true counts less than $c_k-\alpha$
%are correctly excluded from the output.
%The mechanism may make arbitrary mistakes within the range $[c_k-\alpha,c_k+\alpha]$.

In the rest of the paper, we will describe how \system designs differentially private mechanisms to answer queries with the above defined bounds on accuracy. The advantages of our accuracy definitions are that
they are intuitive (when $\alpha$ increases, noisier answers are expected)
and we can design privacy-preserving mechanisms that introduce noise while satisfying these accuracy guarantees. 
On the other hand, this measure is not equivalent to other bounds on accuracy like relative error and precision/recall which can be very sensitive to small amounts of noise (when the counts are small, or when lie within a small range).
For instance, if the counts of all the predicates in \icq lie outside $[c-\alpha,c+\alpha]$,
a mechanism $M$ that perturbs counts within $\pm \alpha$ and then answers an \icq will have precision and recall of 1.0 with high probability as it makes no mistakes.
However, if all the query answers lie within $[c-\alpha,c+\alpha]$,
then the precision and recall of the output of $M$ could be 0.
%Similarly, if the counts of all the predicates in \tcq (except $c_k$) lie outside $[c_k-\alpha,c_k+\alpha]$,
%the mechanism will have precision and recall of 1.0 with high probability as it makes no mistakes.
%However, if all the query answers lie within $[c_k-\alpha,c_k+\alpha]$,
%then there is no bound on the precision and recall (which could be 0).
%\todo{We show this does not happen in our experiments for both \icq and \tcq.}
Incorporating other error measures like precision/recall and relative error into \system is an interesting avenue for future work.

%!TEX root=./main.tex
\section{\system Overview}\label{sec:system-overview}

\begin{algorithm}[t]
	\caption{\system Overview}\label{algo:pe}
	\label{algo:lcm}
	{\small \begin{algorithmic}[1]
			\Require  Dataset $D$, privacy budget $B$
			\State Initialize privacy loss $B_0 \gets 0$, index $i\gets 1$
			\Repeat
			\State Receive $(q_i,\alpha_i, \beta_i)$ from analyst
			\State \label{applicable} $\mathcal{M} \gets $ mechanisms applicable to $q_i$'s type
			\State \label{estimateloss} $\mathcal{M}^* \gets \{ M\in \mathcal{M} \mid M.\transFunc(q_i,\alpha_i,\beta_i).\upperbound\leq B-B_{i-1} \}$
			\If{$\mathcal{M}^* \neq \emptyset$} 
			\State \textit{// Pessimistic Mode}
			\State \label{pessimistic}\ \ \label{translate} $M_i \gets {\tt argmin}_{M\in \mathcal{M}^*} ~~M.\transFunc(q_i,\alpha_i,\beta_i).\upperbound$
			\State \textit{// Optimistic Mode}
			\State \label{optimistic}\ \ \label{translate} $M_i \gets {\tt argmin}_{M\in \mathcal{M}^*} ~~M.\transFunc(q_i,\alpha_i,\beta_i).\lowerbound$
			\State $(\omega_i,\epsilon_i) \gets M_i.\runFunc(q_i,\alpha_i,\beta_i,D)$
			\State\label{analyzeloss} $B_i \gets B_{i-1} + \epsilon_i$, $i$++
			\State \Return $\omega_i$ 
			\Else
			\State $B_i = B_{i-1}$, $i$++
			\State\label{denied} \Return `Query Denied'  
			\EndIf\label{endcomposition}
			\Until{No more queries sent by local exploration}
	\end{algorithmic}}
\end{algorithm}

%Algorithm~\ref{algo:pe} describes the end-to-end \system. \am{need something here?}
This section outlines how \system translates analyst queries with accuracy bounds into differentially private mechanisms, and how it ensures the privacy budget $B$ specified by the data owner is not violated.

\stitle{Accuracy Translator.}
Given an analyst's query $(q, \alpha, \beta)$, \system first uses the \textit{accuracy translator} to choose a mechanism $M$ that can (1) answer $q$ under the specified accuracy bounds, with (2) minimal privacy loss. To achieve these, \system supports a  set of differentially private mechanisms that can be used to answer each query type (\wcq, \icq, \tcq). Multiple mechanisms are supported for each query type as different mechanisms result in the least privacy loss depending on the query and the dataset (as shown theoretically and empirically in Sections~\ref{sec:translation} and \ref{sec:evaluation}, respectively).

Each mechanism $M$ exposes two functions: $M.\transFunc$, which \textit{translates} a query and accuracy requirement into a lower and upper bound $(\lowerbound, \upperbound)$ on the privacy loss if $M$ is executed, and $M.\runFunc$ that runs the differentially private algorithm and returns an approximate answer $\omega$ for the query. The answer $\omega$ is guaranteed to satisfy the specified accuracy requirement. Moreover, $M$ satisfies $\upperbound$ differential privacy. The mechanisms supported by \system and the corresponding $\transFunc$ functions are described in Section~\ref{sec:translation}. In some cases (e.g. Algorithm~\ref{algo:mpm} in Section~\ref{sec:icq_mpm}), the privacy loss incurred by $M$ may be $\epsilon \in (\lowerbound, \upperbound)$ that is smaller than the worst case, depending on the characteristics of the dataset. %We show an example of one such mechanism in Section~\ref{sec:icq_mpm}.

As described in Algorithm~\ref{algo:pe}, \system first identifies the mechanisms $\mathcal{M}$ that are applicable for the type of the query $q_i$ (Line~\ref{applicable}). Next, it runs $M.\transFunc$ to get conservative estimates on privacy loss $\upperbound$ for all these mechanisms (Line~\ref{estimateloss}). \system picks one of the mechanisms $M$ from those that can be safely run using the remaining privacy budget, executes $M.\runFunc$, and returns the output to the analyst. As we will see there exist mechanisms where the privacy loss can vary based on the data in a range between $[\lowerbound, \upperbound]$, and the actual privacy loss is unknown before running the mechanism.  In such cases, \system can choose to be \textit{pessimistic} and pick the mechanism with the least $\upperbound$ (Line~\ref{pessimistic}), or choose to be \textit{optimistic} and pick the mechanism with the least $\lowerbound$ (Line~\ref{optimistic}).   

\stitle{Privacy Analyzer.}
%The second functionality of \system named as {\em privacy analyzer} analyzes the accumulated privacy cost.
Given a sequence of queries $(M_1,\ldots,M_{i})$ already executed
by the privacy engine that satisfy an overall $B_{i-1}$-differential privacy and a new query $(q_i,\alpha_i,\beta_i)$,
\system identifies a set of mechanisms $\mathcal{M}^*$ that all will have a worst case privacy loss smaller than  $B-B_{i-1}$ (Line~\ref{estimateloss}). That is, running any mechanism in $\mathcal{M}^*$ will not result in exceeding the privacy budget in the worst case. If $\mathcal{M}^*=\emptyset$, then \system returns `Query Denied' to the analyst (Line~\ref{denied}). Otherwise, \system runs one of the mechanisms $M_i$ from $\mathcal{M}^*$ by executing $M_i.\runFunc()$
and the output $\omega_i$ will be returned to the analyst.
\system then increments $B_{i-1}$ by the actual privacy loss $\epsilon_i$ rather than the upperbound $\upperbound$ (Line~\ref{analyzeloss}). As explained above, in some cases $\epsilon_i < \upperbound$ as different execution paths in the mechanism can have different privacy loss. Nevertheless, 
the privacy analyzer guarantees that the execution
of any sequence of mechanisms $(M_1,M_2,\ldots,M_i)$ before it halts
is $B$-differentially private (see Section~\ref{sec:analyzer}, Theorem~\ref{thm:privacy}).

%!TEX root=./main.tex
\section{Accuracy Translator} \label{sec:translation}

In this section, we present the accuracy-to-privacy translation mechanisms
supported by \system and the corresponding $\runFunc$ and $\transFunc$ functions.
We first present a general baseline mechanism for all three types of exploration queries
including workload counting query (\wcq), iceberg counting query (\icq), and
top-$k$ counting query (\tcq).
Then we show specialized mechanisms for each type of exploration queries which
consumes smaller privacy cost under different scenarios than the baseline.

We represent the workload in these three queries in a matrix form,
like the prior work for \wcq~\cite{Li:2010:OLC:1807085.1807104, Li:2015:MMO:2846574.2846647, hdmm18}.
There are many possible ways to transform a workload into a matrix.
\change{
Given a query with $L$ predicates, the number of domain partitions can be as large as $2^L$.
In this work, we consider the following transformation to reduce complexity.
}
Given a workload counting query $q_W$ with the set of predicates $W=\{\phi_1,\ldots,\phi_L\}$,
the full domain of the relation $dom(R)$ is partitioned based on $W$ to form the new discretized domain $dom_W(R)$
such that any predicate $\phi_i \in W$ can be expressed as a union of partitions in the new domain $dom_W(R)$
and the number of partitions is minimized.
For example, given $W=\{Age>50\wedge State=\text{AL}, \ldots,Age>50\wedge State=\text{WY}\}$,
one possible partition is $dom_W(R) = \{ Age>50\wedge State=\text{AL}, \ldots ,Age>50\wedge State=\text{WY}, Age\leq 50\}$.

Let $\data$ represent the histogram of the table $D$ over $dom_W(R)$.
The set of corresponding counting queries $\{c_{\phi_1},\ldots, c_{\phi_L}\}$ for $q_{W}$ can be
represented by a matrix $\workload =[\wq_1,\ldots,\wq_L]^{T}$ of size $L \times |dom_W(R)|$.
Hence, the answer to each counting query is $c_{\phi_i}(D) = \wq_i \cdot \data$ and
the answer to the workload counting query is simply $\workload\data$.
We denote and use this transformation by $\workload\gets\transform(W), \data\gets \transform_W(D)$
throughout this paper.
Unlike prior work~\cite{Li:2010:OLC:1807085.1807104, Li:2015:MMO:2846574.2846647, hdmm18}
 which aims to bound the expected total error for one query,
\system aims to bound the maximum error per query with high probability
which is more intuitive in the process of data exploration.

\eat{
Each mechanism $M$ is supported by two functions $M.\transFunc$ and $M.\runFunc$.
The main execution of mechanism $M$ is specified by $M.\runFunc()$,
which takes in the query $q$, accuracy requirement $(\alpha,\beta)$-$q$.type, and data $D$
and outputs a noisy answer $\omega$ and the actual privacy cost by running this mechanism $\epsilon$.
The design of this function is that
the same mechanism $M.\runFunc()$ with different noise parameters $b'$
cannot both simultaneously achieve $(\alpha,\beta)$-$q_i$.type accuracy
and $\epsilon'$-differential privacy for $\epsilon'<\epsilon$.
The second function $M.\transFunc()$ returns two privacy values $(\upperbound,\lowerbound)$
which correspond to the largest possible privacy value returned by $M.\transFunc()$,
and the smallest possible privacy value returned by $M.\transFunc()$ respectively.
In many of the mechanisms shown in this section, these two privacy values are the same,
and these mechanisms are known as data independent translations.
Mechanisms with $\upperbound >\lowerbound$ (Section~\ref{sec:icq_mpm})
are known as data dependent translations.
}

%Note that this section considers mechanisms for one query at a time. The overall privacy of a sequence of queries will be presented in Section~\ref{sec:analyzer}.
%All the translations are based on the theoretical properties of Laplace noise, and hence are different from the prior work \cite{DBLP:journals/corr/LigettNRWW17} that computes true error for a complex analysis.

\subsection{Baseline Translation}\label{sec:baseline}
The baseline translation for all three query types is based on
the Laplace mechanism~\cite{Dwork:2014:AFD:2693052.2693053, DBLP:conf/tcc/DworkMNS06},
a classic and widely used differentially private algorithm,
which can be used for \wcq, \icq, and \tcq. Formally,
\begin{definition}[Laplace Mechanism (Vector Form)\cite{Dwork:2014:AFD:2693052.2693053, DBLP:conf/tcc/DworkMNS06}]
Given an $L\times |dom_W(R)|$ query matrix $\workload$, the randomized algorithm LM that
outputs the following vector is $\epsilon$-differentially private:
$LM(\workload,\data) = \workload\data + Lap(b_{\workload})^L$ where $b_{\workload} = \frac{\|\workload\|_1}{\epsilon}$,
and $Lap(b)^L$ denote a vector of $L$ independent samples $\eta_i$ from a Laplace distribution
with mean $0$ and variance $2b^2$, i.e., $\Pr[\eta_i = z] \propto e^{-z/b}$ for $i=1,\ldots,L$.
\end{definition}
The constant $\|\workload\|_1$ is equal to the sensitivity of queries set defined by the workload
$\workload$ \cite{Li:2010:OLC:1807085.1807104, Li:2015:MMO:2846574.2846647}.
It measures the maximum difference in the answers to the queries in $\workload$ on any two databases
that differ only a single record. \change{Mathematically, it is the maximum of $L1$ norm of a column of $\workload$.} 

\begin{algorithm}[t]
\caption{Laplace Mechanism (LM) ($q,\alpha,\beta,D$)}\label{algo:lm}
{\small \begin{algorithmic}[1]
\State Initialize $\workload \gets \transform(W=\{\phi_1,\ldots,\phi_L\}), \data \gets \transform_W(D), \alpha,\beta$
\Function{\runFunc}{$q,\alpha,\beta, D$}
    \State $\epsilon \gets \transFunc(q_W,\alpha,\beta).\upperbound$
    \State $[\tilde{x}_1,\ldots,\tilde{x}_L] \gets \workload\data + Lap(b)^L$, where $b=\|\workload\|_1/\epsilon$
    \If{$q$.type==\wcq (i.e., $q_W$)}
        \State {\bf return} $([\tilde{x}_1,\ldots,\tilde{x}_L], \epsilon)$
    \ElsIf{$q$.type==\icq (i.e., $q_{W,>c}$)}
        \State {\bf return} $(\{ \phi_i \in W~~|~~\tilde{x}_i>c\}, \epsilon)$\label{alg:lm:icq}
    \ElsIf{$q$.type==\tcq (i.e., $q_{W,k}$)}
        \State \textbf{return} $({\tt argmax}^k_{\phi_1,\ldots,\phi_L} \tilde{x}_i ,\epsilon)$\label{alg:lm:tcq}
    \EndIf
\EndFunction
\Function{\transFunc}{$q,\alpha,\beta$}
    \If{$q$.type==\wcq (i.e., $q_W$)}
        \State {\bf return} $(\upperbound = \frac{\|\workload\|_1\ln(1/(1-(1-\beta)^{1/L}))}{\alpha}, \lowerbound = \upperbound)$
    \ElsIf{$q$.type==\icq (i.e., $q_{W,>c}$)}
        \State {\bf return} $(\upperbound =\frac{\|\workload\|_1 (\ln(1/(1-(1-\beta)^{1/L}))-\ln2)}{\alpha},\lowerbound = \upperbound)$
    \ElsIf{$q$.type==\tcq (i.e., $q_{W,k}$)}
        \State {\bf return} $(\upperbound = \frac{\|\workload\|_12(\ln (L/(2\beta)))}{\alpha}, \lowerbound = \upperbound)$
    \EndIf
\EndFunction
\end{algorithmic}}
\end{algorithm}

Algorithm~\ref{algo:lm} provides the $\runFunc$ and $\transFunc$ of Laplace mechanism for all three query types.
This algorithm first transforms the query $q_W$ and the data $D$
into matrix representation $\workload$ and $\data$.
The $\transFunc$ outputs a lower and upper bound $(\lowerbound,\upperbound)$ for each query type with a given accuracy requirement
and these two bounds are the same as Laplace mechanism is data independent.
However, these bounds vary among query types.
The $\runFunc$ takes the privacy budget computed by $\transFunc(q,\alpha,\beta)$ (Line~3)
and adds the corresponding Laplace noise $[\tilde{x}_1,\ldots,\tilde{x}_L]$
to the true workload counts $\workload\data$.
When $q$ is a \wcq, the noisy counts are returned directly;
when $q$ is an \icq, the bin ids (the predicates)  that have noisy counts $\geq c$ are returned;
when $q$ is a \tcq, the bin ids (the predicates) that have the largest $k$ noisy counts are returned.
Beside the noisy output, the privacy budget consumed by this mechanism is returned as well.
The following theorem summarizes the properties of the two functions $\runFunc$ and $\transFunc$.
\begin{theorem}\label{theorem:pc_lm}
Given a query $q$ where $q$.type $\in \{\wcq,\icq,\tcq\}$,
Laplace mechanism (Algorithm~\ref{algo:lm}) denoted by $M$
can achieve $(\alpha,\beta)$-q.type accuracy
by executing the function $\runFunc(q,\alpha,\beta,D)$ for any $D\in \tabledomain$,
and satisfy differential privacy with a minimal cost of $\transFunc(q,\alpha,\beta).\upperbound$.
\end{theorem}
The accuracy and privacy proof is mainly based on the noise property of Laplace mechanism. Refer to Appendix~\ref{app:lm_proof} for the detailed proof.
\eat{
\begin{proof} The accuracy and privacy proof is mainly based on the noise property of Laplace mechanism.
Refer to Appendix~\ref{app:lm_proof} for the proof.

Give a \wcq $q_W$, for $D\in \tabledomain$,
setting noise parameter $b\leq \frac{\alpha}{\ln(1/(1-(1-\beta)^{1/L}))}$
guarantees that the failing probability
(when at least one of the predicates in $W$ is perturbed with noise more than $\alpha$)
is bounded by $\beta$.

Given an \icq $q_{W,>c}$,  for $D\in \tabledomain$,
setting noise parameter $b\leq \frac{\alpha}{\ln(1/(1-(1-\beta)^{1/L}))-\ln2}$
guarantees that the failing probability
(when at least one of the predicates in $W$ with counts smaller than $c-\alpha$ is incorrectly labeled as $>c$
or when at least one of the predicates in $W$ with counts greater than $c+\alpha$ is incorrectly labeled as $<c$)
is bounded by $\beta$.

Given  a \tcq $q_{W,k}$, for $D\in \tabledomain$,
setting the noise parameter $b \leq \frac{\alpha}{2(\ln(L/(2\beta)))}$
guarantees that the failing probability
(when at least one of the predicates in $W$ with counts smaller than $c_k-\alpha$ is incorrectly included in the top-$k$ output or when at least one of the predicates in $W$ with counts larger than $c_k + \alpha$ is incorrectly excluded from the top-$k$ output) is bounded by $\beta$.

This mechanism also satisfies $\epsilon$-differential privacy,
where $\epsilon=\|\workload \|_1/b$
(the output of $\transFunc()$ by the property of Laplace mechanism and the post-processing property of differential privacy (Theorem~\ref{theorem:post}).

\end{proof}
}

%This baseline translation is based on Laplace mechanism, and it also allows the noisy counts to be shown to analyst without hurting the privacy cost for \icq and \tcq.

%%%%%%%%%%%%%%%%%%%OLD MATERIALS%%%%%%%%%%%%%%%%%%%%%%%%%%%
\eat{
\subsubsection{Laplace Mechanism for \wcq}
Laplace mechanism can be directly used to answer a workload counting query $q_W(\cdot)$,
as shown in Algorithm~\ref{algo:wcq_lm}.
This algorithm will first transform the query $q_W$ and the data $D$
into matrix representation $\workload$ and $\data$ and has two functions.
The first function $\transFunc$ computes the worst minimal privacy cost required
for this mechanism to achieve $(\alpha,\beta)$-\wcq accuracy for the given query $q_W$,
and the second function takes this privacy cost and adds the corresponding noise to the true query answer.

\begin{theorem}\label{theorem:pc_lm}
Given a workload counting query $q_W:\tabledomain \rightarrow \mathbb{R}^L$,
Laplace mechanism (Algorithm~\ref{algo:wcq_lm}) denoted by $LM^{\alpha,\beta}_{W}(\cdot)$,
can achieve $(\alpha,\beta)$-\wcq accuracy by executing function $\runFunc(q_W,\alpha,\beta,D)$,
with the minimal differential privacy cost returned by function $\epsilon=\transFunc(q_W,\alpha,\beta)$.
%$\epsilon=\|\workload\|_1/b= \frac{\|\workload\|_1\ln(1/(1-(1-\beta)^{1/L}))}{\alpha}$.
\end{theorem}
\begin{proof}
\eat{%%%%%%%%%%loose bound
Set $b \leq \frac{\alpha}{\ln(L/\beta)}$. Then we have
\begin{eqnarray}
&&\Pr[\|LM(\workload,\data)-q_W(D)\|_{\infty} \geq \alpha] \nonumber \\
&\leq & \Pr[\cup_{i\in[1,L]} |\eta_i|\geq \alpha] \leq \sum_{i\in[1,L]} \Pr[|\eta_i|\geq \alpha] \leq L\cdot e^{-\alpha/b} \leq \beta \nonumber
\end{eqnarray}
However, the bound is not really tight here.}%%%%%%%%%%loose bound
To bound the failing probability for any table $D\in \tabledomain$,
\begin{eqnarray}
&&\Pr[\|LM(\workload,\data)-q_W(D)\|_{\infty} \geq \alpha] \nonumber \\
&=& 1-\prod_{i\in[1,L]} \Pr[|\eta_i| < \alpha] = 1-(1- e^{-\alpha/b})^L \leq \beta \nonumber
\end{eqnarray}
it requires $b\leq \frac{\alpha}{\ln(1/(1-(1-\beta)^{1/L}))}$.
This mechanism also satisfies $\epsilon = \|\workload\|_1/b$-differential privacy,
and hence setting $b= \frac{\alpha}{\ln(1/(1-(1-\beta)^{1/L}))}$ gives the least cost.
\end{proof}

If $W$ itself is a histogram query (a horizontal partition of the table), then $\|\workload\|_1=1$.
At worst case,  $\|\workload\|_1 = L$ for example for a prefix counting query as all queries will be affected by adding or removing a tuple from the table.

When $W$ is a single counting query, i.e. $L=1$, then the noise parameter for Laplace mechanism is just $b=\frac{\alpha}{\ln(1/\beta)}$
and the minimal privacy cost is $\epsilon = \frac{\ln(1/\beta)}{\alpha}$.

\subsubsection{Laplace Comparison Mechanism for \icq}
An iceberg counting query (\icq) can also be answered using
Laplace mechanism, similar to \wcq.
As shown in Algorithm~\ref{algo:icq-lcm},
noise $\eta$ drawn from Laplace distribution with parameter
$b=\frac{\alpha}{\ln(1/(1-(1-\beta)^{1/L}))-\ln2}$
is added to the differences between $\workload\data$ and $c$.
If the noisy difference for $\phi_i\in W$ is greater than 0,
the mechanism adds $\phi_i$ to the output.
This mechanism, referred as \emph{Laplace comparison mechanism}
has the following property.

\begin{theorem}\label{theorem:pc_lcm}
Given an iceberg counting query $q_{W,>c}:\tabledomain \rightarrow O$,
for a table  $D\in \tabledomain$,
The Laplace comparison mechanism (Algorithm~\ref{algo:icq-lcm})
denoted by $LCM^{\alpha,\beta}_{W,>c}(\cdot)$,
can achieve $(\alpha,\beta)$-\icq accuracy by executing function $\runFunc(q_{W,>c},\alpha,\beta,D)$,
with minimal differential privacy cost returned by function $\epsilon=\transFunc(q_W,\alpha,\beta)$.
%$\epsilon=\frac{\|\workload\|_1 (\ln(1/(1-(1-\beta)^{1/L}))-\ln2)}{\alpha}$.
\end{theorem}
\begin{proof}(sketch)
To bound the failing probability for any $D\in \tabledomain$,
\begin{eqnarray}
&& \Pr[|\{\phi \in M(D) ~|~ c_{\phi}(D) < c-\alpha\}|>0] \nonumber \\
&=& 1- \Pr[|\{\phi \in M(D) ~|~ c_{\phi}(D) < c-\alpha\}|=0] \nonumber \\
&=& 1- \prod_{\phi\in W: c_{\phi}(D)-c+\alpha<0} (1-\Pr[c_{\phi}(D)-c+\eta>0]) \nonumber\\
&\leq & 1- \prod_{\phi\in W: c_{\phi}(D)-c+\alpha<0} (1-\Pr[\eta>\alpha]) \nonumber\\
%&=& 1- \prod_{\phi\in W: c_{\phi}(D)- c+\alpha<0} (1-e^{-\alpha/b}/2) \nonumber \\
&\leq& 1-(1-e^{-\alpha/b}/2)^L < \beta \nonumber
\end{eqnarray}
and similarly,
\begin{eqnarray}
&& \Pr[|\{\phi \in (W-M(D)) ~|~ c_{\phi}(D) > c+\alpha\}|>0] \nonumber \\
%&=& 1- \Pr[|\{\phi \in (W-M(D)) ~|~ c_{\phi}(D) > c+\alpha\}|=0] \nonumber \\
%&=& 1 - \prod_{\phi \in W: c_{\phi}(D)- c-\alpha>0} (1-\Pr[c_{\phi}(D)-c+\eta<0]) \nonumber\\
%&\leq & 1-  \prod_{\phi \in W: c_{\phi}(D)- c-\alpha>0} (1-\Pr[\eta<-\alpha]) \nonumber\\
%&=& 1- \prod_{\phi\in W: c_{\phi}(D)- c-\alpha >0} (1-e^{-\alpha/b}/2) \nonumber \\
&\leq& 1-(1-e^{-\alpha/b}/2)^L < \beta \nonumber
\end{eqnarray}
it requires $b\leq \frac{\alpha}{\ln(1/(1-(1-\beta)^{1/L}))-\ln2}$.
This mechanism also satisfies $\epsilon = \|\workload \|_1/b$-differential privacy
by post-processing (Theorem~\ref{theorem:post}).
Hence, setting $\epsilon=\frac{\|\workload\|_1 (\ln(1/(1-(1-\beta)^{1/L}))-\ln2)}{\alpha}$ gives the least cost.
\end{proof}

When $W$ is a single counting query, i.e., $L=1$,
then the noise parameter used in Algorithm~\ref{algo:icq-lcm} is $b=\frac{\alpha}{(\ln(1/\beta)-\ln 2}$,
and hence the minimal privacy cost is $\frac{\ln(1/\beta)-\ln 2}{\alpha}$.

Alternatively, the analyst can pose a workload counting query $q_{W}$
with $(\alpha,\beta)$-\wcq tolerance via \system,
and then use the noisy answer of $q_{W}(D)$ to learn  $q_{W,>c}(D)$ locally.
This approach adds a smaller expected noise to $q_{W}(D)$,
and hence achieves $(\alpha,\beta)$-\icq tolerance, i.e.,
\begin{lemma}
Using the output of a Laplace mechanism $LM^{\alpha,\beta}_{W}$
to answer $q_{W,>c}$ can achieve $(\alpha',\beta)$-\icq accuracy,
 where $\alpha'=(1-\frac{\ln 2}{\ln(1/(1-(1-\beta)^{1/L}))})\alpha <\alpha$.
\end{lemma}
This approach also allows a data analyst  to make more local decisions,
for example, how much to adjust the threshold of a similarity function,
but consumes a larger privacy cost of $\epsilon=\frac{\|\workload\|_1 (\ln(1/(1-(1-\beta)^{1/L})))}{\alpha}$
compared to  $LCM^{\alpha,\beta}_{W,>c}$ with a privacy cost of
$\frac{\|\workload\|_1 (\ln(1/(1-(1-\beta)^{1/L}))-\ln2)}{\alpha}$.

\subsubsection{Top-$k$ Counting Query for \tcq} \label{sec:noise_topk}
}

\subsection{Special Translation for \wcq}\label{sec:translation_wcq}
\begin{algorithm}[t]
  \caption{\wcq-SM $(q_W,\alpha,\beta,D,\strategy)$}\label{algo:wcq-sm}
{\small \begin{algorithmic}[1]
\State Initialize $\workload \gets \transform(W), \data \gets \transform_W(D), \alpha,\beta,\strategy$
\Function{\runFunc}{$q_W,\alpha,\beta, D$}
\State $\epsilon \gets \transFunc(\workload,\alpha,\beta).\upperbound$
    \State $\omega \gets \workload\strategy^+(\strategy\data + Lap(b)^l)$, where $b=\|A\|_1/\epsilon$
    \State {\bf return} $(\omega, \epsilon)$
\EndFunction
\Function{\transFunc}{$q_W,\alpha,\beta$}
\State Set $u = \frac{\|\strategy\|_1\|\workload\strategy^+\|_F}{\alpha\sqrt{\beta/2}}$ and  $l =0$
\State $\epsilon=\textsc{binarySearch}(l,u,\textsc{estimateBeta}(\cdot,\alpha,\beta,\workload\strategy^+))$
%\Repeat
%    \State Set $\epsilon = (l+u)/2$
%     \State $\beta_e \gets \textsc{estimateBetaMC}(\epsilon,\alpha,\beta,\workload\strategy^+)$
%     \If{$\beta_e < \beta$}
%         \State $u = \epsilon$
%     \Else
%         \State $l = \epsilon$
%     \EndIf
%\Until{$u-l< 10^{-5}$}
\State {\bf return} $(\upperbound=\epsilon, \lowerbound=\epsilon)$
\EndFunction
\Function{estimateBeta}{$\epsilon,\alpha, \beta,\workload\strategy^+$}
\State Sample size $N = 10000$ and failure counter $n_f=0$
    \For{$i\in [1,\ldots,N]$}
          \State Sample noise $\noisev_i ~\sim Lap(\|\strategy\|_1/\epsilon)^l$
          \If{$\|(\workload\strategy^+)\noisev_i\|_{\infty}>\alpha$}
               \State $n_f$++
          \EndIf
     \EndFor
     \State $\beta_e = n_f/N$, $p=\beta/100$
     \State $\delta\beta = z_{1-p/2}\sqrt{\beta_e(1-\beta_e)/N}$ %p=0.000002, 4.75342 , with 1-p chance, beta_e +/- \delta\beta
     %http://web.mst.edu/~dux/repository/me360/ch8.pdf
     \State {\bf return} $(\beta_e+\delta\beta+p/2)<\beta$
\EndFunction
\end{algorithmic}}
\end{algorithm}

The privacy cost of the Laplace mechanism increases linearly with $\|\workload\|\change{_1}$,
the sensitivity of the workload in the query. For example, a prefix workload
has a sensitivity equals to the workload size $L$.
When $L$ is very large, the privacy cost grows drastically.
To address this problem,
\system provides a special translation for \wcq, called {\em strategy-based mechanism}.
This mechanism considers a different {\em strategy} workload $\strategy$
such that (i) $\strategy$ has a low sensitivity $\|\strategy \|_1$,
and (ii) rows in $\workload$ can be reconstructed using
a small number of rows in $\strategy$.

Let strategy matrix $\strategy$ be a $l\times |dom_W(R)|$ matrix,
and $\strategy^+$  denote its Moore-Penrose pseudoinverse,
such that $\workload\strategy\strategy^+=\workload$.
Given such a strategy workload $\strategy$,
we can first answer $\strategy$ using Laplace mechanism,
i.e. $\hat{y} =\strategy\data + \noisev$, where $\noisev\sim Lap(b)^l$ and $b=\frac{\|\strategy\|_1}{\epsilon}$,
and then reconstruct answers to $\workload$
from the noisy answers to $\strategy$ (as a postprocessing step), i.e. $(\workload \strategy^+)\hat{y}$.
This approach is formally known as the {\em Matrix mechanism}~\cite{Li:2010:OLC:1807085.1807104, Li:2015:MMO:2846574.2846647}
and shown in the $\runFunc$ of Algorithm~\ref{algo:wcq-sm}.
If a strategy $\strategy$ is used for this mechanism,
we denote it by $\strategy$-strategy mechanism.
In this work, we consider several popular strategies
in prior work~\cite{Li:2010:OLC:1807085.1807104, Li:2015:MMO:2846574.2846647, hdmm18},
such as hierarchical matrix $H_2$.
Techniques like HDMM~\cite{hdmm18}
can automatically solve for a good strategy (but this is not our focus).

\eat{
Let strategy matrix $\strategy$ be a $l\times |dom_W(R)|$ matrix,
and $\strategy^+$  denote its Moore-Penrose pseudoinverse,
such that $\workload\strategy\strategy^+=\workload$.
The matrix mechanism is given by the following:
\begin{eqnarray}
M_{\strategy}(\workload,\data) = \workload\data + \workload\strategy^+Lap(\|\strategy\|_1/\epsilon)^l
\end{eqnarray}
}

\eat{
\begin{example}
Consider a workload over attribute $Income$,
$W=\{Income<100$, $Income<200$, $\ldots$, $Income<999900$, $Income <1000000\}$.
The transformation based on $W$ gives a partition of the domain of rows
into a set of ranges over Income $dom_W(R) =\{[100,200),[200,300),\ldots,[999900,1000000)\}$.
Hence, the corresponding data vector is $\data = [\data[0],\ldots,\data[L-1]]^T$,
where $\data[i]$ corresponds to the number of rows in $D$
with attribute Income falling within $[100i, 100(i+1))$.
The corresponding workload matrix $\workload$ of size $L \times L$
is known as {\em prefix workload},
\[
\workload=
  \begin{bmatrix}
    1 & 0 & 0 & \ldots & 0 \\
    1 & 1 & 0 & \ldots & 0 \\
    1 & 1 & 1 & \ldots & 0 \\
      & \ldots  &   & \ldots & 0 \\
    1 & 1 & 1 & \ldots & 1 \\
  \end{bmatrix}
\]
The sensitivity for this query is $\|\workload\|_{1} = L$.
Hence, using the baseline translation,
the expected worst error for $q_W$ is $O(|W|)$.

If we consider a strategy matrix $\strategy$ known as $H_2$ strategy,
which corresponds to a binary tree over the domain $dom_W(R)$, e.g., when $|W|=4$,
\[
\strategy=
  \begin{bmatrix}
    1 & 1 & 0 & 0 \\
    0 & 0 & 1 & 1 \\
    1 & 0 & 0 & 0 \\
    0 & 1 & 0 & 0 \\
    0 & 0 & 1 & 0 \\
    0 & 0 & 0 & 1
  \end{bmatrix}
\]
Using this strategy matrix, the error for each noisy count for $\strategy$ is $O(\log^2(L))$,
as the sensitivity of this strategy is only $O(\log(L))$.
For each row in $\workload$ requires at most $O(\log(L))$ to reconstruct,
and hence the error in answering $\workload$ using $H_2$-strategy is $O(\log^3(|W|))$.
\eat{However, if consider a histogram workload over attribute $Income$,
$W'=\{Income<1000$, $1000\leq Income<2000$, $\ldots$, $999000 \leq Income <1000000\}$
the expected worst error for $q_W'$ is $O(1)$
which is much smaller than the error using $H_2$-strategy-based mechanism, $O(\log^3(|W'|))$.}
\end{example}
}

\eat{
We would like to minimize the largest error added to the set of counting queries, i.e.,
\begin{eqnarray}
\min_{\strategy} \mathbb{E}\|(\workload\strategy^+)\noisev\|_{\infty}.
\end{eqnarray}
HDMM~\cite{hdmm18} can be applied to find the optimal strategy efficiently,
with minor change on optimizing the maximum error instead of the total error in expectation.
Suppose $\strategy$ is identified as the optimal strategy
that achieves the smallest maximum error in expectation for a given privacy cost $\epsilon$,
what is the minimal privacy cost for strategy-based mechanism using $\strategy$ to achieve $(\alpha,\beta)$-\wcq?
}

However, translating the accuracy requirement on WCQ-SM is nontrivial as the answers to the query $q_W(D)$ %, $\|(\workload\strategy^+)\noisev)\|_{\infty}$,
are linear combinations of noisy answers.
The errors are due to the sums of weighted Laplace random variables
which have non-trivial CDFs for the accuracy translation.
Hence, we propose an accuracy to privacy translation method
shown in the $\transFunc$ function of Algorithm~\ref{algo:wcq-sm}.
We first set an upper bound $u$ for the privacy to achieve $(\alpha,\beta)$-\wcq accuracy
based on Theorem~\ref{theorem:strategy_upperbound} (Appendix~\ref{app:strategy_proof})
and conduct a binary search on a privacy cost $\epsilon$ between $l$ and $u$ (Line~9)
such that  the failing probability to bound the error by $\alpha$ equals to $\beta$.
During this binary search,
for each $\epsilon$ between $l$ and $u$, we run Monte Carlo simulation to learn the empirical failing rate $\beta_e$
to bound the error by $\alpha$ shown in the Function $\textsc{estimateBeta}()$
such that with high confidence $1-p$, the true failing probability $\beta_t$ lies within $\beta_e \pm \delta\beta$.
If the empirical failing rate $\beta_e$ is sufficiently smaller than $\beta$, then the upper bound is set to $\epsilon$;
otherwise the lower bound is set to $\epsilon$. The next value for $\epsilon$ is $(l+u)/2$.
This search process stops when $l$ and $u$ is sufficiently small.
This approach can be generalized to all data-independent differentially private mechanisms.
The simulation can be done offline.

\begin{theorem}\label{theorem:strategy_tightbound}
Given a workload counting query $q_W:\tabledomain \rightarrow \mathbb{R}^L$.
Answering $q_W$ with $\strategy$-strategy mechanism
by executing $\runFunc(q_W,\alpha,\beta,D)$ in Algorithm~\ref{algo:wcq-sm}
achieves $(\alpha,\beta)$-\wcq accuracy for any $D\in \tabledomain$
with an approximated minimal privacy cost of $\transFunc(q_W,\alpha,\beta).\upperbound$.
\end{theorem}
Refer to Appendix~\ref{app:strategy_proof} for the proof.
\eat{
\begin{proof}
Refer to Appendix~\ref{app:strategy_proof}.
Given an $\epsilon$, the simulation in the Function $\textsc{estimateFailingRateMC}()$ ensures that
with high probability $1-p$, the true failing probability $\beta_t$ to
bound $\| (\workload\strategy^+)\eta\|_{\infty}$ by $\alpha$
lies within $\beta_e\pm \delta \beta$.
The failing probability to bound $\beta_t <\beta_e+\delta\beta$ is $p/2$.
By union bound, this $\epsilon$ ensures $(\alpha,\beta')$-\wcq accuracy, where $\beta'< \beta+\delta\beta+p/2$.
If $(\beta+\delta\beta)(1-p/2)+p/2 < \beta+\delta\beta +p/2 <\beta$, then this $\epsilon$ ensures $(\alpha,\beta)$-\wcq accuracy.
Beside this estimation, in the binary search of $\transFunc()$,
we stop when $\epsilon_{\min}$ and $\epsilon_{\max}$ are sufficiently close.
Hence, the privacy cost returned by $\transFunc()$ is an approximated
minimal privacy cost required for $(\alpha,\beta)$-\wcq accuracy.
\end{proof}
}
In the future work, we can add similar optimizers like HDMM~\cite{hdmm18}
which finds the optimal strategy efficiently for a given query.

%Briefly discuss Gaussian noise.
%Post-processing Optimization

\eat{
This approach also allows a data analyst  to make more local decisions,
for example, how much to adjust the threshold of a similarity function,
but consumes a larger privacy cost of $\epsilon=\frac{\|\workload\|_1 (\ln(1/(1-(1-\beta)^{1/L})))}{\alpha}$
compared to  $LCM^{\alpha,\beta}_{W,>c}$ with a privacy cost of
$\frac{\|\workload\|_1 (\ln(1/(1-(1-\beta)^{1/L}))-\ln2)}{\alpha}$.
}

\subsection{Special Translation for \icq} \label{sec:translation_icq}
The strategy-based mechanism that we used for \wcq
can be adapted to answer \icq
if used in conjunction with a post-processing step (Section~\ref{sec:icq_sm}).
%We note that for \icq not all queries need be evaluated due to dependencies between the queries.
%We use this observation to design an Ordered Mechanism (Section~\ref{sec:icq_om}).
We also present another novel data dependent translation strategy for \icq
that may result in different privacy loss for different datasets
given the same accuracy requirement (Section~\ref{sec:icq_mpm}).

\eat{
Similar to \wcq, strategy-based mechanisms are also applicable to \icq
with a post-processing step, shown in Section~\ref{sec:icq_sm}.
We will then present several novel differentially private mechanisms
with data-dependent privacy cost for \icq.
For these mechanisms, we will see the privacy budget consumed at the end of $\runFunc()$
is possibly smaller than the predication given by the function $\transFunc()$.
}

\subsubsection{Strategy-based Mechanism (\icq-SM)}\label{sec:icq_sm}
The analyst can pose a workload counting query $q_{W}$
with $(\alpha,\beta)$-\wcq requirement via \system,
and then use the noisy answer of $q_{W}(D)$ to learn $q_{W,>c}(D)$ locally.
This corresponds to a post-processing step of a differentially private mechanism
and hence still ensures the same level of differential privacy guarantee (Theorem~\ref{theorem:post}).
On the other hand, $(\alpha,\beta)$-\icq accuracy
only requires to bound one-sided noise by $\alpha$ with probability $(1-\beta)$,
and $(\alpha,\beta)$-\wcq accuracy requires to bound two-sided noise with the same probability.
Hence, if a mechanism has a failing probability of $\beta$ to bound the error for \wcq,
then using the same mechanism has a failing probability of $\beta/2$ to bound the error for \icq.

\eat{
\begin{lemma}
Using the output of the function $\runFunc(q_W,\alpha,\beta,D)$
of \wcq-SM in Algorithm~\ref{algo:wcq-sm}, denoted by \icq-SM,
to answer $q_{W,>c}$ achieves $(\alpha,\beta')$-\icq accuracy
where $\beta' <\beta$
%to answer $q_{W,>c}$ achieves $(\alpha',\beta)$-\icq accuracy,
%where $\alpha'<\alpha$, %$\alpha'=(1-\frac{\ln 2}{\ln(1/(1-(1-\beta)^{1/L}))})\alpha <\alpha$.
and with an approximated minimal privacy cost of
$\transFunc(q_W,\alpha,\beta).\upperbound$.
\end{lemma}
}

\eat{
\subsubsection{Ordering Mechanism (\icq-OM)}\label{sec:icq_om}
We observe that there is ordering relationship between predicates within the same workload for many exploration queries,
such as a prefix workload $\{Age<10, Age<20,\ldots, Age<100\}$, or a hierarchical workload over a 2D space.
To answer an \icq with such a workload and threshold $c$,
it is unnecessary to test all the predicates whether they have a count greater than $c$.
If the number of people with $Age<10$ is more than $c$, then the number of people with $Age<20$ is also more than $c$.
Hence, given $c_{\phi_i}(D)<c_{\phi_j}(D)$ for any $D \in \tabledomain$,
we  add to $\phi_j$ to the output if the noisy count for $\phi_i$ passes the threshold without evaluating $\phi_j$
or prune $\phi_i$ if the noisy count for $\phi_j$ fails the threshold without evaluating $\phi_i$.
To enable this pruning process, we build a directed acyclic graph $G_W$ from $W$,
where $W$ form the set of nodes in $G_W$ and there is a directed edge from $\phi_i$ to $\phi_j$
if $\phi_i \rightarrow \phi_j$, which implies $c_{\phi_i}(D)\leq c_{\phi_j}(D)$ for any $D \in \tabledomain$.
Hence, this graph has the following property:
for every path $\phi_{i_1}\rightarrow\phi_{i_2}\cdots \rightarrow \phi_{i_j}$ found in $G_W$,
it is true that  $c_{\phi_{i_1}}(D)\leq c_{\phi_{i_2}}(D) \cdots \leq c_{\phi_{i_j}}(D)$
for any $D\in \tabledomain$.

\begin{algorithm}[t]
\caption{\icq-OM($q_{W,>c}, \alpha,\beta,D$)}\label{algo:icq-om}
{\small \begin{algorithmic}[1]
\State Initialize $G_W(V=\{\phi_1,\ldots,\phi_L\}, E=\{(\phi_i,\phi_j)|\phi_i \rightarrow \phi_j\},\alpha,\beta, D$
\Function{\runFunc}{$q_{W,>c}, \alpha,\beta,D$}
    \State $\mathcal{P} \gets {\tt argmin}_{P_1,\ldots,P_l: \cup (P_i.V)=G_W.V} \sum_{i=1}^l \ceil{1+\log_2|P_i|}$
    \State $\omega=[],\epsilon=0$
    \While{$\mathcal{P}$ is not empty}
        \State Pick a longest path $P=(\phi_{i_1}\rightarrow \phi_{i_2}\cdots \rightarrow \phi_{i_{|P|}})$ in $\mathcal{P}$
        \State Set $l=1$ and $r=|P|$
        \Repeat
            \State Set $m  = (l+r)/2$ and $\phi = \phi_{i_m}$
            \State $\tilde{x}=q_{\phi}(D)-c+ Lap(b)$, where $b=\frac{\alpha}{\ln(L/2\beta)}$
            \If{$\tilde{x} \geq 0$}
                \State $\omega.{\tt add}(\phi,r_{W_G}(\phi\rightarrow *))$ and remove them from $\mathcal{P}$
                \State Set $r=m-1$
            \Else
                \State Remove $r_{W_G}(*\rightarrow\phi)$ from $\mathcal{P}$
                \State Set $l=m+1$
            \EndIf
            \State $\epsilon = \epsilon + 1/b$
        \Until{$l=r$}
    \EndWhile
    \State {\bf return} $(\omega,\epsilon)$
\EndFunction
\Function{\transFunc}{$q_{W,>c},\alpha,\beta$}
    \State $k \gets \min_{P_1,\ldots,P_l: \cup (P_i.V)=G_W.V} ~~~~~\sum_{i=1}^l \ceil{1+\log_2|P_i|}$
    \State {\bf return} $\epsilon = \frac{k\ln(L/2\beta)}{\alpha}$
\EndFunction
\end{algorithmic}}
\end{algorithm}

Based on this structure, Algorithm~\ref{algo:icq-om} captures the ordering property
within the predicates in the workload to prune unnecessary comparisons
to save the privacy budget for achieving $(\alpha,\beta)$-\icq accuracy.
In this algorithm, $r_{W_G}(\phi \rightarrow *)$ denotes the set of nodes that are reachable from $\phi$ in $G_W$ and
$r_{W_G}(* \rightarrow \phi)$ denotes the set of nodes that reach $\phi$ in $G_W$. %(excluding $\phi$ itself).
The function $\transFunc()$ partitions $G_W$ into a set of paths $\mathcal{P}$
such that the total number of comparisons is minimized
and computes the worst privacy cost if all the paths are evaluated.
In $\runFunc()$, along each path in $\mathcal{P}$,
we start a binary search on the predicate.
If the predicate $\phi$ with noisy count $\geq 0$, all the predicates reachable from $\phi$ are added to the output,
and removed from $\mathcal{P}$; otherwise, all the nodes that reach $\phi$ are removed from $\mathcal{P}$.
\eat{
we just need to identify one $\phi$ with noisy count sufficiently bigger than the threshold,
then all the predicates that are reachable from $\phi$ can be included in the output
and all the predicates that reaches $\phi$ can be removed from the graph.
}

\begin{theorem}\label{theorem:pc_om}
Given an iceberg counting query $q_{W,>c}:\tabledomain \rightarrow O$,
for a table  $D\in \tabledomain$,
the ordering mechanism (Algorithm~\ref{algo:icq-om}) for \icq
can achieve $(\alpha,\beta)$-\icq accuracy by executing function $\runFunc(q_{W,>c},\alpha,\beta,D)$,
with differential privacy cost of $\transFunc(q_W,\alpha,\beta).\upperbound$.
\end{theorem}
\begin{proof}(sketch) There are $k\log_2(L)$ number of noisy comparisons at most.
To ensure $(\alpha,\beta/L)$-\icq accuracy for each predicate,
the total privacy loss is at most $\frac{k\ln(L/(2\beta))}{\alpha}$.
\end{proof}

%\todo{How to take account of parallel composition between some predicates?}
}

\subsubsection{Multi-Poking Mechanism (\icq-MPM)}\label{sec:icq_mpm}
We propose a data-dependent translation for \icq,
which can be used as a subroutine for mechanisms that involve threshold testing.
For ease of explanation,
we will illustrate this translation with a special case of \icq
when the workload size $L=1$, denoted by, $q_{\phi,>c}(\cdot)$.
Intuitively, when $c_{\phi}(D)$ is much larger (or smaller) than $c$,
then a much larger (smaller resp.) noise can be added to $c_{\phi}(D)$
without changing the output of \system. Consider the following example.

\begin{example}
Consider a query $q_{\phi,>c}$, where $c=100$.
To achieve $(\alpha, \beta)$ accuracy for this query,
where $\alpha=10, \beta=0.1^{10}$, the
Laplace mechanism requires a privacy cost of $\frac{\ln(1/(2\beta))}{\alpha} = 2.23$
by Theorem~\ref{theorem:pc_lm}, regardless of input $D$.
Suppose $c_{\phi}(D)=1000$.
In this case, $c_{\phi}(D)$ is much larger than the threshold $c$,
and the difference is $\frac{(1000-100)}{\alpha}=90$ times of the accuracy bound $\alpha=10$.
Hence, even when applying Laplace comparison mechanism
with a privacy cost equals to $\frac{2.23}{90}\approx 0.25$
wherein the noise added is bounded by $90\alpha$ with high probability $1-\beta$,
the noisy difference $c_{\phi}(D)-c+\eta_{sign}$ will still be greater than 0 with high probability.
\end{example}

This is an example where a different mechanism
rather than Laplace mechanism
achieves the same accuracy with a smaller privacy cost.
Note that the tightening of the privacy cost in this example
requires to know the value of $c_{\phi}(D)$.
It is difficult to determine a privacy budget for poking without looking at the query answer.
To tackle this challenge, we propose an alternative approach that allows
of $m$ \emph{pokes} with increasing privacy cost.
This approach is summarized in Algorithm~\ref{algo:mpm}
as Multi-Poking Mechanism (MPM).
This approach first computes the privacy cost if all $m$ pokes are needed,
$\epsilon_{\max} = \frac{\ln(m/(2\beta))}{\alpha}$.
The first poke checks if bins have either sufficiently large noisy differences $\tilde{\ans}$
with respect to the accuracy $\alpha_{0}$ for the current privacy cost (Lines 8-10).
If this is true (Lines 10), then the set of predicates with sufficiently large positive differences is returned;
otherwise, the privacy budget is relaxed with additional $\epsilon_{\max}/m$.
At $(i+1)$th iteration, instead of sampling independent noise,
we apply the $\relaxprivacy$ Algorithm (details refer to \cite{DBLP:journals/corr/KoufogiannisHP15a})
to correlate the new noise $\noisev_{i+1}$ with noise $\noisev_i$ from the previous iteration.
In this way, the privacy loss of the first $i+1$ iterations is $\epsilon_{i+1}$,
and the noise added in the $i+1$th iteration is equivalent to
a noise generated with Laplace distribution with privacy parameter $b=(1/\epsilon_{i+1})$.
This approach allows the data analyst to learn the query answer with a gradual relaxation of privacy cost.
This process repeats until all $\epsilon_{\max}$ is spent.
We show that Algorithm~\ref{algo:mpm} achieves both accuracy and privacy requirements.

\begin{algorithm}[t]
\caption{\icq-MPM($q_{W,>c},\alpha,\beta,D,m$)} \label{algo:mpm}
{\small \begin{algorithmic}[1]
\State Initialize $\workload \gets \transform(W), \data \gets \transform_W(D), \alpha,\beta, m=10$
\Function{\runFunc}{$q_{W,>c},\alpha,\beta,D$}
\State Compute  $\epsilon_{\max} = \transFunc(q_{W,>c},\alpha,\beta).\upperbound$
\State Initial privacy cost $\epsilon_{0} = \epsilon_{\max}/m$
\State $\tilde{\ans}_0 =  \workload\data -c + \noisev_0$, where $\noisev_0 \sim Lap(\|\workload\|_1/\epsilon_{0})^L$
\For{$i=0,1,\ldots,m-2$}
    \State Set $\alpha_{i} = \|\workload\|_1\ln (mL/(2\beta))/\epsilon_{i}$
    \State $W_{+} \gets \{\phi_j \in W~~~|~~~(\tilde{\ans}_i[j] - \alpha_{i})/\alpha \geq -1\}$
    \State $W_{-} \gets \{\phi_j \in W~~~|~~~(\tilde{\ans}_i[j] + \alpha_{i})/\alpha \leq 1\}$
    %\For{$j=1,\ldots, L$}
    %    \If {$(\tilde{\ans}_i[j] - \alpha_{i})/\alpha \geq -1$}
    %        \State Add $\phi_j$ to $O_{+}$
    %    \ElsIf {$(\tilde{\ans}_i[j] + \alpha_{i})/\alpha \leq 1$}
    %        \State Add $\phi_j$ to $O_{-}$
    %    \EndIf
    %\EndFor
    \If{$ (W_{+}\cup W_{-}) = W$}
        \State \textbf{return} $(W_{+},\epsilon_i)$
    \Else
        \State Increase privacy budget $\epsilon_{i+1} = \epsilon_{i} + \epsilon_{\max}/m$
        \For{$j=1,\ldots,L$}
            \State $\noisev_{i+1}[j]  = \relaxprivacy(\noisev_i[j],\epsilon_{i},\epsilon_{i+1})$
               \cite{DBLP:journals/corr/KoufogiannisHP15a}
         \EndFor
         \State New noisy difference $\tilde{\ans}_{i+1} = \workload\data -c + \noisev_{i+1}$
    \EndIf
\EndFor
\State \textbf{return} $(\{\phi_j\in W~~~|~~~\tilde{\ans}_{m-1}[j] >0 \} ,\epsilon_{\max})$
\EndFunction
\Function{\transFunc}{$q_{W,>c},\alpha,\beta$}
\State \textbf{return} $\upperbound=\frac{\|\workload\|_1\ln(mL/(2\beta))}{\alpha}, \lowerbound = \frac{\upperbound}{m}$
\EndFunction
\end{algorithmic}}
\end{algorithm}

\eat{
\begin{algorithm}[t]
\caption{\icq-MPM($q_{\phi,>c},\alpha,\beta,D,m$)} \label{algo:mpm}
{\small \begin{algorithmic}[1]
\State Initialize $q_{\phi,>c},\alpha,\beta, D,f,m$
\Function{\runFunc}{$q_{\phi,>c},\alpha,\beta,D$}
\State Compute  $\epsilon_{\max} = \transFunc(q_{\phi,>c},\alpha,\beta)$
\State Initial privacy cost $\epsilon_{0} = \epsilon_{\max}/m$
\State $\tilde{x}_0 =  q_{\phi}(D) -c + \eta_0$, where $\eta_0 = Lap(1/\epsilon_{0})$
\For{$i=0,1,\ldots,m-2$}
    \State Set $\alpha_{i} = \ln (m/(2\beta))/\epsilon_{i}$
    \If {$(\tilde{x}_i - \alpha_{i})/\alpha \geq -1$}
        \State \textbf{return} $(\phi,\epsilon_i)$
    \ElsIf {$(\tilde{x}_i + \alpha_{i})/\alpha \leq 1$}
        \State \textbf{return} $(\emptyset,\epsilon_i)$
    \Else
        \State Increase privacy budget $\epsilon_{i+1} = \epsilon_{i} + \epsilon_{\max}/m$
        \State Update noise $\eta_{i+1}  = NoiseDown(\eta_i,\epsilon_{i},\epsilon_{i+1})$
               \cite{DBLP:journals/corr/KoufogiannisHP15a}
        \State New noisy difference $\tilde{x}_{i+1} = q_{\phi}(D) -c + \eta_{i+1}$
    \EndIf
\EndFor
\If {$\tilde{x}_{m-1} >0$}
    \State \textbf{return} $(\phi,\epsilon_{\max})$
\Else
     \State \textbf{return} $(\emptyset,\epsilon_{\max})$
\EndIf
\EndFunction
\Function{\transFunc}{$q_{\phi,>c},\alpha,\beta$}
\State \textbf{return} $\epsilon=\ln(m/(2\beta))/\alpha$
\EndFunction
\end{algorithmic}}
\end{algorithm}

\begin{theorem}\label{theorem:pc_mpm}
Given a query $q_{\phi,>c}$,  Multi-Poking Mechanism (Algorithm~\ref{algo:mpm}),
achieves $(\alpha,\beta)$-\icq accuracy by executing function $\runFunc(q_{\phi,>c}, \alpha,\beta,D)$,
with differential privacy cost of $\transFunc(q_{\phi,>c},\alpha,\beta).\upperbound$.
\end{theorem}
\begin{proof}
Refer to Appendix~\ref{app:mpm_proof}.
\end{proof}
}

\begin{theorem}\label{theorem:pc_mpm}
Given a query $q_{W,>c}$,  Multi-Poking Mechanism (Algorithm~\ref{algo:mpm}),
achieves $(\alpha,\beta)$-\icq accuracy by executing function $\runFunc(q_{W,>c}, \alpha,\beta,D)$,
with differential privacy of  $\transFunc(q_{W,>c},\alpha,\beta).\upperbound$.
\end{theorem}
Refer to Appendix~\ref{app:mpm_proof} for proof.
The privacy loss of multi-poking mechanism at the worst case (the value returned by $\transFunc$)
is greater than that of the baseline LM, but this mechanism may stop before $\epsilon_{\max}$ is used up,
and hence it potentially saves privacy budget for the subsequent queries.

%In fact, the privacy loss of the multi-poking mechanism can be a $\frac{\ln (m-1)/2\beta}{m}$ fraction
%of the privacy loss of Laplace mechanism, if this mechanism returns in the first iteration.
%We also show in the evaluation that this optimization
%allows more queries to be posed within the overall budget constraint $B$,
%resulted in better cleaning quality than LCM.

\eat{
To answer $q_{W,>c}$ where $|W|>$,
we can use \icq-MPM to answer a set of \icq with a single predicate
i.e. $q' = \{q_{\phi,>c} |\phi\in W \}$.
For each query $q_{\phi}\in q'$,
executing \icq-MPM.$\runFunc(q_{\phi},\alpha,\beta/L)$
to achieve $(\alpha,\beta/L)$-\icq accuracy for $q_{W,>c}$.
This will cost a privacy cost of $\frac{\ln(mL/(2\beta))}{\alpha}$.

Moreover, this multi-poking mechanism can
be combined with the ordering mechanism
by replacing the single count comparison
Lines~10-18 in the $\runFunc()$ of Algorithm~\ref{algo:icq-om} by
the multi-poking mechanism.
This further improves the actual privacy cost in many scenarios.
We denote this mechanism by \icq-OMPM,
where the $\transFunc()$ outputs the worst privacy cost
$\epsilon=\frac{k\ln(mL/(2\beta))}{\alpha}$.
}

%%%%%%%%%%%%%%%%%%%%%%%%OLD MATERIALS%%%%%%%%%%%%%%%%
\eat{
\begin{algorithm}[t]k
\caption{LCOM($q_{W,>c}, \alpha,\beta,D$)}
\label{algo:icom}
{\small \begin{algorithmic}[1]
\Require \icq $q_{W,>c}(\cdot)$, $(\alpha,\beta)$-\lcc , table $D$
\Ensure  Answer $a \subseteq W$
\State Build $G_W(V=\{\phi_1,\ldots,\phi_L\}, E=\{(\phi_i,\phi_j)|\phi_i \rightarrow \phi_j\}$
%\State $\epsilon = \frac{8\ln L* + \ln (2/\beta)}{\alpha}$
%\State Let $\tilde{c} = c+ Lap(2/\epsilon)$
\While{$G_W$ is not empty}
     \State Pick a longest path $P=(\phi_{i_1}\rightarrow \phi_{i_2}\cdots \rightarrow \phi_{i_{|P|}})$ in $G_W$
     \State Estimate privacy cost for SVT: $\epsilon_{svt} = \frac{8(\ln(|P|) + \ln(2/\beta))}{\alpha}$
     \State Estimate privacy cost for BS: $\epsilon_{bs} = \frac{\log_2(|P|)\ln(1/(2\beta))}{\alpha}$
     \If{$\epsilon_{svt} < \epsilon_{bs}$}
         \State $\phi \leftarrow SVT(P,\alpha,\beta,D)$ %Run SVT on this path from left to right: output $\phi_i$
         \State Output $\phi$ and $r_{W_G}(\phi\rightarrow *)$  and remove them from $G_W$
         \State Remove $r_{W_G}(*\rightarrow\phi)$ from $G_W$
     \Else
        \State Set $l=i_1$ and $r=i_{|P|}$
        \While{$l \neq r$}
            \State Set $m  = (l+r)/2$ and $\phi = \phi_{i_m}$
            \If{$\phi \in LCM(q_{\{\phi\},>c},\alpha,\beta,D)$}
                \State Output $\phi$ and $r_{W_G}(\phi\rightarrow *)$ and remove them from $G_W$
                \State Set $r=m-1$
            \Else
                \State Remove $r_{W_G}(*\rightarrow\phi)$ from $G_W$
                \State Set $l=m+1$
            \EndIf
        \EndWhile
     \EndIf
\EndWhile
\end{algorithmic}}
\end{algorithm}
}

\eat{
\subsubsection{LCM with Poking}
The first approach is summarized in Algorithm~\ref{algo:poking_lcm},
named as Laplace Comparison Mechanism with Poking (LCMP).
This algorithm first computes the privacy cost of $LCM^{\alpha,\beta}_{q_{\phi,>c}}$
based on Theorem~\ref{theorem:pc_lcm}, denoted by $\epsilon_{LCM}$ (Line~1). It then chooses to run LCM with a \emph{small fraction} of this privacy cost: i.e., it adds $Lap(1/\epsilon_0)$ to the difference $q(D) - c$, with $\epsilon_0 = f \cdot\epsilon_{LCM}$ (Line~2,3). If the noisy difference is too large (Line~5), then LCMP returns `True'. If it is too small (Line~6), LCMP return `False'. In both these cases, LCMP incurs a fraction of the privacy loss of LCM. If the noisy difference is neither too small or too large, it runs LCM (Line~10), and incurs an additional privacy loss of $\epsilon_{LCM}$.

\begin{theorem}\label{theorem:poking_lcm}
Given a \lcc query $q_{\phi,>c}$,
for any table $D\in \tabledomain$, LCM with Poking
(Algorithm~\ref{algo:poking_lcm}),
denoted by $LCMP_{q_{\phi,>c}}^{\alpha,\beta}$
achieves $(\alpha,\beta)$-\lcc accuracy,
and satisfies $\epsilon$-differential privacy cost,
where $\epsilon=\frac{(1+f)\ln(1/\beta)}{\alpha}$.
\end{theorem}
\eat{
\begin{proof}
The probability to fail the accuracy requirement is
(i) when $q_{\phi}(D)<c-\alpha$,
\begin{eqnarray}
&& \Pr[LCMP^{\alpha,\beta}_{q_{\phi,>c}}(D)=\text{True} ~|~ q_{\phi}(D)<c-\alpha] \nonumber \\
&=& \Pr[q_{\phi}(D)-c +\eta -\alpha_{0} + \alpha >0 ~|~q_{\phi}(D)-c +\alpha <0]\nonumber \\
&& + \Pr[LCM^{\alpha,\beta/2}_{q_{\phi,>c}}(D)=\text{True} ~|~ q_{\phi}(D)<c-\alpha] \nonumber \\
&<& \Pr[\eta >\alpha_{0}] + \beta/2 = e^{-\alpha_{0}\epsilon_0}/2 +\beta/2 = \beta.
\end{eqnarray}
and (ii) when $q_{\phi}(D)<c+\alpha$,
\begin{eqnarray}
&& \Pr[LCMP^{\alpha,\beta}_{\phi,>c}(D)=\text{False} ~|~ q_{\phi}(D) >c+\alpha] \nonumber \\
&=& \Pr[q_{\phi}(D)-c +\eta +\alpha_{0} - \alpha<0 ~|~ q_{\phi}(D)-c -\alpha >0] \nonumber \\
&& + \Pr[LCM^{\alpha,\beta/2}_{q_{\phi,>c}}(D)=\text{False} ~|~ q_{\phi}(D)>c+\alpha] \nonumber \\
&<& \Pr[\eta_{pp} <-\alpha_{0}]  + \beta/2 = e^{-\alpha_{0}\epsilon_0}/2 + \beta/2 =\beta
\end{eqnarray}
Hence, LCMP achieves $\alpha$-\lcc with probability $1-\beta$.
The poking part spends a privacy budget of $\epsilon_{0}$
and the second part using LCM spends a budget of $\epsilon_{LCM}$ Theorem~\ref{theorem:pc_lcm}.
By composition of differential privacy (Theorem~\ref{theorem:seq}),
LCMP satisfies $(\epsilon_0+\epsilon_{LCM})$-differential privacy.
\end{proof}
}
%The proof can be found in
%\ifpaper
%the full paper \cite{fullpaper}.
%\else
%Appendix~\ref{app:lccproof}.
%\fi
Note that LCMP has a higher privacy loss than LCM in the worst case. However, if LCMP returns in either Line~6 or Line~8, then the privacy loss is much smaller, and it occurs often in our experiments. The privacy engine (Algorithm~\ref{algo:pe}) would use the worst case privacy loss to decide whether to answer a query using LCMP (in {\tt estimateLoss}), but use the actual loss (which could be much smaller) to compute the overall loss (in {\tt analyzeLoss}) if LCMP has been run.

We need to ensure this step is also differentially private
and propose two approaches next.
\begin{algorithm}[t]
\caption{LCMP($q_{\phi,>c},\alpha,\beta, D$)}
\label{algo:poking_lcm}
{\small \begin{algorithmic}[1]
%\Require \lcc  $q_{\phi,>c}$, $(\alpha,\beta)$-\lcc tolerance, table $D$, poking fraction $f$
%\Ensure  Answer $a\in \{\text{True, False}\}$
\State Initialize $q_{\phi,>c},\alpha,\beta, D,f$
\Function{\runFunc}{$q_{\phi,>c},\alpha,\beta, D$}
\State Privacy cost for LCM $\epsilon_{LCM} = \frac{\ln(1/(2\beta))}{\alpha}$ %\Comment{Theorem~\ref{theorem:pc_lcm}}
\State Prepaid privacy cost $\epsilon_{0} = \epsilon_{LCM}\cdot f$%, where $f=0.05$
\State Add noise $\tilde{x} = q_{\phi}(D)-c + \eta$,where $\eta \sim Lap(1/\epsilon_{0})$
\State Set $\alpha_{0} = \frac{\ln (1/\beta)}{\epsilon_0}$
\If {$\tilde{x} - \alpha_{0} +\alpha \geq 0$}
   \State \textbf{return} ($\phi,\epsilon_0$)
\ElsIf {$\tilde{x} + \alpha_{0} -\alpha \leq 0$}
   \State \textbf{return} ($\emptyset,\epsilon_0$)
\Else
   \State \textbf{return} ($LCM(q_{\phi,>c},\alpha,\beta/2,D),(1+f)$)
\EndIf
\EndFunction
\Function{\transFunc}{$q_{\phi,>c},\alpha,\beta, D$}
    \State \textbf{return} $\epsilon=(1+f)\ln(1/\beta)/\alpha$
\EndFunction
\end{algorithmic}}
\end{algorithm}

}

\subsection{Special Translation for \tcq}

This section provides a translation mechanism, known as\emph{ Laplace top-$k$ Mechanism} (shown in Algorithm~\ref{algo:tcq-ltm}). This mechanism is a generalized \emph{report-noisy-max} algorithm~\cite{Dwork:2014:AFD:2693052.2693053}: when $k = 1$, it adds noise drawn from $Lap(1/\epsilon)$ to all queries, and only reports the query number that has the maximum noisy count (and not the noisy count).  When $k \geq 1$,  this mechanism first perturbs $\workload\data$  with Laplace noise $\eta\sim Lap(b)^L$, where $b=k/\epsilon$. %$b = \frac{\alpha}{2\ln(L/(2\beta))}$.
These predicates are then sorted based on their corresponding noisy counts in descending order, and the first $k$ boolean formulae are outputted. The privacy cost is summarized in Theorem~\ref{theorem:pc_ltm} and the proof follows that of the report-noisy-max algorithm as shown in Appendix~\ref{app:ltm}.

\eat{
This section provides a translation mechanism, known as
{\em Laplace top-$k$ Mechanism} (shown in Algorithm~\ref{algo:tcq-ltm}).
Unlike the baseline LM (Algorithm~\ref{algo:lm}), the privacy cost of this mechanism does not depend on
$\|W\|_1$, but only the workload size $L$ and $k$.
Algorithm~\ref{algo:tcq-ltm} first perturbs
$\workload\data$ with Laplace noise $\eta\sim Lap(b)^L$,
where $b = \frac{\alpha}{2\ln(L/(2\beta))}$.
These predicates are then sorted
based on their corresponding noisy counts in descending order
and the first $k$ boolean formulae are outputted.
}

\begin{theorem}\label{theorem:pc_ltm}
Given a top-k counting query $q_{W,k}(\cdot)$, where $W=\{\phi_1,\ldots,\phi_L\}$,
for a table $D\in \tabledomain$, Laplace top-$k$ mechanism (Algorithm~\ref{algo:tcq-ltm})
denoted by $LTM^{\alpha,\beta}_{W,k}(\cdot)$,
can achieve $(\alpha,\beta)$-\tcq accuracy by executing $\runFunc(q_{W,k},\alpha,\beta,D)$
with minimal differential privacy cost of $\transFunc(q_{W,k},\alpha,\beta).\upperbound$.
\end{theorem}

\eat{
\begin{proof}(sketch)
The proof for accuracy is similar to the proof for Laplace mechanism,
which requires the noise parameter $b\leq \frac{\alpha}{2(\ln (L/(2\beta))}$.
The proof for the privacy guarantee is different from
the privacy proof for Laplace mechanism (which uses post-processing property of differential privacy).
Refer to Appendix for detailed proof.
\end{proof}
}

\change{

Note that the privacy proof of the report-noisy-max algorithm does not work for releasing both the noisy count and the query number simultaneously. Hence we consider only releasing the bin identifiers in Algorithm~\ref{algo:tcq-ltm}. Moreover, the privacy cost of Algorithm~\ref{algo:tcq-ltm} is independent of the workload $\|W\|_1$.
On the other hand, the baseline LM (Algorithm~\ref{algo:lm}) for answering \tcq queries is different from Algorithm~\ref{algo:tcq-ltm}. Algorithm~\ref{algo:lm} uses noise drawn from $Lap(\|W \|_1/\epsilon)$ to release noisy counts for all queries, and then picks the top-$k$ as a post-processing step. Algorithm~\ref{algo:lm} allows the noisy counts to be shown to the analyst without hurting the privacy cost. Hence, Algorithm~\ref{algo:lm} has a simpler privacy proof than Algorithm~\ref{algo:tcq-ltm} and a privacy loss that depends on the workload. \system supports both Algorithm~\ref{algo:lm} and Algorithm~\ref{algo:tcq-ltm} as there is no clear winner between them when $k > 1$. \system chooses the one with the least epsilon for a given accuracy bound.

}

\eat{
The key differences between this mechanism and the baseline Laplace mechanism in Algorithm~\ref{algo:lm} are:
(a) the privacy cost of this mechanism is independent of the workload $\|\workload\|_1$,
and (b) this mechanism outputs only the set of predicates, not the noisy counts, while
Laplace mechanism allows the counts to be shown to the analyst without hurting the privacy cost.

We propose two distinct algorithms (Algorithm~\ref{algo:lm} and Algorithm~\ref{algo:tcq-ltm}) for answering TCQ queries. Algorithm~\ref{algo:lm} uses noise drawn from $Lap(||W||_1/\epsilon)$ to release noisy counts for all queries, and then picks the top-$k$ as a post-processing step. This is not the same as the \emph{report-noisy-max} algorithm~\cite{Dwork:2014:AFD:2693052.2693053} when $k=1$. Hence, Algorithm~\ref{algo:lm} has a simpler privacy proof than the report-noisy-max algorithm. On the other hand, Algorithm~\ref{algo:tcq-ltm} is a generalized report-noisy-max algorithm: when $k=1$, it adds noise drawn from $Lap(1/\epsilon)$ to all queries, and only reports the query number that has the maximum noisy count (and not the noisy count). The privacy analysis of Algorithm ~\ref{algo:tcq-ltm} in Appendix~\ref{app:ltm}  follows that of the report-noisy-max algorithm. \system supports both Algorithms~\ref{algo:lm} and ~\ref{algo:tcq-ltm} as there is no clear winner between them when $k>1$. \system chooses the one with the least epsilon for a given accuracy bound.

In addition, the privacy proof of report-noisy-max algorithm does not work for releasing both the noisy count and the query number simultaneously~\cite{DBLP:journals/corr/abs-1710-09951}. Hence we consider only releasing the bin identifiers in Algorithm~\ref{algo:tcq-ltm}. 

\begin{proof} (sketch)
Set the noise parameter $b \leq \frac{\alpha}{2(\ln L + \ln(k/\beta))}$, we can show
that (i) if $c_{\phi}(D)<c_k-\alpha$, where $c_k$ is the answer to the $k$th largest linear counting query,
$\phi$ will be included in the output $a$ with small probability:
\begin{eqnarray}
&&\Pr[\phi \in a ~|~ c_{\phi}(D) < c_k-\alpha] < L e^{-\frac{\alpha}{2b}} = \beta/k,
\end{eqnarray}
and hence, the probability that the output $a$ has at least  boolean formula
with count smaller than $c_k-\alpha$ is $\beta$ for $|a|=k$;
(ii) if $c_{\phi}(D)>c_k+\alpha$, then the probability of missing it in the output is small:
\begin{eqnarray}
&&\Pr[\phi \notin a ~|~ c_{\phi}(D) > c_k+\alpha] < k e^{-\frac{\alpha}{2b}} = \beta/L,
\end{eqnarray}
and hence, the probability of missing at least one boolean formula
with count greater than $c_k +\alpha$ is at most $\beta$.
We can also show this mechanism satisfies $k/b$-differential privacy
and hence consider the largest possible value for $b$ to minimize the privacy cost.
\end{proof}

Alternatively, the data analyst can pose a workload counting query
$q_{W}$ with $(\alpha,\beta)$-\wcq tolerance,
and then answer the top-$k$ counting query $q_{W,k}$ locally
However, this approach does not achieve $(\alpha,\beta)$-\tcq tolerance, i.e.,
%due to smaller expected noise per query (smaller $b$). More specifically,
\begin{lemma}
Using the outputs of $LM^{\alpha,\beta}_{q_W}$  to answer
$q_{W,k}$ achieves $(\alpha',\beta)$-\tcq tolerance,
where $\alpha' = (1+\frac{2\ln(Lk)}{\ln(1/(1-(1-\beta)^{1/L}))})\alpha>\alpha$.
\end{lemma}
}

\begin{algorithm}[t]
\caption{\tcq-LTM($q_{W,k}, \alpha,\beta, D)$)}\label{algo:tcq-ltm}
{\small \begin{algorithmic}[1]
\State Initialize $\workload \gets \transform(W), \data \gets \transform_W(D), \alpha,\beta$
\Function{\runFunc}{$q_{W,k},\alpha,\beta,D$}
    \State $\epsilon \gets \transFunc(q_{W,k},\alpha,\beta).\upperbound$
    \State $(\tilde{x}_1,\ldots,\tilde{x}_L) = \workload\data + Lap(b)^L$, where $b=k/\epsilon$
    \State $(i_1,\ldots,i_k) = {\tt argmax}^k_{i=1,\ldots,L} \tilde{x}_i$
    \State \textbf{return} $(\{\phi_{i_1},\ldots,\phi_{i_k}\},\epsilon)$
\EndFunction
\Function{\transFunc}{$q_{W,k},\alpha,\beta$}
    \State \textbf{return} $\upperbound = \frac{2k\ln (L/(2\beta))}{\alpha}, \lowerbound = \upperbound$
\EndFunction
\end{algorithmic}}
\end{algorithm}

\eat{
\subsection{Inference Techniques}
\todo{Using private multiplicative weights or SVT}
Laplace noise is used as an example for translation
which achieves $\epsilon$-differential privacy per query in this section.
Other type of noises are also possible, such as Guassian noise
which provides $(\epsilon,\delta)$-differential privacy per query,
and can be adapted into this system.
}

%!TEX root=./main.tex
\section{Privacy Analysis}
\label{sec:analyzer}
The privacy analyzer ensures that every sequence of queries answered by \system results in a $B$-differentially private execution, where $B$ is the data owner specified privacy budget. The formal proof of privacy primarily follows from the well known composition theorems (described for completeness in Appendix~\ref{sec:seq}). According to sequential composition (Theorem~\ref{theorem:seq}), the privacy loss of a set of differentially private mechanisms (that use independent random coins) is the sum of the privacy losses of each of these mechanisms. Moreover, postprocessing the outputs of a differentially private algorithm does not degrade privacy (Theorem~\ref{theorem:post}).

The main tricky (and novel) part of the privacy proof (described in Section~\ref{sec:proof}) arises due to the fact that (1) the $\epsilon$ parameter for a mechanism is chosen based on the analyst's query and accuracy requirement, which in turn are adaptively chosen by the analyst based on previous queries and answers, and (2) some mechanisms may have an actual privacy loss that is dependent on the data.  \system accounts for privacy based on the actual privacy loss (and not the worst case privacy loss (see Line~\ref{analyzeloss}, Algorithm~\ref{algo:pe}).

\renewcommand{\trans}{\ensuremath{\mathbb{T}}}
\newcommand{\epsmax}{\ensuremath{\epsilon^{\max}}}
\subsection{Overall Privacy Guarantee}\label{sec:proof}
We show the following guarantee: any sequence of interactions between the data analyst and \system satisfies $B$-differential privacy, where $B$ is the privacy budget specified by the data owner. In order to state this guarantee formally, we first need the notion of a \emph{transcript of interaction} between \system and the data analyst.

We define the transcript of interaction $\trans$ as an alternating sequence of queries (with accuracy requirements) posed to \system and answers returned by \system. $\trans$ encodes the analyst's view of the private database. More formally,
\squishlist
	\item The transcript $\trans_i$ after $i$ interactions is a sequence $[(q_1, \alpha_1, \beta_1), (\omega_1, \epsilon_1), \ldots, (q_i, \alpha_i, \beta_i), (\omega_i, \epsilon_i)]$, where $(q_i, \alpha_i, \beta_i)$ are queries with accuracy requirements, and $\omega_i$ is the answer returned by \system and $\epsilon_i$ the actual privacy loss.% associated with that answer.
	\item Given $\trans_{i-1}$, analyst chooses the next query $(q_{i+1}, \alpha_{i+1}, \beta_{i+1})$ adaptively. We model this using a (possibly randomized) algorithm $\mathbb{C}$ that maps a transcript $\trans_{i-1}$ to $(q_i, \alpha_i, \beta_i)$; i.e., $\mathbb{C}(\trans_{i-1}) = (q_i, \alpha_i, \beta_i)$. Note that the analyst's algorithm $\mathbb{C}$ does not access the private database $D$.
    \item Given $(q_i, \alpha_i, \beta_i)$, \system select a subset of mechanisms $\mathcal{M}^*$ such that $\forall M \in \mathcal{M}^*$,  $M.\transFunc(q_i,\alpha_i,\beta_i).\upperbound \leq B-B_i$. Furthermore, if $\mathcal{M}^*$ is not empty, \system chooses one mechanism $M_i \in \mathcal{M}^*$ deterministically (either based on $\lowerbound$ or $\upperbound$) to run. The selection of $M_i$ is deterministic and independent of $D$.% the private database.  
    \item  If \system find no mechanism to run $(\mathcal{M}^* = \emptyset)$, then the query is \emph{declined} by \system. In this case, $\omega_i= \bot$ and $\epsilon_i = 0$. 
    \item If the \system chosen algorithm $M_i$ is LM, WCQ-SM, ICQ-SM or TCQ-LTM, $\epsilon_i = \upperbound_i$, where $\epsilon_i$ is the upperbound on the privacy loss returned by $M_i.\transFunc$. For ICQ-MPM, the actual privacy loss can be smaller; i.e., $\epsilon_i \leq \upperbound_i$.
	\item Let $Pr[\trans_i| D]$ denote the probability that the transcript of interaction is $\trans_i$ given input database $D$. The probability is over the randomness in the analyst's choices $\mathbb{C}$ and the randomness in the mechanisms $M_1, \ldots, M_i$ executed by \system.
\squishend

Not all transcripts of interactions are realizable under \system. Given a privacy budget $B$, the set of valid transcripts is defined as:
\begin{definition}[Valid Transcripts]
	A transcript of interaction $\trans_i$ is a \emph{valid} \system transcript generated by Algorithm~\ref{algo:pe} if given a privacy budget $B$ the following conditions hold:
	\begin{itemize}
		\item $B_{i-1} = \sum_{j = 1}^{i-1} \epsilon_j \leq B$, and
		\item Either $\omega_i = \bot$, or $B_{i-1} + \upperbound_i \leq B$.
	\end{itemize}
\end{definition}

We are now ready to state the privacy guarantee:
\begin{theorem}[\system Privacy Guarantee]\label{thm:privacy}
	Given a privacy budget $B$, any valid \system transcript $\trans_i$, and any pair of databases $D$, $D'$ that differ in one row (i.e., $|D \setminus D' \cup D' \setminus D| = 1$), we have:
	\begin{enumerate}
		\item $B_i = \sum_{j = 1}^{i} \epsilon_i \leq B$, and
		\item $Pr[\trans_i | D] \leq e^{B_i} Pr[\trans_i | D']$.
	\end{enumerate}
\end{theorem}
%Please refer to Appendix~\ref{app:apex:proof} for the proof.
\begin{proof}
(1) directly follows from the definition of a valid transcript and these are the only kinds of transcripts an analyst sees when interacting with \system.

(2) can be shown as follows using induction.

\noindent{\underline{Base Case:}} When the transcript is empty, $Pr[\emptyset | D] \leq e^0 Pr[\emptyset | D']$.

\noindent{\underline{Induction step:}}
Now suppose for all $\trans_{i-1}$ of that encode valid \system transcripts of length $i-1$, $Pr[\trans_{i-1}| D] \leq e^{B_{i-1}} Pr[\trans_{i-1} | D']$. Let $\trans_i = \trans_{i-1} || [(q_i, \alpha_i, \beta_i), (\omega_i, \epsilon_i)]$ be a valid \system transcript of length $i$. Then:
\begin{eqnarray*}
\lefteqn{Pr[\trans_i | D] = Pr[\trans_{i-1} | D] Pr[[(q_i, \alpha_i, \beta_i), (\omega_i, \epsilon_i)] | D, \trans_{i-1}]}\\
& = & Pr[\trans_{i-1} | D] Pr[\mathbb{C}(\trans_{i-1}) = (q_i, \alpha_i, \beta_i)] Pr[M_i(D) = (\omega_i, \epsilon_i)]
\end{eqnarray*}
Note that the analyst's choice of query $q_i$ and accuracy requirement depends only on the transcript $\trans_{i-1}$ and not the sensitive database, and thus incurs no privacy loss (from Theorem~\ref{theorem:post}). Thus, it is enough to show that
\[Pr[M_i(D) = (\omega_i, \epsilon_i)] \leq e^{\epsilon_i} Pr[M_i(D') = (\omega_i, \epsilon_i)]\]
\noindent{\underline{Case 1:}}
When $\omega_i\neq \bot$ and $M_i$ is LM, WCQ-SM, ICQ-SM, or TCQ-LTM, the mechanism satisfies $\upperbound_i$-DP and $\epsilon_i = \upperbound_i$. Therefore, $Pr[M_i(D) = (\omega_i, \epsilon_i)] \leq e^{\epsilon_i} Pr[M_i(D') = (\omega_i, \epsilon_i)]$.

\noindent{\underline{Case 2:}}
When $\omega_i \neq \bot$ and $M_i$ is ICQ-MPM, the mechanism satisfies $\upperbound_i$-DP across all outputs. However, when either mechanism outputs $(\omega_i, \epsilon_i)$, for $\epsilon_i < \upperbound_i$, we can show that $Pr[M_i(D) = (\omega_i, \epsilon_i)] \leq e^{\epsilon_i} Pr[M_i(D') = (\omega_i, \epsilon_i)]$. In the case of ICQ-MPM, if the algorithm returns in Line~11 after $i$ iterations of the loop, the noisy answer is generated by a DP algorithm with privacy loss $\epsilon_i = \frac{j}{m}\upperbound_i$.

\noindent{\underline{Case 3:}}
Finally, when $\omega_i = \bot$  (i.e., the query is declined), the decision to decline depends on $\upperbound_i$ of all mechanism applicable to the query (which is independent of the data) rather than $\epsilon_i$ (which could depend on the data in the case of ICQ-MPM). Therefore, $Pr[M_i(D) = (\omega_i, \epsilon_i)]  = Pr[M_i(D') = (\omega_i, \epsilon_i)]$ for all $D, D'$. The proof would fail if the decision to deny a query depends on $\epsilon_i$.
\end{proof}

%\input{rdp}

%!TEX root=./main.tex
\section{Query Benchmark Evaluation}\label{sec:evaluation}

\begin{table*}[h]
\center
{\small
	\begin{tabular}{|c|c|l|l|} \hline
   	{\bf Name} & {\bf $D$} & {\bf Query workload $W$} & {\bf Query output}  \\ \hline
	QW1 & \census & "capital gain"$\in[0,50)$, "capital gain"$\in[50,100)$, ...,"capital gain"$\in[4950,5000)$  &   bin counts \\ \hline
	QW2 & \census & "capital gain"$\in[0,50)$, "capital gain"$\in[0,100)$, ..., "capital gain"$\in[0,5000)$ & bin counts  \\ \hline
	QW3 & \location & "trip distance"$\in [0, 0.1)$, "capital gain"$\in[0,50)$, ...,"capital gain"$\in[0,50)$ & bin counts\\ \hline
	QW4 & \location & ($0\le$"total amount"$<1 \wedge$ "passenger"$=1$),.., ($9\le$"total amount"$<10 \wedge$ "passenger"$=10$) & bin counts\\ \hline
	QI1 & \census & "capital gain"$<50$, "capital gain"$<100$,..., "capital gain"$<5000$ & bin ids having counts $>0.1|D|$\\ \hline
	QI2 & \census & ($0\le$"capital gain"$<100$, "sex"='M'),...($4500\le$"capital gain"$<5000$, "sex"='F') & bin ids having counts $>0.1|D|$ \\ \hline
	QI3 & \location &"fare amount"$\in [0,0,1)$, "fare amount"$\in [0.1,2)$,..., "fare amount"$\in [9.9,10)$ & bin ids having counts $>0.1|D|$ \\ \hline
	QI4 & \location & "total amount"$\in [0,0,1)$, "total amount"$\in [0.1,2)$,..., "total amount"$\in [9.9,10)$ & bin ids having counts $>0.1|D|$\\ \hline
	QT1 & \census & "age"$=0$,"age"$=1$,...,"age"$=99$  & top $10$ bins with highest counts\\ \hline
	QT2 & \census &  100 predicates on different attributes, e.g. "age"$=1$, "workclass"="private",...&  top $10$ bins with highest counts \\ \hline
	QT3 & \location & ("PUID"=1 $\wedge$"DOID"=1), ("PUID"=1 $\wedge$ "DOID"=2),..,("PUID"=10 $\wedge$ "DOID"=10) & top $10$ bins with highest counts \\ \hline
	QT4 & \location & 100 predicates on different attributes, e.g. "pickup date"$=1$, "passenger count"$=1$,... & top $10$ bins with highest counts\\ \hline
   \end{tabular}
   \vspace{.75em}
     \caption{Query benchmarks includes 3 types of exploration queries on 2 datasets.}
  \label{tab:queries}
}
\end{table*}

In this section, we evaluate \system on real datasets using a set of benchmark queries. We show:
\squishlist
\item \system is able to effectively translate queries associated with accuracy bounds into differentially private mechanisms. These mechanisms accurately answer a wide variety of interesting data exploration queries with moderate to low privacy loss.
\item The set of query benchmarks show that no single mechanism can dominate the rest and \system picks the mechanism with the least privacy loss for all the queries.
\squishend

\subsection{Setup}

\stitle{Datasets.}
Our experiments use two real world datasets. The first data set \census was extracted from 1994 US Census release~\cite{Dua:2017}. This dataset includes $15$ attributes ($6$ continuous and $9$ categorical), such as "capital gain", "country", and a binary "label" indicating whether an individual earns more than $5000$ or not, for a total of $32,561$ individuals. The second dataset, refereed as \location, includes $9,710,124$ NYC's yellow taxi trip records~\cite{data_location}. Each record consists of $17$ attributes, such as categorical attributes (e.g., "pick-up-location"), and continuous attributes (e.g., "trip distance").

\stitle{Query Benchmarks.}
We design $12$ meaningful exploration queries on \census and \location datasets, summarized in Table~\ref{tab:queries}. These $12$ queries cover the three types of exploration queries defined in Section~\ref{sec:explorequeries}, QW1-4, QI1-4, and QT1-4 corresponds to \wcq, \icq, and \tcq respectively. Queries with number $1$ and $2$ are for \census, and with number $3$ and $4$ are for \location. The predicate workload $W$ cover 1D histogram, 1D prefix, 2D histogram and count over multiple dimensions. We set $\beta=0.0005$ and vary $\alpha \in \{0.02, 0.04, 0.08, 0.16, 0.32, 0.64\}$.

\stitle{Metrics.}
For each query ($q,\alpha,\beta$), \system outputs ($\epsilon, \omega$) after running a differentially private mechanism, where $\epsilon$ is the actual privacy loss and $\omega$ is the noisy answer. The empirical error of a WCQ $q_W(D)$ is measured as $\|\omega - q_W(D)\|_{\infty}/|D|$, the scaled maximum error of the counts. The empirical errors of  ICQ $q_{W,>c}(D)$ and TCQ $q_{W,k}(D)$ are measured as $\|\alpha\|_{\infty}/|D|$, the scaled maximum distance of mislabeled predicates.

\stitle{Implementation Details.}
\system is implemented using python-$3.4$, and is run on a machine with $64$ cores and $256$ GB memory. We run  \system with optimistic mode. For strategy mechanism, we choose $H_2$ strategy (a hierarchical set of counts~\cite{Li:2010:OLC:1807085.1807104, Li:2015:MMO:2846574.2846647, hdmm18}) for all queries.

\subsection{\system End-to-End Study}

\begin{figure*}[h]
	\includegraphics[width=\textwidth]{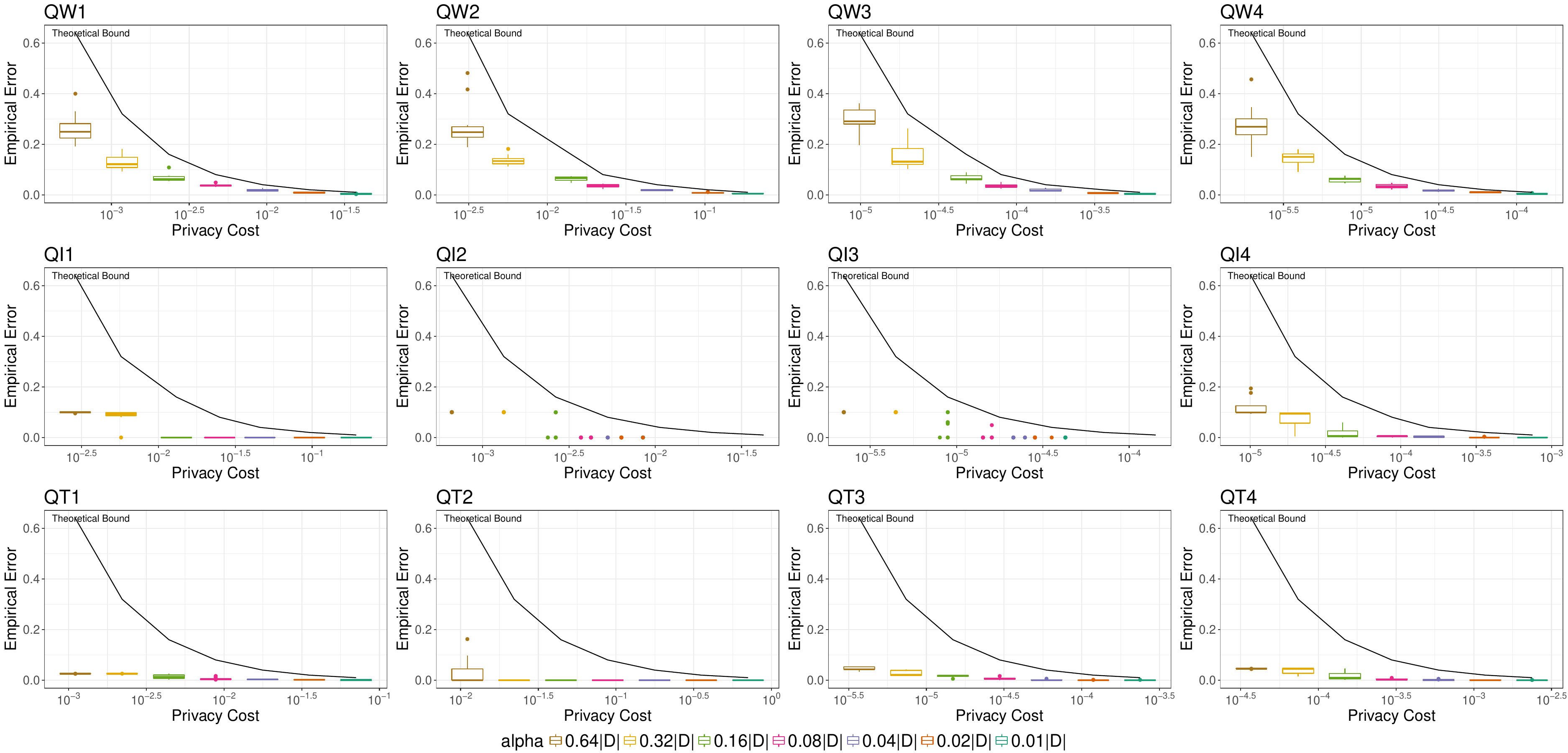}%
	%\caption{Privacy cost and empirical accuracy using optimal mechanism chosen by \system (Optimistic Mode) on the $12$ queries at default paramete setting with $\alpha \in \{0.01, 0.02, 0.04, 0.08, 0.16, 0.32, 0.64\}|D|$ and $\beta = 5 \times 10^{-4}$. On \census data, all queries can be answered with empirical error less than $0.1$ with privacy budget less than $0.1$; on \location data, all queries can be answered with empirical error less than $0.1$ with privacy budget less than $0.001$. When accuracy requirement relaxes (i.e., $\alpha$ increases), the privacy cost decreaess and the empirical accuracy decreases for all queries.}\label{fig:cost_accuracy_all}
\caption{Privacy cost and empirical accuracy using optimal mechanism chosen by \system (Optimistic Mode) on the $12$ queries at default parameter setting with $\alpha \in \{0.01, 0.02, 0.04, 0.08, 0.16, 0.32, 0.64\}|D|$ and $\beta = 5 \times 10^{-4}$. On \census data, all queries can be answered with empirical error $<0.1$ with privacy budget $<0.1$; on \location data, all queries can be answered with empirical error $<0.1$ with privacy budget  $<0.001$. When accuracy requirement relaxes (i.e., $\alpha$ increases), the privacy cost decreaess and the empirical accuracy decreases for all queries.}\label{fig:cost_accuracy_all}
\end{figure*}

We run \system for the 12 queries shown in Table~\ref{tab:queries} with different accuracy requirements from $0.01|D|$ to $0.64|D|$ and $\beta=0.0005$.
We show in Figure~\ref{fig:cost_accuracy_all} that a line connects points $(\alpha,\upperbound)$ where $\upperbound$ is the upper bound on the privacy loss for the mechanism chosen by \system for the given $\alpha$. For all the queries except QI2 and QI3, the mechanism chosen for each $\alpha$ incurs an actual privacy cost at $\epsilon = \upperbound$ and the only variation in the empirical error, so the corresponding $(\hat{\alpha}/|D|, \epsilon)$ of 10 runs is shown as boxplots. For QI2 and QI3, both the empirical error and the actual privacy cost $(\hat{\alpha}/|D|, \epsilon)$ vary across runs and hence are plotted as points in Figure~\ref{fig:cost_accuracy_all}.

The empirical error $\alpha$ is always bounded by the theoretical $\alpha$ for all the queries. The gap between the theoretical line and the actual boxplots/points are: (1) the analysis of the error is not tight due to the use of union bound; (2) for mechanism with data dependent translation (QI2 and QI3), the actual privacy cost is far from the upperbound $\upperbound$ resulting a left shift of the points from the theoretical line. The privacy cost for QW1 and QW2 for \census dataset  is in the range of $(10^{-4},10^{-1})$ for all $\alpha$ values. This privacy cost is 2-3 orders larger than the privacy cost for QW3 and QW4 on the \location dataset, because given the same ratio $\alpha/|D|$, the queries on {\tt NYTaxi} has a larger $\alpha$ than {\tt Adult} because of data size.

\ifpaper
\conf{
% put the F1 part in appendix
Additional experiments using F1 score as the quality measurement are included in Appendix~\ref{sec:f1}.
}\else
\full{
To further understand the relation between our accuracy requirement with commonly used error metric, we use F1 score to measure the similarity between the correct answer set and noisy answer set outputted from the mechanism. Figure~\ref{fig:f1-qi4-qt1} presents the F1 score of two queries: QI4 of an \icq, and QT1 of a \tcq. In QT1, when $\alpha$ from $0.02|D|$ to $0.64|D|$, the median F1 score decreases from $0.9$ to $0.15$, which has a steeper gradient than the changes in Figure~\ref{fig:cost_accuracy_all} and a closer trend with the theoretical $\alpha$. In QI4, the F1 score is even more consistent with $\alpha$ and the empirical error shown in Figure~\ref{fig:cost_accuracy_all}. This shows  that our $(\alpha,\beta)$ accuracy  requirement is still a good indicator in data exploration process.

\begin{figure}[h]
	\centering
        \includegraphics[width=0.45\textwidth]{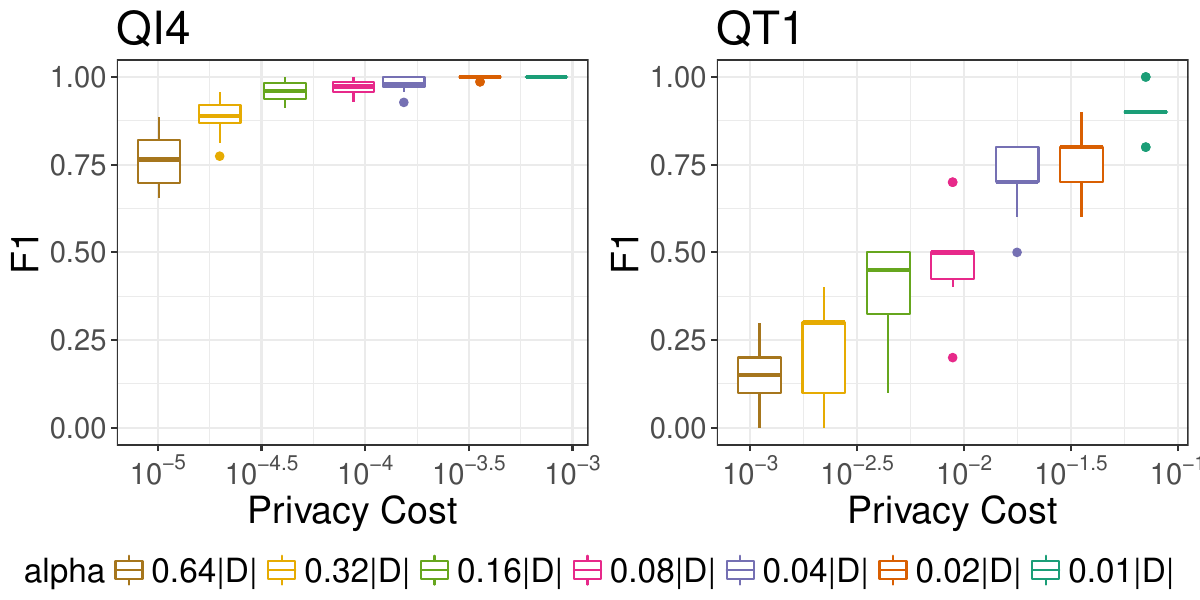}%
	\caption{Privacy cost and empirical accuracy ( F1 score) using optimal mechanism chosen by APEx (Optimistic Mode) on QI4 and QT1.}\label{fig:f1-qi4-qt1}
\end{figure}

}\fi

\eat{
All queries were answered by \system with empirical errors below the theoretical accuracy requirement $\alpha$ (indicated by the black curve in Figure~\ref{fig:cost_accuracy_all}).
In this section, we study whether \system answers exploration queries with high accuracy and reasonably small privacy cost. Consider the 12 query instances from Table~\ref{tab:queries}, with different accuracy requirements $\{0.02,0.04,0.08,0.16, 0.32, 0.64\}|D|$. For each query $(q,\alpha,\beta)$, \system optimistically chooses a differentially private mechanism that achieves this accuracy requirement with the least estimate privacy cost. This step is independent of the data, so \system picks the same mechanism given the same query. However, there are two modes for \system: pessimistic mode and optimistic mode. This experiment shows the choices of \system for all queries under optimistic mode, which allows data-dependent translation to be chosen as well. The chosen mechanism $M$ runs $q$ and returns a noisy answer with an actual privacy cost $\epsilon$ which lie within $(\lowerbound,\upperbound)$ of $M.\transFunc(q,\alpha,\beta)$. We report the empirical error of these noisy answers and the actual privacy cost of $10$ runs of \system for each $(q,\alpha,\beta)$ in Figure~\ref{fig:cost_accuracy_all}. In Figure~\ref{fig:cost_accuracy_all}, we present the median and quartiles using a box plot.

First, all queries were answered by \system with empirical errors below the theoretical accuracy requirement $\alpha$ (indicated by the black curve in Figure~\ref{fig:cost_accuracy_all}). For \wcq, the empirical error (the maximum error of noisy counts) is directly affected by the Laplace noise, and hence directly correlates with the theoretical accuracy $\alpha$. When $\alpha$ decreases from $0.64|D|$ to $0.1|D|$, the empirical error drops accordingly. \system chooses the baseline LM for QW1 and QW3 and chooses \icq-SM for QW2 and QW4. These mechanisms also have a data-independent privacy cost ($\upperbound=\lowerbound$) and thus the actual privacy cost of the 10 runs at each $\alpha$ are all the same. The privacy cost for QW1 and QW2 for \census dataset  is in the range of $(10^{-4},10^{-1})$ for all $alpha$ values. This is 2-3 orders larger than the privacy cost for QW3 and QW4 for \location, because given the same ratio $\alpha/|D|$, the queries on {\tt NYTaxi} has a larger $\alpha$ than {\tt Adult} dataset because of data size.

For \icq, though the empirical errors for all accuracy requirements are still bounded by $\alpha/|D|$, they are much smaller in value and their changes with respect to $\alpha$ is less smooth. This is because the empirical error for \icq measures the  maximal distance from the count of a mislabeled predicate to  the counting threshold $c$, which is related to data, but indirectly to the noise. For example, for QI1, the empirical errors at $\alpha=0.64|D|$ and $\alpha=0.32|D|$ are close because the predicates with similar counts are incorrectly mislabeled. When $\alpha$ decreases to $0.16|D|$, the previously mislabeled predicates are correctly labeled because the noise is much smaller than the distance between their counts and the threshold. Hence, we see a drop in empirical cost at $\alpha=0.16|D|$, and the cost stays when $alpha$ continually decreases. QI4 has a smoother gradient for the empirical error than the rest of \icq queries because the \location data is uniformly distributed over the attribute "total amount" than the other queries.

Moreover, QI1 and QI4 are answered with \icq-SM and hence their privacy cost is also fixed, but QI2 and QI3 were run by \icq-MPM which has an actual privacy cost dependent on data and thus the 10 runs of privacy cost are different and are shown as scatter plot rather than box plot in Figure~\ref{fig:cost_accuracy_all}. We will discuss more about this mechanism in Section~\ref{sec:eval_opt}.

For \tcq, the observations are similar to \icq in Figure~\ref{fig:cost_accuracy_all}, as the empirical error depends on the distance between mislabelled predicates and the $k^{th}$ largest counting values.
}

\begin{table*}[h]
  \center
{\small
\begin{tabular}{|l|l|l|l|l|l|l|l|l|}
\hline
{\bf Mechanism} & \multicolumn{8}{c|}{\bf{Query-$\alpha$}}                                                                                                                      \\ \hline
          & QW1-0.02|D|      & QW1-0.08|D|      & QW2-0.02|D|      & QW2-0.08|D|      & QW3-0.02|D|      & QW3-0.08|D|      & QW4-0.02|D|      & QW4-0.08|D|      \\ \hline
WCQ-LM        & \textbf{0.01874} & \textbf{0.00469} & 1.87430          & 0.46858          & 0.00629          & 0.00157          & \textbf{0.00006} & \textbf{0.00002} \\ \hline
WCQ-SM    & 0.09880          & 0.02383          & \textbf{0.10451} & \textbf{0.02251} & \textbf{0.00036} & \textbf{0.00008} & 0.00033          & 0.00009          \\ \hline
          & QI1-0.02|D|      & QI1-0.08|D|      & QI2-0.02|D|      & QI2-0.08|D|      & QI3-0.02|D|      & QI3-0.08|D|      & QI4-0.02|D|      & QI4-0.08|D|      \\ \hline
ICQ-LM        & 1.76786          & 0.44197          & 0.01768          & 0.00442          & 0.00006          & 0.000015          & 0.00593          & 0.00148          \\ \hline
ICQ-SM    & \textbf{0.10271} & \textbf{0.02682} & 0.10506          & 0.02517          & 0.00033          & 0.00008          & \textbf{0.00034} & \textbf{0.00008} \\ \hline
ICQ-MPM   & 0.63644          & 0.31822          & \textbf{0.00636} & \textbf{0.00371} & \textbf{0.00003} & \textbf{0.000014} & 0.00640          & 0.00178          \\ \hline
          & QT1-0.02|D|      & QT1-0.08|D|      & QT2-0.02|D|      & QT2-0.08|D|      & QT3-0.02|D|      & QT3-0.08|D|      & QT4-0.02|D|      & QT4-0.08|D|      \\ \hline
TCQ-LM        & \textbf{0.03536} & \textbf{0.00884} & 266.24590        & 66.56148         & \textbf{0.00012} & \textbf{0.00003} & 1.40857          & 0.35214          \\ \hline
TCQ-LTM   & 0.35358          & 0.08840          & \textbf{0.35358} & \textbf{0.08840} & 0.00119          & 0.00030          & \textbf{0.00119} & \textbf{0.00030} \\ \hline
\end{tabular}
}
\vspace{.75em}
 \caption{Privacy cost using all applicable mechanisms on the $12$ queries (Table~\ref{tab:queries}) at $\alpha=\{0.02, 0.08\}|D|$ and $\beta=5\times10^{-4}$. Median of $10$ runs is reported for data-dependent mechanisms. It shows that (a) no single mechanism can always win or lose, (b) privacy cost of different mechanisms answering the same query, and privacy cost of same mechanism on different queries can be significantly different. Therefore, it is critical to use \system for choosing optimal mechanisms.}\label{tab:cost_mechanism}
\end{table*}

%Figure~\ref{fig:f1-qi4-qt1} shows that as the accuracy parameter $\alpha$ relaxes, F1 score decreases. However, F1 decrease rate can be different from that of empirical accuracy. For example, as explained in QT1, similar counts are incorrectly labeled, hence F1 can drop. However, its empirical error is determined by the relative distance to $k^{th}$ largest count, even a small distance change can considerably affect F1. For example in QT1, when $\alpha$ from $0.02|D|$ to $0.64|D|$, empirical error increases from $0.001$ to $0.025$, but F1 score decreases from $0.9$ to $0.15$, which is much more significant. On the contrary in QI4, F1 is much resistant to empirical error changes.

%A small empirical error does not always mean great F1 score, as shown in QT1. [Explain why so, QT has many small counts, and hence close to c_k and get mislabeled. Hence the empirical distance of these mislabeled predicates  from c_k is still bounded by c_k. Hence, for such skewed dataset, a tighter accuracy requirement $\alpha<<0.01|D|$ is needed.  However, for less skewed counts like QI4, the trend for F1 score is more consistent with the empirical error.

\subsection{Optimal Mechanism Study}\label{sec:eval_opt}
We run all the  applicable mechanisms for the 12 queries from Table~\ref{tab:queries} at $\alpha \in \{0.02|D|, 0.08|D|\}$
and show the median of the actual privacy costs in Table~\ref{tab:cost_mechanism}.
The privacy cost with the least value is in bold for the given query and accuracy.
Indeed, \system picks the mechanisms with these least privacy cost for all the 12 queries.
\system can save more than 90\% of the privacy cost of the baseline translation (LM), such as QW2, QW3, QT2,QT4, and all the \icq.
In particular, the baseline mechanism LM is highly dependent on the sensitivity of the query.
For example, the workload in QW2 (a cumulative histogram query) and QT2 (counts on many attributes) has a high sensitivity,
so the cost of \wcq-LM is $20$ times larger than \wcq-SM for QW2 at $\alpha=0.08|D|$,
and the cost of \tcq-LM is $760X$ more expensive than \tcq-LTM for QT2 at $\alpha=0.02|D$,
where \wcq-SM and \tcq-LTM are the optimal mechanisms chosen by \system for QW2 and QT2 respectively.

%Under pessimistic mode, \system will not choose ICQ-MPM for QI2 at $\alpha=0.02|D|$, as the worst privacy cost for MPM is $0.02121$ which is bigger than the worst privacy cost forICQ-LM ($0.01768$). Under optimistic model, \system chooses ICQ-MPM based on its lower privacy bound $0.00212$, which is smaller than all the other two mechanisms. The optimal mechanism chosen by \system can change query parameters.

\eat{
In this section, rather than running the optimal mechanisms chosen by \system, we analyze the privacy cost of all the mechanisms applicable to a given query.
First, we consider $12$ queries instantiated from Table~\ref{tab:queries} and two accuracy requirements, $\alpha \in \{0.02|D|, 0.08|D|\}$. For each $(q,$alpha$,$beta$)$, we run each mechanism $10$ runs and report the median of their actual privacy cost.
This cost is shown in Table~\ref{tab:cost_mechanism}, grouped by query type \wcq, \icq, and \tcq.
We see that  different mechanisms on the same query cost substantially different.
Some mechanisms are highly dependent on the sensitivity of the query, such as the baseline LM.
For example, the workload in QW2 (a cumulative histogram query) and QT2 (counts on many attributes) has a high sensitivity,
so the cost of \wcq-LM is $20$ times larger than \wcq-SM for QW2 at $\alpha=0.08|D|$,
and the cost of \tcq-LM is $760X$ more expensive than \tcq-LTM at $\alpha=0.02|D$.
Also,  the same mechanism has distinct privacy cost across different queries.
For example, the privacy cost of \icq-MPM on QI1 is $100X$ of that on QI2, even thought QI1 and QI2 have the same workload size.
This difference for \icq-MPM is mainly due to data and the threshold for a query,
which will elaborate later in sensitivity study for query threshold.

\system is able to identify the optimal mechanisms under different modes. Under pessimistic mode, \system chooses the mechanisms with the smallest privacy cost upper bound $\upperbound$. For QI2 at $\alpha=0.02|D|$, the worst privacy cost for MPM is $0.02121$ which is bigger than the worst privacy cost for LM ($0.01768$). Hence under pessisitic mode, MPM will not be chosen. However, its actual privacy cost can be much smaller than its worst privacy cost. Hence, under optimistic model, \system will chooses MPM based on its lower privacy bound $0.00212$, which is smaller than all the other two mechanisms. 
}

The optimal mechanism chosen by \system can change query parameters: workload size $L$, $k$ in \tcq, and threshold $c$ in \icq.
\ifpaper
\conf{
As parameter $L$ and $k$ play a direct role in the privacy cost shown in Section~\ref{sec:translation}, we will focus on the sensitivity of the privacy on $c$ and leave the rest in the full version~\cite{fullpaper}.

\begin{figure}[b]
	\centering
		\includegraphics[width=0.8\columnwidth]{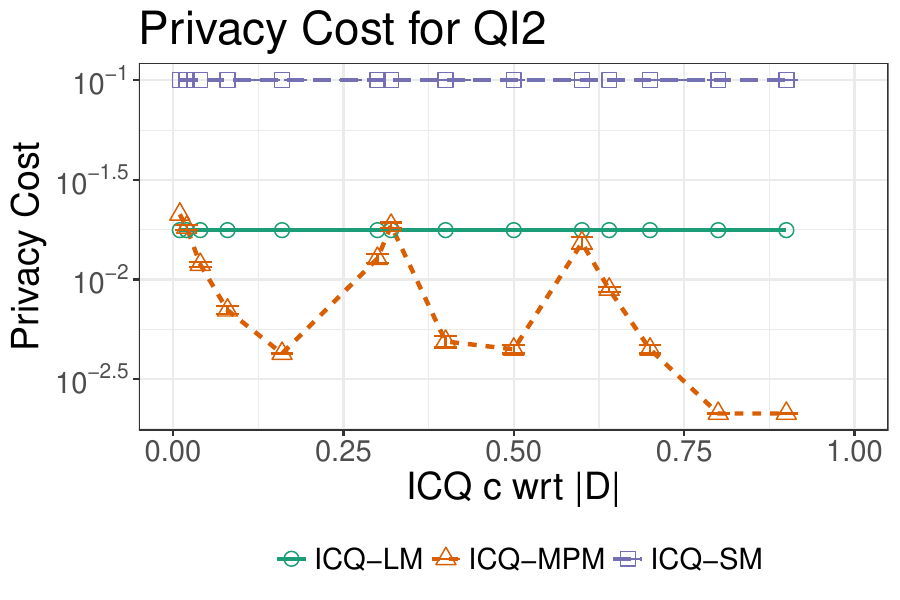}
		\caption{Vary ICQ $c$}
		\label{fig:cost_icq_c}
\end{figure}

\ifpaper
\else
\stitle{Vary counting threshold $c$.}
\fi
Figure~\ref{fig:cost_icq_c} shows the actual privacy cost of mechanisms for QI2 with different thresholds.
All the mechanisms for \icq except \icq-MPM, have a fixed privacy cost which is dependent of the data and the query (including $c$).
However, we observe an interesting trend for the actual privacy cost used by \icq-MPM as $c$ increases.
The smallest privacy cost which takes place after $c=0.8$ is $1/10$ of the upper bound of \icq-MPM.
The privacy cost of \icq-MPM depends on the number of poking times before returning a output, which is related to the distance between the threshold and the true counts associated to the predicates. If the predicates are far from the threshold, the fewer number of poking is required and hence a smaller privacy budget is spent. On the other hand, if the true count is very close to the threshold, then it requires a small noise and hence all the budget to decide the label of this predicate with confidence. When $c=0.01|D|$, $98\%$ predicates are within the range $[c-\alpha, c+\alpha]$, and hence to confidently decide the label for all these predicates require more poking and hence a larger privacy cost.
Consider the many predicates having counts close to $0.01|D|$, the cost of \icq-MPM is high.  As $c$ increases to $0.10|D|$, all predicates have sufficient different counts as $c$, then $1$ or $2$ times of poking are sufficient. When $c$ continue increases to $0.32|D|$, there is a single predicate that with true count (which is $0.3117|D|$) closer to $c$ , it again requires more pokings to make a confident decision. A similar behavior is seen when $c$ is around $ 0.6050|D|$.

Moreover, in the bad cases where $c$ is close to true counts, the actual privacy cost of \icq-MPM might be more expensive than the baseline \icq-LM.
% depending the query parameters and maximum allowed poking times. 
For example, when $c=0.01|D|$, \icq-LM is better. This is a case where \system under optimistic mode fails to choose \icq-LM as the optimal mechanism.

}\else
\full{
Figure~\ref{fig:cost_query_spec} shows the effects of these parameters.
\begin{figure*}[t]
	\centering
	\begin{subfigure}[h]{0.32\textwidth}%
		\includegraphics[width=\textwidth]{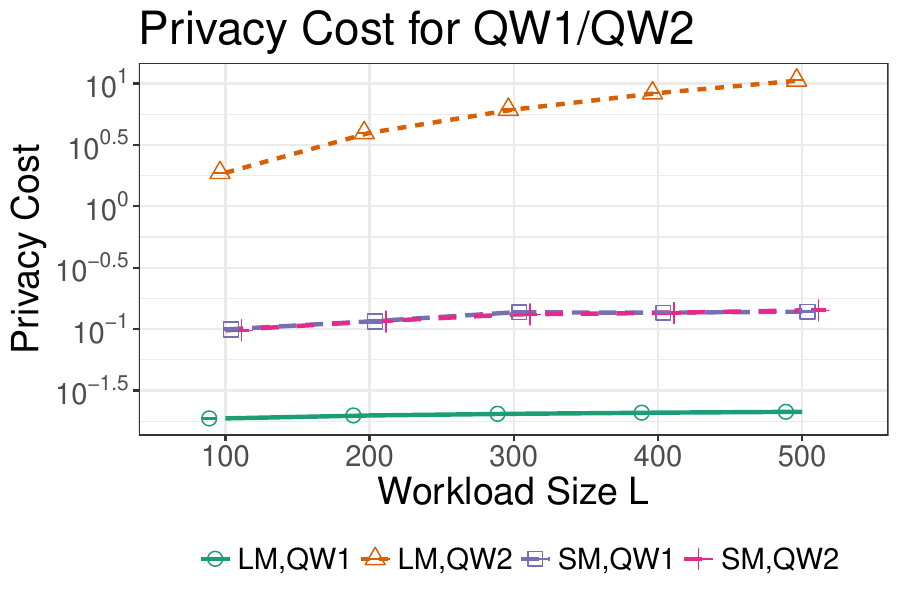}%
		\caption{Vary Workload Size}
		\label{fig:cost_workload_size_QW1}
	\end{subfigure}
	\begin{subfigure}[h]{0.32\textwidth}%
		\includegraphics[width=\textwidth]{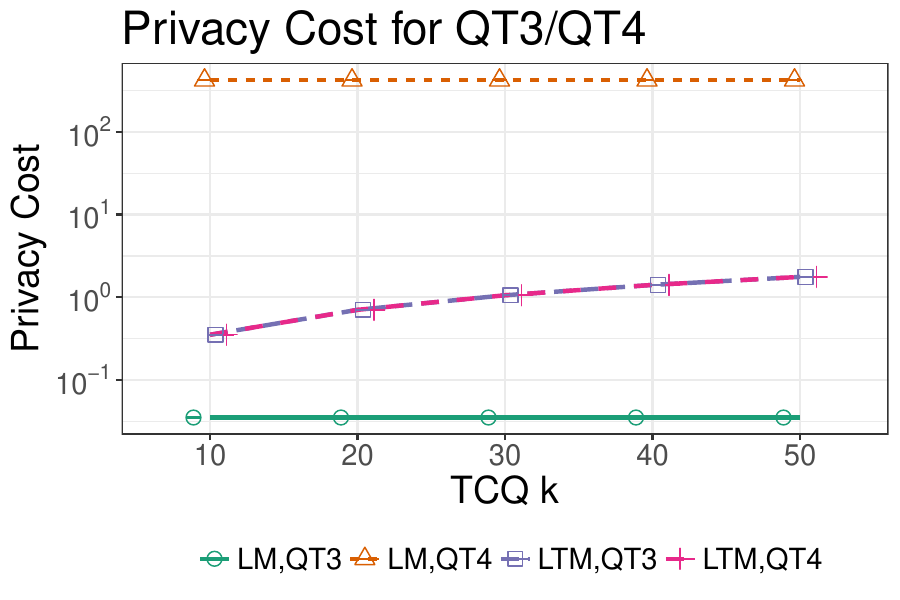}%
		\caption{Vary TCQ $k$}
		\label{fig:cost_tcq_k}
	\end{subfigure}
	\begin{subfigure}[h]{0.32\textwidth}%
		\includegraphics[width=\textwidth]{cost_icq_c}%
		\caption{Vary ICQ $c$}
		\label{fig:cost_icq_c}
	\end{subfigure}
	\caption{Privacy cost with query specific parameters: (a) Increasing workload size causes faster privacy cost increase fro \wcq-SM than \wcq-LM. (b) Increasing \tcq $k$ leads to faster privacy cost increase for \tcq-TM than \tcq-LM. (c) Varying \icq threshold $c$ affects privacy cost \icq-MPM. The closer $c$ relates to the true count, more attempts is needed, hence more privacy cost.}\label{fig:cost_query_spec}
\end{figure*}

\stitle{Vary Workload Size $L$.}
Figure~\ref{fig:cost_workload_size_QW1} shows the privacy cost of the baseline mechanism \wcq-LM and the special mechanism \wcq-SM when varying workload size using QW1 and QW2 templates.
\wcq-LM is highly dependent on the sensitivity of the query.
Hence, the privacy cost of \wcq-LM on QW1 changes very slowly with the workload size as the sensitivity of the workload in QW1 is $1$ while the privacy cost on QW2 is linear to workload size $L$ because the the sensitivity of cumulative histogram is  $L$.
On the other hand, \wcq-SM incurs similar privacy cost on QW1 and QW2, since both queries share the same workload size and accuracy parameters. The privacy cost of \wcq-SM for QW1 and QW2 are similar as they are using the same $H_2$ strategy matrix. Their difference in Table~\ref{tab:cost_mechanism} is mainly due to different runs of privacy cost estimation using MC simulation. However, their accuracy can be different, as QW2 with a cumulative histogram workload needs to add up ${\log L}$ number of noisy counts from the strategy matrix for one of its largest count. On the other hand, the histogram counts from QW1 requires only one noisy count.
Similar findings have been observed on other query types and mechanisms; therefore, we omit to show  other plots here.

\stitle{Vary top threshold $k$.}
Figure~\ref{fig:cost_tcq_k} shows the sensitivity of privacy cost of \tcq-LM and \tcq-TM with respect to varying $k$, using QT3 and QT4 as examples. 
The privacy cost of \tcq-LM is independent of $k$ as shown in Figure~\ref{fig:cost_tcq_k}, but the cost of \tcq-LTM increases linearly with k. 
The privacy cost of \tcq-LTM for both QT3 and QT4 have the same privacy cost, because its cost only depends on $k$ and the accuracy parameters, but the privacy cost of LTM for QT3 and QT4 are very different because the HD QT4 has larger sensitivity.

}\fi

In summary, we see that the optimal mechanism with the least privacy cost changes among queries and even the same query with different parameters.
This shows there is a great need for systems like \system to provide translation and identify optimal mechanisms.

\section{Case Study}\label{sec:casestudy}
In this section, we design an application benchmark based on entity resolution to show that (1) we can express real data exploration workflow using our exploration language;
(2) \system allows entity resolution workflow to be conducted with high accuracy and strong privacy guarantee, (3) \system allows trade-off between privacy and final task quality for a given exploration task.

%an external data analyst to conduct an exploration task on sensitive data with high accuracy under a reasonable privacy constraint specified by the data owner, (b) \system allows trade-off between privacy and final task quality for a given exploration task. This study sets a total privacy budget $B$ for a given exploration task which consists of a sequence of exploration queries. \system aims to answer as many queries as possible from the data analyst.
%The more query we are able to answers, the better we can tune the workflow and hence, the more the accuracy.

\subsection{Case Study Setting}
Entity Resolution (ER) is an application of identifying table records that refer to the same real-world object. In this case, we use the {\tt citations}~\cite{magellandata} dataset to perform exploration tasks. Each row in the table is a pair of citation records, with a binary label indicating whether they are duplicates or not. All the citation records share the same schema, which consists of $3$ text attributes of title, authors and venue, and one integer attribute of publication year.
A training set $D$ of size 4000 is sampled from {\tt citations} such that every record appears at most once. Two exploration tasks for entity resolution on $D$ are considered: \emph{blocking} and \emph{matching}. These two tasks are  achieved by learning a boolean formula $P$ (e.g. in DNF) over \emph{similarity predicates}. We express a similarity predicate $p$ as a tuple $(A, t, sim, \theta) \in \attrlist \times \trans \times \simf \times \simt$,
where  $A \in \attrlist$ is an attribute in the schema
 $t \in \trans$ is a transformation of the attribute value,
and $sim \in \simf$ is a similarity function that usually
takes a real value often restricted to $[0,1]$, and $\theta\in \simt$ is a threshold.
Given a pair of records $(r_1,r_2)$, a similarity predicate $p$ returns either `True' or `False' with semantics: $p(r_1, r_2)\equiv (sim(t(r_1.A), t(r_2.A)) > \theta)$.

In this case study, the exploration task for blocking is to find a boolean formula $P_b$ that identifies a small set of candidate matching pairs that cover most of the true matches, known as high \emph{recall}, with a small blocking cost (a small fraction of pairs in the data that return true for $P_b$). The exploration task for matching is to identify a boolean formula $P_m$ that identifies matching records that achieves high \emph{recall} and \emph{precision} on $D_t$.  Precision measures whether $P_m$ classifies true non-matches as matches, and recall measures the fraction of true matches that are captured by $P_m$. The quality of this task is measured by $F_1^{P_m}$, the harmonic mean of precision and recall.

Using exploration queries supported by \system, we can express two exploration strategies for each task: BS1 (MS1) using \wcq only to complete blocking (matching) task, and BS2 (MS2) using \icq/\tcq to conduct the blocking (matching) task. Each strategy generates a sequence of exploration queries which interact with \system and constructs a boolean formula for the corresponding task.
For each strategy, we randomly sample a concrete cleaner from our cleaner model and report mean and quartiles of its cleaning quality over $100$ runs.
\ifpaper
\conf{
Both the cleaner model and the exploration strategies are detailed in the full version~\cite{fullpaper}.
}
\else
\full{
Both the cleaner model and the exploration strategies are detailed by Figure~\ref{fig:blockstrategies} and Figure~\ref{fig:matchstrategies}  in Appendix~\ref{sec:cleaner_model}}.
\fi

%To measure the ER task quality, we use recall to represent blocking task quality, and use F1 score to measure matching task quality. In order to measure such qualities, we extracted the ground truth from {\tt citations}~\cite{magellandata} to form the data set. Each row in the table is a pair of citation records, with a binary label indicating whether they are duplicates or not. All the citation records share the same schema, which consists of $3$ text attributes of title, authors and venue, and one integer attribute of publication year. We ensure that any citation record can be sampled at most once into $D$ to set sensitivity of being $1$, and the duplicated and non-duplicated record pairs are evenly distributed in two data sets.

\begin{figure*}[h]
    \centering
%    \begin{subfigure}[h]{0.23\textwidth}%
%        \includegraphics[width=\textwidth]{../../figure/experiments/1b_blocking_BS1_1000_t008}%
%    \end{subfigure}%	
%    \begin{subfigure}[h]{0.23\textwidth}%
%        \includegraphics[width=\textwidth]{../../figure/experiments/1b_blocking_BS2_1000_t008}%
%    \end{subfigure}%
%    \begin{subfigure}[h]{0.23\textwidth}%
%       \includegraphics[width=\textwidth]{../../figure/experiments/1b_matching_MS1_1000_t008}%
%    \end{subfigure}%
%    \begin{subfigure}[h]{0.23\textwidth}%
%       \includegraphics[width=\textwidth]{../../figure/experiments/1b_matching_MS2_1000_t008}%
%    \end{subfigure}%
    
    \begin{subfigure}[h]{0.23\textwidth}%
        \includegraphics[width=\textwidth]{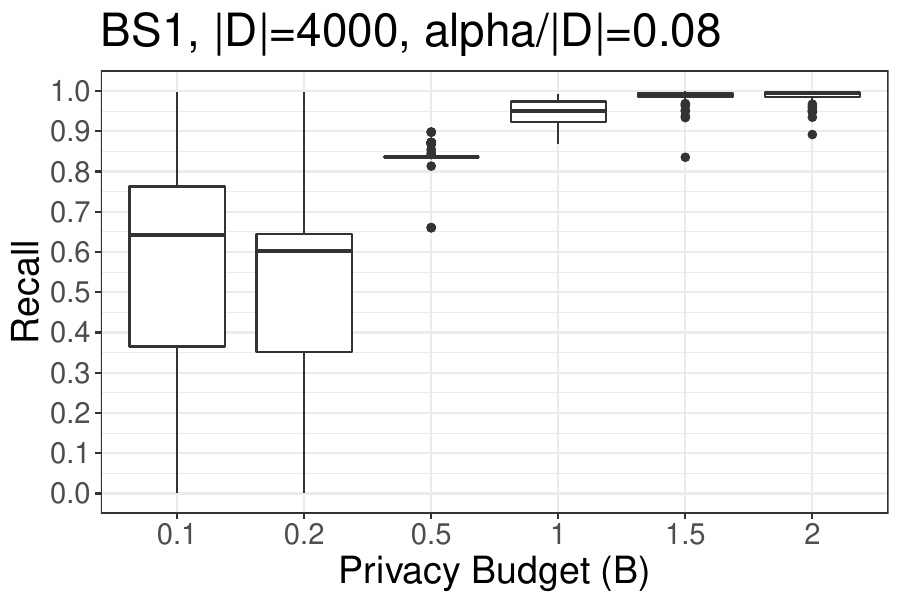}%
    \end{subfigure}%
    \begin{subfigure}[h]{0.23\textwidth}%
        \includegraphics[width=\textwidth]{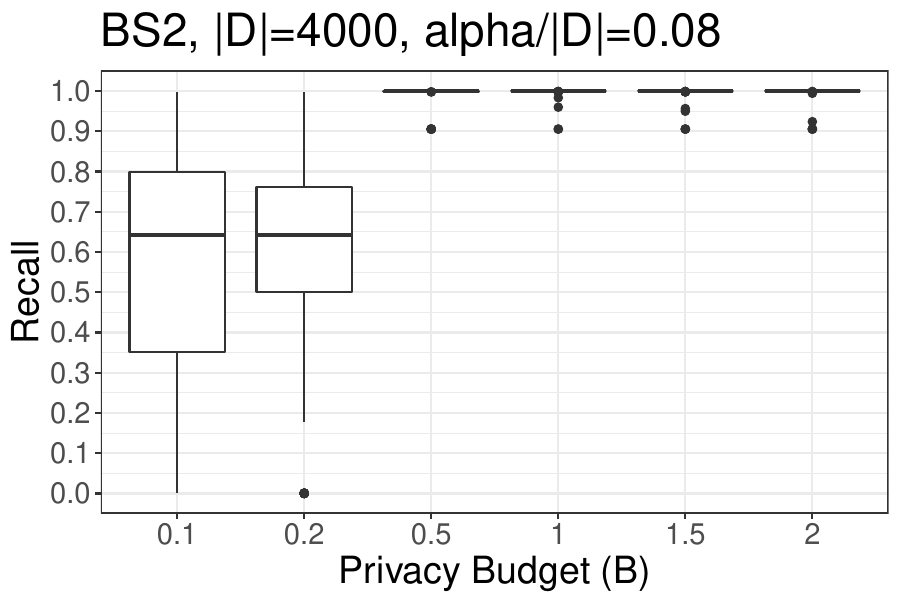}%
    \end{subfigure}%
    \begin{subfigure}[h]{0.23\textwidth}%
        \includegraphics[width=\textwidth]{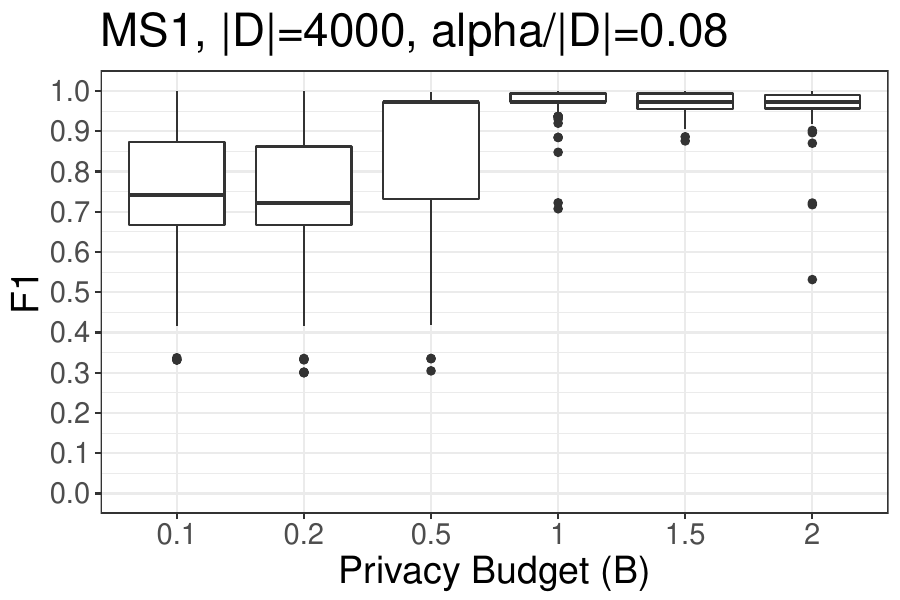}%
    \end{subfigure}
    \begin{subfigure}[h]{0.23\textwidth}%
        \includegraphics[width=\textwidth]{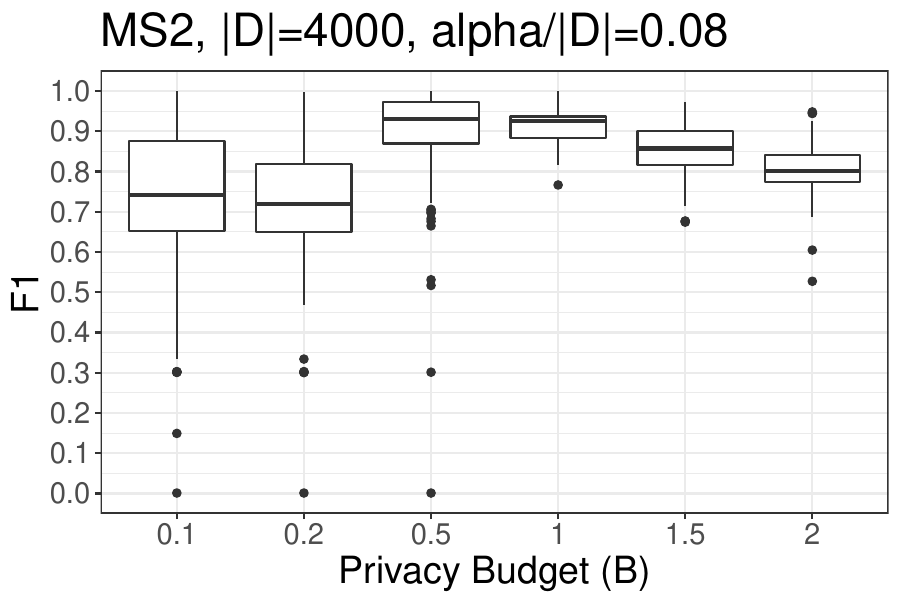}%
    \end{subfigure}%
    \vspace{-.75em}
    \caption{Performance of \system for blocking (BS1, BS2) and matching (MS1, MS2) tasks with increasing privacy budget $B$ at fixed $\alpha=0.08|D|$: the expected task quality improves with smaller variance as the budget constraint increases and gets stable. Fixing $\alpha$ fixes privacy cost per operation. Thus increasing B increases the number of queries answered. }
    \label{fig:budget_fixt}
\end{figure*}

\begin{figure*}[h]
    \centering
%    \begin{subfigure}[h]{0.23\textwidth}%
%    	\includegraphics[width=\textwidth]{../../figure/experiments/1b_blocking_BS1_1000_B1}%
%    \end{subfigure}%
%    \begin{subfigure}[h]{0.23\textwidth}%
%    	\includegraphics[width=\textwidth]{../../figure/experiments/1b_blocking_BS2_1000_B1}%
%    \end{subfigure}%
%    \begin{subfigure}[h]{0.23\textwidth}%
%        \includegraphics[width=\textwidth]{../../figure/experiments/1b_matching_MS1_1000_B1}%
%    \end{subfigure}%
%    \begin{subfigure}[h]{0.23\textwidth}%
%        \includegraphics[width=\textwidth]{../../figure/experiments/1b_matching_MS2_1000_B1}%
%    \end{subfigure}%
    
    \begin{subfigure}[h]{0.23\textwidth}%
        \includegraphics[width=\textwidth]{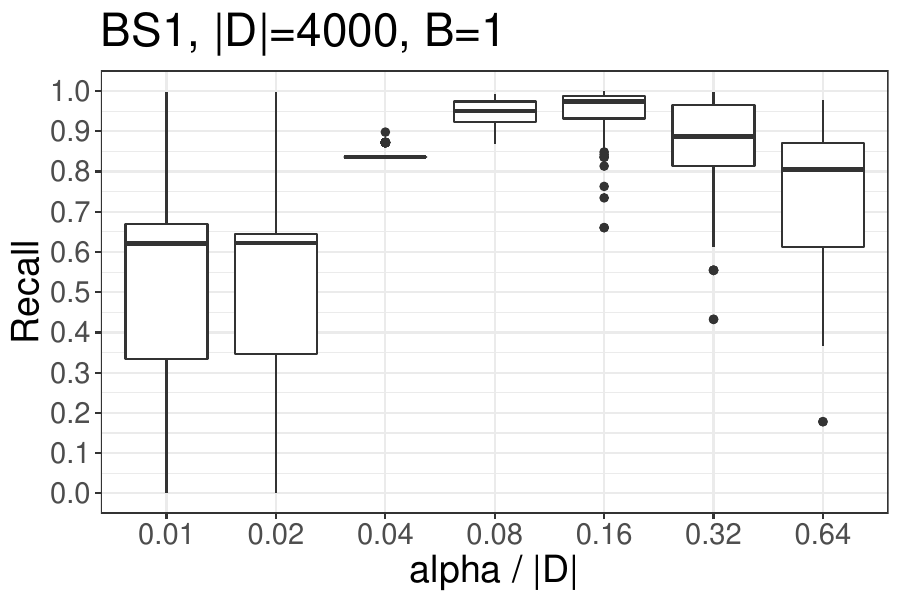}%
    \end{subfigure}%
    \begin{subfigure}[h]{0.23\textwidth}%
        \includegraphics[width=\textwidth]{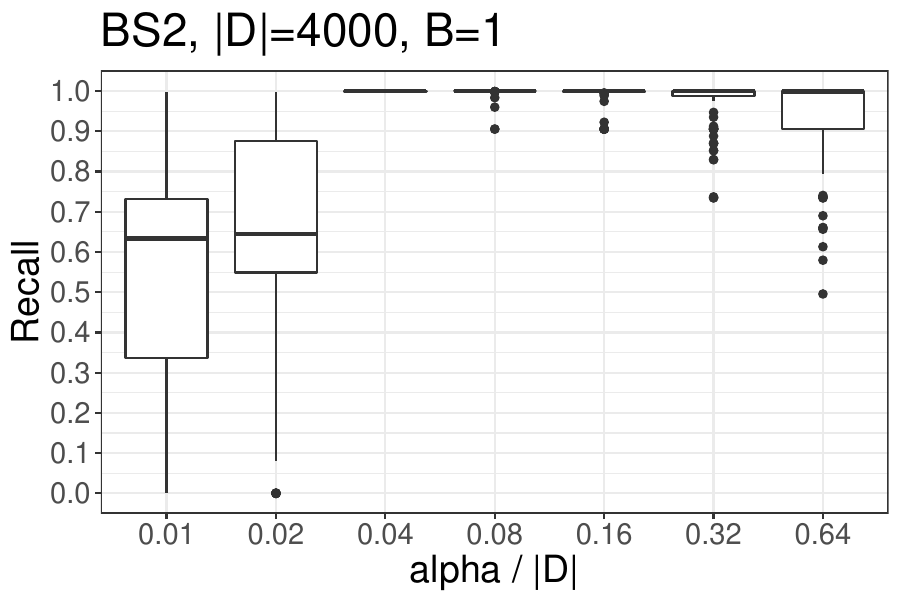}%
    \end{subfigure}%
    \begin{subfigure}[h]{0.23\textwidth}%
        \includegraphics[width=\textwidth]{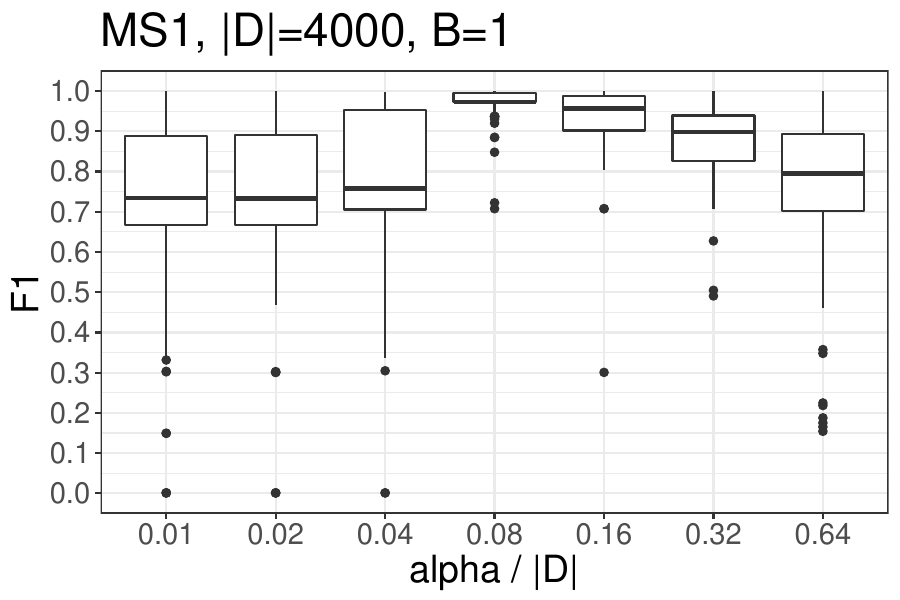}%
    \end{subfigure}
    \begin{subfigure}[h]{0.23\textwidth}%
        \includegraphics[width=\textwidth]{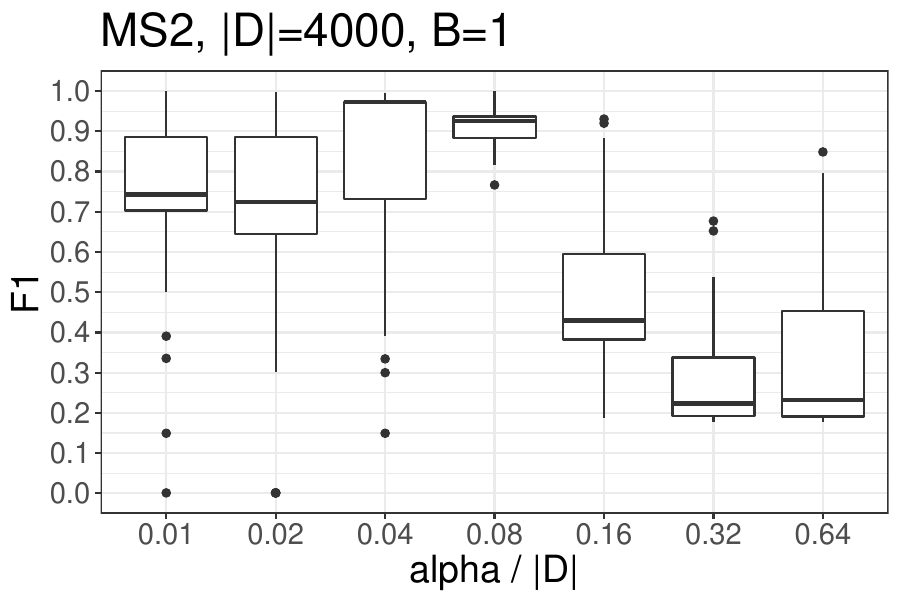}%
    \end{subfigure}%
    \vspace{-0.75em}
    \caption{Performance of \system for blocking (BS1, BS2) and matching (MS1, MS2) at fixed privacy budget $B=1$ with increasing $\alpha$ from $0.01|D_t|$ to $0.64|D_t|$. There exists an optimal $\alpha$ to achieve highest quality at a given privacy constraint. Increasing $\alpha$ decreases privacy cost per operation. Thus for a fixed budget this increases the number of queries. However, many queries each with a low privacy budget is not good for end-to-end accuracy.}
    \label{fig:tolerance_finiteb}
\end{figure*}

\ifpaper
\else
\full{
\begin{figure*}[h]
    \centering
    \begin{subfigure}[h]{0.23\textwidth}%
        \includegraphics[width=\textwidth]{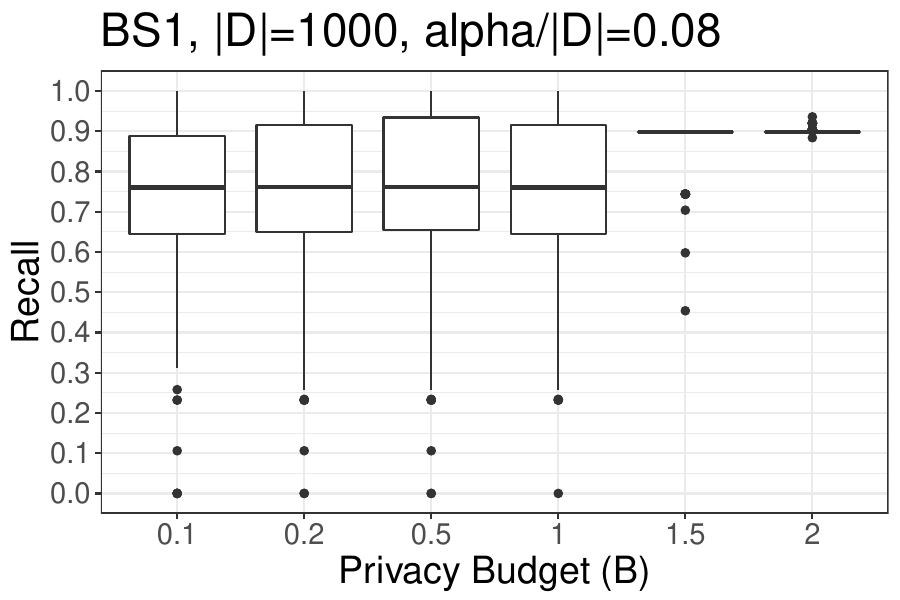}%
    \end{subfigure}%	
    \begin{subfigure}[h]{0.23\textwidth}%
        \includegraphics[width=\textwidth]{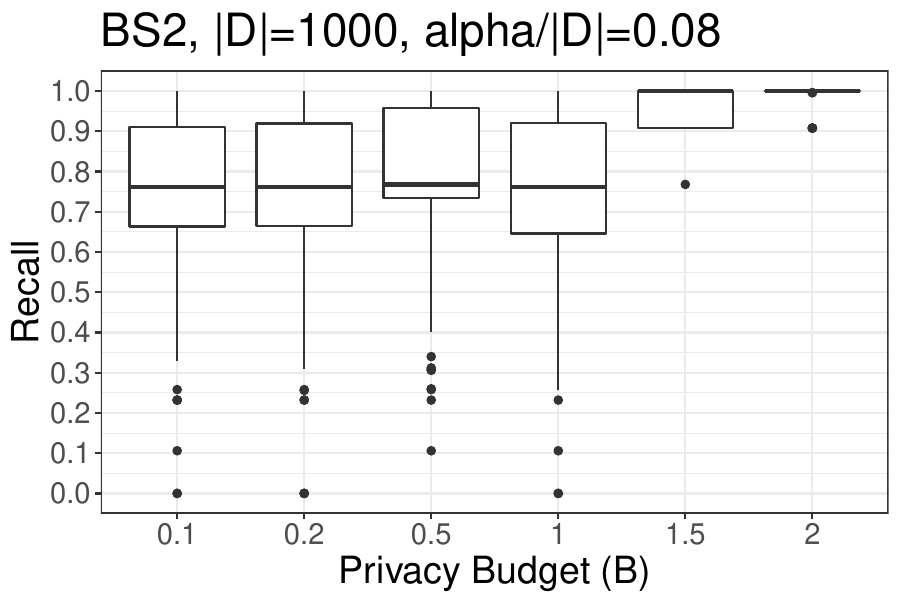}%
    \end{subfigure}%
    \begin{subfigure}[h]{0.23\textwidth}%
    	\includegraphics[width=\textwidth]{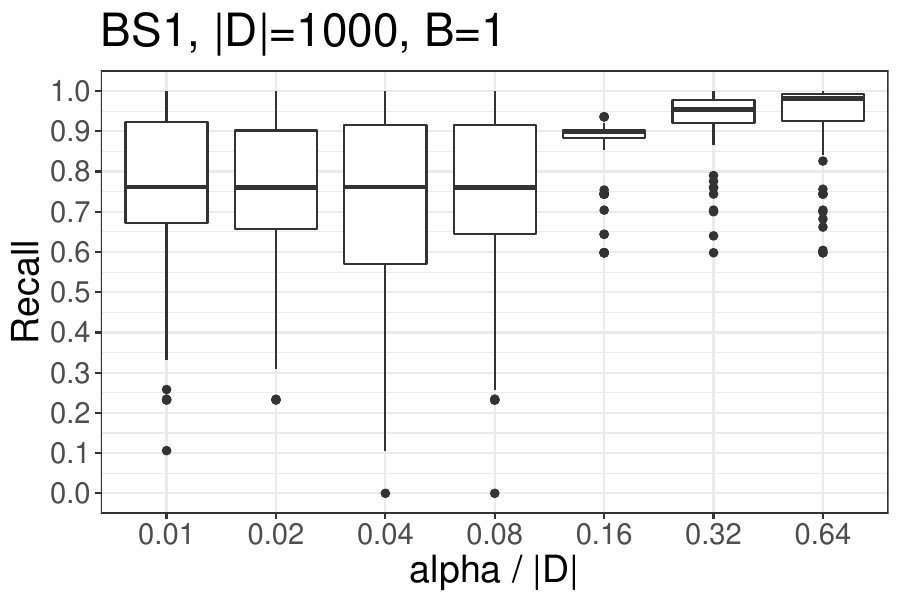}%
    \end{subfigure}%
    \begin{subfigure}[h]{0.23\textwidth}%
    	\includegraphics[width=\textwidth]{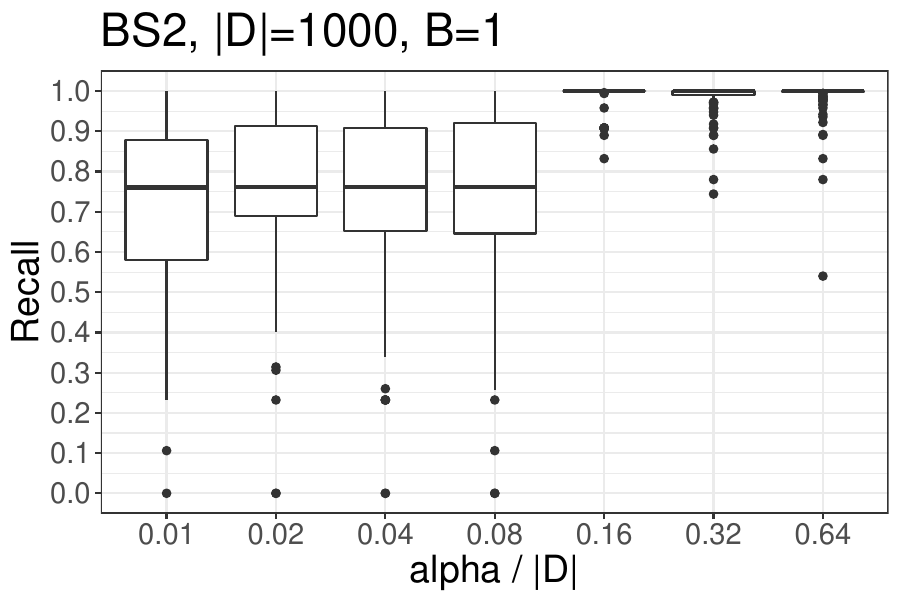}%
    \end{subfigure}%
    \vspace{-.75em}
    \caption{Performance of \system for blocking (BS1, BS2) tasks at $|D|=1000$. Compared with Figure~\ref{fig:budget_fixt} where $\alpha=0.08|D|$, the privacy budget that needs to achieve good recall increases when data size is smaller. Compared with Figure~\ref{fig:tolerance_finiteb} where $B=1$, the optimal $\alpha$ actually increases. }
    \label{fig:vary_dsize}
\end{figure*}
}
\fi

\subsection{End-to-End Task Evaluation}
We report the end-to-end performance of \system on the four exploration strategies.

\stitle{Vary Privacy Constraint.}
Given an exploration strategy from {BS1, BS2, MS1, MS2} on training data $D$, a sequence of queries is generated based on its interaction with \system until the privacy loss exceeds the privacy constraint $B$ specified by the data owner. The accuracy requirement for this set of experiments is fixed to be $\alpha=0.008|D_t|,\beta=0.0005$. The privacy constraint varies from $0.1$ to $2$. Each exploration strategy is repeated $100$ times under a given privacy constraint and we report the output quality of these $100$ runs at different privacy constraint.

Figure~\ref{fig:budget_fixt} shows the exploration quality (recall for blocking task and F1 for matching task) of $100$ runs of each exploration strategy at $B=\{0.1, 0.2, 0.5, 1, 1.5, 2\}$. We observe that the expected task quality (median and variance) improves as the budget constraint increases and gets stable after, for example reaching $B \geq 1$ for MS1.  Since we fixed $\alpha$, the privacy loss for each exploration query is also fixed. Thus, $B$ directly controls the number of queries \system answers before halting. For small $B<0.2$, only a few queries are answered and the cleaning quality is close to random guessing. As $B$ increases, more queries can be learned about the dataset. After $B$ reaches a certain value around $1.0$, \system can answer sufficient number of queries; therefore obtaining high accuracy.  For MS2, when it reaches good F1 score at $B=1$, more noisy answers can mislead the MS2 strategy to add more predicates to the blocking conjunction function, and decrease the quality.

The privacy cost used by each \icq and \tcq is generally less than what has spent on the corresponding \wcq, as less information is shown to the data analyst. Given the same privacy budget, e.g. when $B=0.5$, more queries can be answered in BS2 than BS1, results in a 25\% better recall. 
%Moreover, BS2 reaches a recall of 1.0 when $B=0.5$ earlier than BS1 when $B=1.5$. 
We observe the same trend in MS1 and MS2. 
This shows that it is important for \system to support translation for different types of queries for better end-to-end accuracy.

\stitle{Vary Accuracy Requirement.}
This section shows the task quality of each exploration strategy at different accuracy requirements under a fixed privacy constraint $B=1.0$. For each $\alpha \in \{0.01, 0.02, 0.04, 0.08, 0.16, 0.32, 0.64\}|D|$, the exploration strategy interacts with \system until the privacy loss exceeds $B=1.0$. Each exploration is repeated $100$ times and we report the quality of the constructed boolean formula of these $100$ runs. As shown in Figure~\ref{fig:tolerance_finiteb}
 the quality of the four exploration strategies all improves first as accuracy requirement relaxes and then degrades again under $B=1.0$. This is because when there is a privacy constraint, \system only allows a limited number of queries to be answered. Given this fixed privacy budget $B=1$ and a fixed accuracy requirement $\alpha$, the number of queries can be answered is bounded and related to $\alpha$. When $\alpha$ first increases from $0.01|D|$ to $0.08|D|$, more queries can be answered and give more information on which predicates to choose. However, as $\alpha$ keeps growing, the answers get noisier and misleading, resulting in the drop in quality, even though more questions can be answered.
 \ifpaper
Hence, choosing the optimal $\alpha$ for different queries in an exploration process is an interesting direction for future work.
\else
\full{

\stitle{Vary Data Size.} 
We also experimented the same strategies with a different data size $|D|=1000$, using blocking strategies as examples to study the effects of data size, as shown in Figure~\ref{fig:vary_dsize}. Comparing to Figure~\ref{fig:budget_fixt} where $\alpha$ is fixed at $0.08|D|$, the privacy cost to achieve optimal recall is larger when $|D|$ is smaller. BS1 and BS2 can achieve very good recall when $|D|=4000$ given a privacy budget $B=1$, but they require at least $B=1.5$ to reach similar quality when $|D|=1000$. On the other hand, comparing to Figure~\ref{fig:tolerance_finiteb} where the privacy budget is fixed $B=1$, the optimal $\alpha$ for smaller $|D|$ is very close. BS1 and BS2 reach the highest recall at about $\alpha=0.16\cdot 1000=160$ which is close to the optimal $\alpha=0.8\cdot 4000 =320$ when $|D|=4000$. This observation suggests an interesting direction for future work: choosing the optimal $\alpha$ for different queries in an exploration process.
}
\fi
%\vspace{-5mm}

%!TEX root=./main.tex
\section{Conclusions}\label{sec:discussion}
We proposed \system, a framework that allows data analysts to interact with sensitive data
while ensuring that their interactions satisfy differential privacy.
Using experiments with query benchmarks and entity resolution application,
we established that \system allows high exploration quality with a reasonable privacy loss.

\system opens many interesting future research directions.
First, more functionalities can be added to \system:
(a) a recommender which predicts the subsequent interesting queries
and advises the privacy cost of these queries to the analyst;
(b) an inferencer which uses historical answers to reduce the privacy cost
of a new incoming query; and
(c) a sampler which incorporate approximate query processing to
have 3-way trade-off between efficiency, accuracy, and privacy.
More translation algorithms from accuracy to privacy,
especially for data-dependent mechanisms and non-linear queries
can be implemented in \system
to further save privacy budget throughout the exploration.
We show the exploration queries in \system to support entity resolution tasks.
These exploration queries can be extended to support other exploration tasks,
e.g., schema matching, feature engineering, tuning machine learning workloads.
This would also require extending our system to handle relations with multiple tables,
constraints like functional dependencies, etc.
\system turns the differentially private algorithm design problem on its head
-- it minimizes privacy loss given an accuracy constraint.
This new formulation has applications on sensitive data exploration
and can trigger an entirely new line of research.

\stitle{Acknowledgments:} This work was supported by the National Science Foundation under grants 1253327, 1408982; by DARPA and SPAWAR under contract N66001-15-C-4067; by NSERC Discovery grant; and by Thomson Reuters-NSERC IRC grant.

\eat{
\ifpaper
\eat{
\system opens many interesting future research directions:
(1) more functionalities can be added to \system; 
(2) more translation algorithms especially for data-dependent mechanisms to further save privacy budget;
(3) extend \system to handle relations with multiple tables anbd constraints.
}\else
\full{
\system opens many interesting future research directions.
First, more functionalities can be added to \system:
(a) a recommender which predicts the subsequent interesting queries
and advises the privacy cost of these queries to the analyst;
(b) an inferencer which uses historical answers to reduce the privacy cost
of a new incoming query;
(c) a sampler which incorporate approximate query processing to
have 3-way trade-off between efficiency, accuracy, and privacy.
In terms of translation from accuracy to privacy,
more translation algorithms from accuracy to privacy,
especially for data-dependent mechanisms can be implemented in \system
to further save privacy budget throughout the exploration.
Extending \system to handle non-linear query types
is also desired as discussed in Appendix~\ref{app:otherqueries}.
We show the exploration queries in \system to support entity resolution tasks.
These exploration queries can be extended to support other exploration tasks,
e.g., schema matching, feature engineering, tuning machine learning workloads.
This would also require extending our system to handle relations with multiple tables,
constraints like functional dependencies, etc.
\system turns the differentially private algorithm design problem on its head
-- it minimizes privacy loss given an accuracy constraint.
This problem has applications on sensitive data exploration
and can trigger an entirely new line of research.
}\fi
}

\bibliographystyle{abbrv}
\bibliography{sigproc}

%!TEX root=./main.tex
\appendix

\section{Theorems and Proofs}\label{app:proofs}

\subsection{Laplace Mechanism (Theorem~\ref{theorem:pc_lm})}
\label{app:lm_proof}

\eat{
First, it is easy to see that this mechanism also satisfies $\epsilon = \|\workload \|_1/b$-differential privacy by the property of Laplace mechanism and the post-processing property of differential privacy (Theorem~\ref{theorem:post}). We next show how the resulting mechanism satisfies the accuracy bounds specified in each query.
}

\change{

\noindent\textbf{Privacy Proof for \wcq}\\
{\sc Proof:}  For any pair of neighbouring databases $D$ and $D^\prime$ such that $| D \backslash D^\prime \cup D^\prime \backslash D| = 1$,
given a \wcq query $q_{W}$,
LM adds noise drawn from $Lap(b)$ to query $c_{\phi_i}$, where $b=\| \workload \|_1 / \epsilon$. Consider a column vector of counts $y = [y_1, \cdots, y_L]$,
\ifpaper
\conf{
\begin{eqnarray*}
\frac{\Pr[q_{\workload}(D) + Lap(b_{\workload})^L = y]}{\Pr[q_{\workload}(D^\prime) + Lap(b_{\workload})^L = y]}  \leq exp(\epsilon) 
\end{eqnarray*}
}
\else
\full{
\begin{eqnarray*}
\lefteqn{ \frac{\Pr[q_{\workload}(D) + Lap(b_{\workload})^L = y]}{\Pr[q_{\workload}(D^\prime) + Lap(b_{\workload})^L = y]}  } \\
&=& \frac{\prod_1^L \Pr[c_{\phi_i}(D) + \eta_i = y_i ]}{\prod_1^L \Pr[c_{\phi_i}(D^\prime) + \eta_i = y_i ]} = \prod_1^L \frac{exp(\frac{-\epsilon |y_i-c_{\phi_i}(D)|}{\| \workload \|_1})}{ exp(\frac{-\epsilon |y_i-c_{\phi_i}(D^\prime)|}{\| \workload \|_1})} \\
&=& exp(\frac{\epsilon}{\|\workload \|_1} \sum_1^L (| y_i - c_{\phi_i}(D^\prime)| - | y_i - c_{\phi_i}(D)|)) \\
&\leq & exp( \frac{\epsilon}{\| \workload \|_1} \sum_1^L  (|c_{\phi_i}(D) - c_{\phi_i}(D^\prime)|)) = exp(\epsilon) 
\end{eqnarray*}
}\fi
Therefore, it satisfies  $\epsilon = \|\workload \|_1/b$-differential privacy. \qed

\noindent\textbf{Privacy Proof for \icq and \tcq}\\
{\sc Proof:} For \icq (and \tcq, respectively), Line~\ref{alg:lm:icq} (Line~\ref{alg:lm:tcq}) post-processes the noisy answers only without accessing the true data. By the post-processing property of differential privacy (Theorem~\ref{theorem:post}), Laplace mechanism for \icq (\tcq) satisfies $\epsilon = \|\workload \|_1/b$-differential privacy. \qed
}

\noindent\textbf{Accuracy Proof for \wcq}\\
%Given a query $q$ where $q$.type $=$\wcq, Laplace mechanism (Algorithm~\ref{algo:lm}) denoted by $LM(q,\alpha,\beta,D)$ can achieve $(\alpha,\beta)$-\wcq accuracy by executing the function $\runFunc(q,\alpha,\beta,D)$ for any $D\in \tabledomain$, and satisfy differential privacy with a minimal cost returned by the function $\epsilon=\transFunc(q,\alpha,\beta)$.
%\begin{proof}
{\sc Proof:} Give a \wcq $q_W$, for any $D\in \tabledomain$,
setting $b\leq \frac{\alpha}{\ln(1/(1-(1-\beta)^{1/L}))}$ bounds the failing probability, i.e.,
\begin{eqnarray*}
\lefteqn{\Pr[\|LM(\workload,\data)-q_W(D)\|_{\infty} \geq \alpha] \nonumber} \\
&=& 1-\prod_{i\in[1,L]} (1-\Pr[|\eta_i| > \alpha]) = 1-(1- e^{-\alpha/b})^L \leq \beta
\qed
\end{eqnarray*}
%\end{proof}

\noindent\textbf{Accuracy Proof for \icq}\\
%Given a query $q$ where $q$.type $=\icq$, Laplace mechanism (Algorithm~\ref{algo:lm}) denoted by $LM(q,\alpha,\beta,D)$ can achieve $(\alpha,\beta)$-\icq accuracy by executing the function $\runFunc(q,\alpha,\beta,D)$ for any $D\in \tabledomain$, and satisfy differential privacy with a minimal cost returned by the function $\epsilon=\transFunc(q,\alpha,\beta)$.
%\begin{proof}
{\sc Proof:} Given \icq $q_{W,>c}$, for any $D\in \tabledomain$,
setting $b\leq \frac{\alpha}{\ln(1/(1-(1-\beta)^{1/L}))-\ln2}$ bounds the failure probability, i.e.,
\begin{eqnarray*}
\lefteqn{\Pr[|\{\phi \in M(D) ~|~ c_{\phi}(D) < c-\alpha\}|>0] \nonumber} \\
%&=& 1- \Pr[|\{\phi \in M(D) ~|~ c_{\phi}(D) < c-\alpha\}|=0] \nonumber \\
&=& 1- \prod_{\phi\in W: c_{\phi}(D)-c+\alpha<0} (1-\Pr[c_{\phi}(D)-c+\eta>0]) \nonumber\\
&\leq & 1- \prod_{\phi\in W: c_{\phi}(D)-c+\alpha<0} (1-\Pr[\eta>\alpha]) \nonumber\\
%&=& 1- \prod_{\phi\in W: c_{\phi}(D)- c+\alpha<0} (1-e^{-\alpha/b}/2) \nonumber \\
&\leq& 1-(1-e^{-\alpha/b}/2)^L < \beta \nonumber
\end{eqnarray*}
The proof for the other condition is analogous. $\qed$
\eat{and similarly,
\begin{eqnarray}
&& \Pr[|\{\phi \in (W-M(D)) ~|~ c_{\phi}(D) > c+\alpha\}|>0] \nonumber \\
&=& 1- \Pr[|\{\phi \in (W-M(D)) ~|~ c_{\phi}(D) > c+\alpha\}|=0] \nonumber \\
&=& 1 - \prod_{\phi \in W: c_{\phi}(D)- c-\alpha>0} (1-\Pr[c_{\phi}(D)-c+\eta<0]) \nonumber\\
&\leq & 1-  \prod_{\phi \in W: c_{\phi}(D)- c-\alpha>0} (1-\Pr[\eta<-\alpha]) \nonumber\\
&=& 1- \prod_{\phi\in W: c_{\phi}(D)- c-\alpha >0} (1-e^{-\alpha/b}/2) \nonumber \\
&\leq& 1-(1-e^{-\alpha/b}/2)^L < \beta \nonumber
\end{eqnarray}}
%\end{proof}

\noindent\textbf{Accuracy Proof for \tcq}\\ %\label{app:lm_tcq_proof}
%Given a query $q$ where $q$.type $=\tcq$, Laplace mechanism (Algorithm~\ref{algo:lm}) denoted by $LM(q,\alpha,\beta,D)$ can achieve $(\alpha,\beta)$-\tcq accuracy by executing the function $\runFunc(q,\alpha,\beta,D)$ for any $D\in \tabledomain$, and satisfy differential privacy with a minimal cost returned by the function $\epsilon=\transFunc(q,\alpha,\beta)$.
%\begin{proof}
{\sc Proof:}
Given a \tcq $q_{W,k}$, for any $D\in \tabledomain$,
W.L.O.G. let $c_{\phi_1}(D)\geq \cdots \geq c_{\phi_k}(D) \cdots \geq c_{L}(D)$
and the noise added to these counts be $\eta_1,\ldots,\eta_L$ respectively.
Let $c_k = c_{\phi_k}(D)$ the $k^{th}$ largest counting value.
Setting noise parameter $b\leq \frac{\alpha}{2\ln(L/(2\beta))}$
bounds the failing probability, i.e.,
\begin{eqnarray}
&& \Pr[ |\{\phi \in M(D) ~|~ c_{\phi}(D) < c_k-\alpha \}| > 0] \nonumber \\
&\leq&1-\Pr[(\max_{i>k,c_{\phi_i}(D)<c_k-\alpha} \eta_i <\frac{\alpha}{2})\wedge (\min_{i\leq k} \eta_i > - \frac{\alpha}{2})] \nonumber \\
&=& \Pr[(\max_{i>k,c_{\phi_i}(D)<c_k-\alpha} \eta_i \geq \frac{\alpha}{2})\vee (\min_{i\leq k} \eta_i \geq - \frac{\alpha}{2})] \nonumber \\
&\leq & (L-k)e^{-\alpha/(2b)}/2 + ke^{-\alpha/(2b)}/2 = L e^{\alpha/(2b)}/2 \leq \beta
\end{eqnarray}
The proof for $\Pr[ |\{\phi \in (\Phi-M(D)) ~|~ c_{\phi}(D) > c_k + \alpha \}| > 0 ] \leq \beta$ is analogous. \qed

\eat{
\todo{FIX proof as discussed} The output of Algorithm~\ref{algo:lm} for \tcq are
$\{\phi_{i_1},\ldots, \phi_{i_k}\}$
which have the highest noisy counts.
Suppose their noisy counts are in the order of
$\tilde{x}_{i_1} \geq \tilde{x}_{i_2} \geq \cdots \geq \tilde{x}_{i_k}$.
As $c_k$ is the true answer to the $k$th largest linear counting query of all $L$ linear counting queries,
the largest answer to the smallest $L-k+1$ queries should be no less than $c_k$. Hence, we have
%\begin{equation}
$\max_{j\in \{i_{k},\ldots,i_{L}\}} c_{\phi_{j}}(D) \geq c_k$.
%\end{equation}
Similarly, the smallest of the $k$ largest counting queries should be no greater than $c_k$. That is, 
$\min_{j\in \{i_{1},\ldots,i_{k}\}} c_{\phi_{j}}(D) \leq c_k$.

First we show that when
$b\leq \frac{\alpha}{2(\ln L + \ln(k/\beta))}$
if a count $c_{\phi}(D)$ is too small,
$\phi$ will not be included in the answer $a$ with high probability:
\begin{eqnarray}
\lefteqn{\Pr[\phi \in a ~|~ c_{\phi}(D) < c_k-\alpha] \nonumber} \\
&=& \Pr[c_{\phi}(D)+\eta \geq \max_{j \in \{i_{k},\ldots,i_L\}} (c_{\phi_j}(D) + \eta_j) ~|~ c_{\phi}(D) < c_k-\alpha] \nonumber \\
&\leq & \Pr[c_{\phi}(D)+\eta \geq \max_{j \in \{i_{k},\ldots,i_L\}} (c_k + \eta_j) ~|~ c_{\phi}(D) < c_k-\alpha] \nonumber \\
&= & \Pr[c_{\phi}(D) - c_k \geq \max_{j \in \{i_{k},\ldots,i_L\}} \eta_j -\eta ~|~ c_{\phi}(D) < c_k-\alpha] \nonumber \\
&<& \Pr[\max_{j \in \{i_{k},\ldots,i_L\}} \eta_j -\eta < -\alpha] \nonumber \\
&\leq & \Pr[\max_{j \in \{i_{k},\ldots,i_L\}} \eta_j < -\alpha/2] + \Pr[\eta >\alpha/2] \nonumber \\
&\leq & (L-k+1) e^{-\frac{\alpha}{2b}}/2 + e^{-\frac{\alpha}{2b}}/2 \nonumber 
\ < \ L e^{-\frac{\alpha}{2b}} \leq \beta/k \nonumber
\end{eqnarray}
Hence, every selected $\phi\in a$ has $c_{\phi}(D) > c-\alpha$ with probability $1-\beta$ for $|a|=k$.

Then, we would like to show that if a count $c_{\phi}(D)$ is sufficiently large, then the probability of missing it in the output is small.
\begin{eqnarray}
&&\Pr[\phi \notin a ~|~ c_{\phi}(D) > c_k+\alpha] \nonumber \\
&=& \Pr[c_{\phi}(D)+\eta < \min_{j \in \{i_{1},\ldots,i_{k}\}} (c_{\phi_j}(D) + \eta_j) ~|~ c_{\phi}(D) > c_k+\alpha] \nonumber \\
&\leq & \Pr[c_{\phi}(D)+\eta < \min_{j \in \{i_{1},\ldots,i_k\}} (c_k + \eta_j) ~|~ c_{\phi}(D) > c_k+\alpha] \nonumber \\
&= & \Pr[c_{\phi}(D) - c_k < \min_{j \in \{i_{1},\ldots,i_k\}} \eta_j -\eta ~|~ c_{\phi}(D) > c_k+\alpha] \nonumber \\
&<& \Pr[\min_{j \in \{i_{1},\ldots,i_k\}} \eta_j -\eta > \alpha] \nonumber \\
&\leq & \Pr[\min_{j \in \{i_{1},\ldots,i_k\}} \eta_j >\alpha/2]
+ \Pr[\eta <-\alpha/2] \nonumber \\
&\leq & k e^{-\frac{\alpha}{2b}}/2 + e^{-\frac{\alpha}{2b}}/2 \nonumber \\
&<& k e^{-\frac{\alpha}{2b}} \leq \beta/L
\end{eqnarray}
Hence, all boolean formulae $\phi\in \{\phi_1,\ldots,\phi_L\}$ with $c_{\phi}(D)>c+\alpha$ are selected with probability $1-\beta$. $\qed$
%With probability $1-\beta$, every linear counting query with true answer at least $c_k+\alpha$
%is selected, and every selected linear counting query has a true answer at least $c_k-\alpha$.
%\todo{double check \cite{Li:2012:PFI:2350229.2350251}}
%\end{proof}
}

\subsection{Strategy-based Mechanism for \wcq}\label{app:strategy_proof}

\begin{theorem}\label{theorem:strategy_upperbound}
Given a \wcq $q_W:\tabledomain \rightarrow \mathbb{R}^L$,
for a table $D\in \tabledomain$, let $(\workload,\data) = \transform(W,D)$.
Let $\strategy$ be a strategy used in Algorithm~\ref{algo:wcq-sm} to answer $q_W$.
When $\epsilon\geq  \frac{\|\strategy\|_1\|\workload\strategy^+\|_F}{\alpha\sqrt{\beta/2}}$,
%When $b\leq \frac{\alpha\sqrt{\beta/2}}{\|\workload\strategy^+\|_F}$,
$\strategy$-strategy mechanism achieves $(\alpha,\beta)$-\wcq accuracy.
\end{theorem}

\begin{proof}
The noise vector added to the final query answer of $q_W(\cdot)$ using $\strategy$-strategy mechanism
is $\hat{\noisev}=[\hat{\eta}_1,\ldots,\hat{\eta}_L] = (\workload \strategy^+) \noisev$.
Each noise variable $\hat{\eta}_i$ has a mean of 0, and a variance of $\sigma_i^2 = c_i \cdot (2b^2)$,
where $c_i = \sum_{j=1}^l (\workload\strategy^+)[i,j]^2$,
and hence $\Pr[|\hat{\eta}_i| \geq \alpha] \leq \frac{2c_ib_{\strategy}^2}{\alpha^2}$ by Chebyshev's inequality.
By union bound, the failing probability is bounded by
\begin{eqnarray}
&&\Pr[\|LM(\workload,\data)-q_W(D)\|_{\infty} \geq \alpha] \nonumber \\
&=& \Pr[\cup_{i\in[1,L]}|\eta_i| \geq \alpha] \leq \sum_{i\in[1,L]}\frac{2c_ib^2}{\alpha^2} \leq \beta \nonumber
\end{eqnarray}
It requires $b\leq \frac{\alpha\sqrt{\beta/2}}{\|\workload\strategy^+\|_F}$, hence $\epsilon \geq \|A\|_1/b \geq \frac{\|\strategy\|_1\|\workload\strategy^+\|_F}{\alpha\sqrt{\beta/2}}$.
\end{proof}

\stitle{Proof for Theorem~\ref{theorem:strategy_tightbound}}
\begin{proof}
Given an $\epsilon$, the simulation in the function $\textsc{estimateFailingRateMC}()$ ensures that
with high probability $1-p$, the true failing probability $\beta_t$ to
bound $\| (\workload\strategy^+)\eta\|_{\infty}$ by $\alpha$
lies within $\beta_e\pm \delta \beta$.
The failing probability to bound $\beta_t <\beta_e+\delta\beta$ is $p/2$.
By union bound, $\epsilon$ ensures $(\alpha,\beta')$-\wcq accuracy, where $\beta'< \beta+\delta\beta+p/2$.
If $(\beta+\delta\beta)(1-p/2)+p/2 < \beta+\delta\beta +p/2 <\beta$, then this $\epsilon$ ensures $(\alpha,\beta)$-\wcq accuracy.
Beside this estimation, in the binary search of $\transFunc()$,
we stop when $\epsilon_{\min}$ and $\epsilon_{\max}$ are sufficiently close.
Hence, the privacy cost returned by $\transFunc()$ is an approximated
minimal privacy cost required for $(\alpha,\beta)$-\wcq accuracy.
\end{proof}

\eat{
We would like to show Theorem~\ref{theorem:strategy_tightbound}:
given a workload counting query $q_W:\tabledomain \rightarrow \mathbb{R}^L$,
answering $q_W$ with $\strategy$-strategy mechanism
by executing $\runFunc(q_W,\alpha,\beta,D)$ in Algorithm~\ref{algo:wcq-sm}
achieves $(\alpha,\beta)$-\wcq accuracy for any $D\in \tabledomain$
with (close to) minimal privacy cost of $\epsilon$ as returned by $\transFunc(q_W,\alpha,\beta)$.
}

\subsection{Multi-Poking Mechanism for \icq}\label{app:mpm_proof}
We would like to show Theorem~\ref{theorem:pc_mpm} that
given a query $q_{W,>c}$,  multi-poking mechanism (Algorithm~\ref{algo:mpm}),
achieves $(\alpha,\beta)$-\icq accuracy by executing function $\runFunc(q_{W,>c}, \alpha,\beta,D)$,
with differential privacy cost of $\epsilon$ returned by function $\transFunc(q_{W,>c},\alpha,\beta)$.
\begin{proof}
For each $\phi\in W$, 
(i) when $q_{\phi}(D)<c-\alpha$,
\begin{eqnarray}
&& \Pr[\phi \in MPM^{\alpha,\beta}_{q_{\phi,>c}}(D) ~|~ c_{\phi}(D)<c-\alpha] \nonumber \\
&=& \sum_{i=0}^{m-1} \Pr[c_{\phi}(D)-c +\eta -\alpha_{i} + \alpha >0 ~|~c_{\phi}(D)-c +\alpha <0]\nonumber \\
&<& \sum_{i=0}^{m-1} \Pr[\eta_i >\alpha_{i}] = \sum_{i=0}^{m-1} e^{-\alpha_{i}\epsilon_i/\|\workload\|_1}/2 =m \cdot \beta/(mL) = \beta/L \nonumber
%&<& \Pr[\eta >\alpha_{0}] + \beta/2 = e^{-\alpha_{0}\epsilon_0}/2 +\beta/2 = \beta.
\end{eqnarray}
and (ii) when $q_{\phi}(D)<c+\alpha$,
\begin{eqnarray}
&& \Pr[\phi \notin MPM^{\alpha,\beta}_{\phi,>c}(D) ~|~ c_{\phi}(D) >c+\alpha] \nonumber \\
&=& \sum_{i=0}^{m-1} \Pr[c_{\phi}(D)-c +\eta +\alpha_{i} - \alpha<0 ~|~ c_{\phi}(D)-c -\alpha >0] \nonumber \\
&<& \sum_{i=0}^{m-1} \Pr[\eta_i <-\alpha_{i}] = \sum_{i=0}^{m-1} e^{-\alpha_{i}\epsilon_i/\|\workload\|_1}/2 =m \cdot \beta/(mL) = \beta/L \nonumber
\end{eqnarray}
As $|W|=L$, the failing probability is bounded by $\beta$.

The $\relaxprivacy$ Algorithm \protect\cite{DBLP:journals/corr/KoufogiannisHP15a}
correlates the new noise $\eta_{i+1}$ with noise $\eta_i$ from the previous iteration
In this way, the composition of the first $i+1$ iterations is $\epsilon_{i+1}$,
and the noise added in the $i+1$th iteration is equivalent to
a noise generated with Laplace distribution with privacy budget $\epsilon_{i+1}$
and the first $i$ iterations also satisfy $\epsilon_i$-DP for $i=0,1,\ldots,m-1$
(Theorem~9 for single-dimension and Theorem~10 for high-dimension~ \cite{DBLP:journals/corr/KoufogiannisHP15a}).
%This approach allows data cleaner to learn the query answer with a gradual relaxation of privacy cost.
\end{proof}

\eat{
\begin{algorithm}[t]
\caption{$\relaxprivacy(\eta,\epsilon,\epsilon')$ \protect\cite{DBLP:journals/corr/KoufogiannisHP15a}}
\label{algo:noisedown}
\begin{algorithmic}[1]
\Require Old noise sample $\eta$, old privacy budget $\epsilon$, new privacy budget $\epsilon'>\epsilon$
\Ensure  New noise $\eta_{new}$
%\Procedure {NoiseDown}{$\eta,\epsilon,\epsilon'$}
\State $pdf = [\frac{\epsilon}{\epsilon'}e^{-(\epsilon'-\epsilon)|\eta|}, \frac{\epsilon'-\epsilon}{2\epsilon'},\frac{\epsilon'+\epsilon}{2\epsilon'}(1-e^{-(\epsilon'-\epsilon)|\eta|})]$
\State $p = random.double()$
\If{$p\leq pdf[0]$}
    \State \textbf{return} $\eta_{new}=\eta$
\ElsIf{$p \leq sum(pdf[0:1])$}
    \State
    \[
      \text{Draw }   z =
      \begin{cases}
        e^{(\epsilon'+\epsilon)z},& \text{if } z\leq 0\\
        0,              & \text{otherwise}
      \end{cases}
    \]
   \State \textbf{return} $\eta_{new} = sgn(\eta)z$
\ElsIf{$p \leq sum(pdf[0:2])$}
    \State
    \[
      \text{Draw }   z =
      \begin{cases}
        e^{-(\epsilon'-\epsilon)z},& \text{if } 0 \leq z \leq |\eta|\\
        0,              & \text{otherwise}
      \end{cases}
    \]
   \State \textbf{return} $\eta_{new} = sgn(\eta)z$
\Else
     \State
    \[
      \text{Draw }   z =
      \begin{cases}
        e^{-(\epsilon'+\epsilon)z},& \text{if } z\geq |\eta|\\
        0,              & \text{otherwise}
      \end{cases}
    \]
   \State \textbf{return} $\eta_{new} = sgn(\eta)z$
\EndIf
%\EndProcedure
\end{algorithmic}
\end{algorithm}
}

\subsection{Laplace Top-$k$ Mechanism for \tcq}\label{app:ltm}
We would like show the privacy and tolerance requirement
of Laplace Top-$k$ Mechanism stated in Theorem~\ref{theorem:pc_ltm}.
The proof for accuracy is similar to the accuracy proof of Laplace mechanism
for \tcq (Appendix~\ref{app:lm_proof}).
\eat{
\begin{proof}
The output of Algorithm~\ref{algo:ltm} are $\{q_{\phi_{i_1}},\ldots,q_{\phi_{i_k}}\}$ which have the highest noisy counts. Suppose their noisy counts are in the order of $\tilde{x}_{i_1} \geq \tilde{x}_{i_2} \geq \cdots \geq \tilde{x}_{i_k}$. As $c_k$ is the true answer to the $k$th largest linear counting query of all $L$ linear counting queries, the largest answer to any $L-k+1$ linear counting queries should be no less than $c_k$. Hence, we have
\begin{equation}
\max_{j\in \{i_{k},\ldots,i_{L}\}} q_{\phi_{j}}(D) \geq c_k.
\end{equation}
Similarly, the smallest answer to any $k$ linear counting queries should be no greater than $c_k$. Hence, we have
\begin{equation}
\min_{j\in \{i_{1},\ldots,i_{k}\}} q_{\phi_{j}}(D) \leq c_k.
\end{equation}

First we would like to show that by setting
$b\leq \frac{\alpha}{2(\ln(L/(2\beta)))}$
if a count $q_{\phi}(D)$ is too small,
$\phi$ will be included in the answer $a$ with small probability:
\begin{eqnarray}
&&\Pr[\phi \in a ~|~ q_{\phi}(D) < c_k-\alpha] \nonumber \\
&=& \Pr[q_{\phi}(D)+\eta \geq \max_{j \in \{i_{k},\ldots,i_L\}} (q_{\phi_j}(D) + \eta_j) ~|~ q_{\phi}(D) < c_k-\alpha] \nonumber \\
&\leq & \Pr[q_{\phi}(D)+\eta \geq \max_{j \in \{i_{k},\ldots,i_L\}} (c_k + \eta_j) ~|~ q_{\phi}(D) < c_k-\alpha] \nonumber \\
&= & \Pr[q_{\phi}(D) - c_k \geq \max_{j \in \{i_{k},\ldots,i_L\}} \eta_j -\eta ~|~ q_{\phi}(D) < c_k-\alpha] \nonumber \\
&<& \Pr[\max_{j \in \{i_{k},\ldots,i_L\}} \eta_j -\eta < -\alpha] \nonumber \\
&\leq & \Pr[\max_{j \in \{i_{k},\ldots,i_L\}} \eta_j < -\alpha/2] + \Pr[\eta >\alpha/2] \nonumber \\
&\leq & (L-k+1) e^{-\frac{\alpha}{2b}}/2 + e^{-\frac{\alpha}{2b}}/2 \nonumber \\
&<& L e^{-\frac{\alpha}{2b}} \leq \beta/k
\end{eqnarray}
Hence, every selected $\phi\in a$ has $q_{\phi}(D) > c-\alpha$ with probability $1-\beta$ for $|a|=k$.

Then, we would like to show that if a count $q_{\phi}(D)$ is sufficiently large, then the probability of missing it in the output is small.
\begin{eqnarray}
&&\Pr[\phi \notin a ~|~ q_{\phi}(D) > c_k+\alpha] \nonumber \\
&=& \Pr[q_{\phi}(D)+\eta < \min_{j \in \{i_{1},\ldots,i_{k}\}} (q_{\phi_j}(D) + \eta_j) ~|~ q_{\phi}(D) > c_k+\alpha] \nonumber \\
&\leq & \Pr[q_{\phi}(D)+\eta < \min_{j \in \{i_{1},\ldots,i_k\}} (c_k + \eta_j) ~|~ q_{\phi}(D) > c_k+\alpha] \nonumber \\
&= & \Pr[q_{\phi}(D) - c_k < \min_{j \in \{i_{1},\ldots,i_k\}} \eta_j -\eta ~|~ q_{\phi}(D) > c_k+\alpha] \nonumber \\
&<& \Pr[\min_{j \in \{i_{1},\ldots,i_k\}} \eta_j -\eta > \alpha] \nonumber \\
&\leq & \Pr[\min_{j \in \{i_{1},\ldots,i_k\}} \eta_j >\alpha/2]
+ \Pr[\eta <-\alpha/2] \nonumber \\
&\leq & k e^{-\frac{\alpha}{2b}}/2 + e^{-\frac{\alpha}{2b}}/2 \nonumber \\
&<& k e^{-\frac{\alpha}{2b}} \leq \beta/L
\end{eqnarray}
Hence, all boolean formulae $\phi\in \{\phi_1,\ldots,\phi_L\}$ with $q_{\phi}(D)>c+\alpha$ are selected with probability $1-\beta$.
%With probability $1-\beta$, every linear counting query with true answer at least $c_k+\alpha$
%is selected, and every selected linear counting query has a true answer at least $c_k-\alpha$.
%\todo{double check \cite{Li:2012:PFI:2350229.2350251}}
\end{proof}
}%%%end of eat for accuracy proof
Then we show that Algorithm~\ref{algo:tcq-ltm} achieves
minimal $\epsilon$-differential privacy cost,
where $\epsilon=k/b=\frac{2k\ln (L/(2\beta))}{\alpha}$.
\begin{proof}
Fix $D=D'\cup \{t\}$. Let $(x_1,\ldots,x_L)$, respectively $(x'_1,\ldots,x'_L)$, denote the vector of answers to the set of linear counting queries $q_{\phi_1},\ldots,q_{\phi_L}$ when the table is $D$, respectively $D'$. Two properties are used: (1) Monotonicity of Counts: for all $j\in [L]$, $x_j \geq x'_j$; and (2) Lipschitz Property: for all $j\in [L]$, $1+x'_j \geq x_j$.

Fix any $k$ different values $(i_1,\ldots,i_k)$ from $[L]$, and fix noises $(\eta_{i_{k+1}},\ldots,\eta_{i_L})$ drawn from $Lap(k/\epsilon)^{L-k}$ used for the remaining linear counting queries.
Given these fixed noises, for $l\in \{i_1,\ldots,i_k\}$, we define
\begin{equation}
\eta^*_{l} = \min_{\eta}: (x_{l} + \eta > (\max_{j\in {i_{k+1},\ldots,i_L}} x_{j} + \eta_{j}))
\end{equation}
For each $l\in \{i_1,\ldots,i_k\}$, we have
\begin{eqnarray}
\lefteqn{x'_{l} + (1+\eta_{l}^*) = (1+x'_{l}) + \eta_{l}^* \geq x_{l}+\eta_{l}^* } \nonumber \\
&>&  \max_{j\in {i_{k+1},\ldots,i_L}} x_j+\eta_j \geq \max_{j\in {i_{k+1},\ldots,i_L}} x'_j+\eta_j \nonumber
\end{eqnarray}
Hence, if $\eta_{l}\geq r_{l}^*+1$ for all $l\in \{i_1,\ldots,i_k\}$, then $(i_1,\ldots,i_k)$ will be the output when the table is $D'$ and the noise vector is $(\eta_{i_1},\ldots,\eta_{i_k},\ldots,\eta_L)$. The probabilities below are over the choices of $(\eta_{i_1},\ldots,\eta_{i_k})\sim Lap(k/\epsilon)^k$.
\begin{eqnarray}
\lefteqn{\Pr[(i_1,\ldots,i_k) ~|~ D', \eta_{i_{i+1}},\ldots,\eta_{i_L}]} \nonumber \\
&\geq& \prod_{l\in \{ i_1,\ldots,i_k\}} \Pr[\eta_l \geq 1+\eta_l^*]  \
\geq \ \prod_{l\in \{ i_1,\ldots,i_k\}} e^{-k/\epsilon}\Pr[\eta_l \geq \eta_l^*] \nonumber \\
&\geq& e^{-\epsilon} \Pr[(i_1,\ldots,i_k) ~|~ D, \eta_{i_{i+1}},\ldots,\eta_{i_L}]
\end{eqnarray}
Proof of the other direction follows analogously. 

\eat{%%%%% eating the reverse direction of the proof
For the other direction, given these fixed noises, for $l\in \{i_1,\ldots,i_k\}$, we define
\begin{equation}
\eta^*_{l} = \min_{\eta}: (x_{l} + \eta > (\max_{j\in {i_{k+1},\ldots,i_L}} x'_{j} + \eta_{j}))
\end{equation}
For each $l\in \{i_1,\ldots,i_k\}$, we have
\begin{eqnarray}
&&  x_{l} + (1+\eta_{l}^*) \geq x'_{l} + (1+ \eta_{l}^*)   \nonumber \\
&>&  \max_{j\in {i_{k+1},\ldots,i_L}} (1+x'_j)+\eta_j \geq \max_{j\in {i_{k+1},\ldots,i_L}} x_j+\eta_j \nonumber
\end{eqnarray}
Hence, if $\eta_{l}\geq r_{l}^*+1$ for all $l\in \{i_1,\ldots,i_k\}$, then $(i_1,\ldots,i_k)$ will be the output when the table is $D$ and the noise vector is $(\eta_{i_1},\ldots,\eta_{i_k},\ldots,\eta_L)$. The probabilities below are over the choices of $(\eta_{i_1},\ldots,\eta_{i_k})\sim Lap(k/\epsilon)^k$.
\begin{eqnarray}
&&\Pr[(i_1,\ldots,i_k) ~|~ D, \eta_{i_{i+1}},\ldots,\eta_{i_L}] \nonumber \\
&\geq& \prod_{l\in \{ i_1,\ldots,i_k\}} \Pr[\eta_l \geq 1+\eta_l^*] \nonumber \\
&\geq& \prod_{l\in \{ i_1,\ldots,i_k\}} e^{-k/\epsilon}\Pr[\eta_l \geq \eta_l^*] \nonumber \\
&\geq& e^{-\epsilon} \Pr[(i_1,\ldots,i_k) ~|~ D', \eta_{i_{i+1}},\ldots,\eta_{i_L}]
\end{eqnarray}}
Therefore, $\epsilon$-differential privacy is guaranteed.
\end{proof}

\eat{
\subsection{\system Privacy Guarantee}\label{app:apex:proof}
\stitle{Proof for Theorem~\ref{thm:privacy}}
\begin{proof}
(1) directly follows from the definition of a valid transcript and these are the only kinds of transcripts an analyst sees when interacting with \system.

(2) can be shown as follows using induction.

\noindent{\underline{Base Case:}} When the transcript is empty, $Pr[\emptyset | D] \leq e^0 Pr[\emptyset | D']$.

\noindent{\underline{Induction step:}}
Now suppose for all $\trans_{i-1}$ of that encode valid \system transcripts of length $i-1$, $Pr[\trans_{i-1}| D] \leq e^{B_{i-1}} Pr[\trans_{i-1} | D']$. Let $\trans_i = \trans_{i-1} || [(q_i, \alpha_i, \beta_i), (\omega_i, \epsilon_i)]$ be a valid \system transcript of length $i$. Then:
\begin{eqnarray*}
\lefteqn{Pr[\trans_i | D] = Pr[\trans_{i-1} | D] Pr[[(q_i, \alpha_i, \beta_i), (\omega_i, \epsilon_i)] | D, \trans_{i-1}]}\\
& = & Pr[\trans_{i-1} | D] Pr[\mathbb{C}(\trans_{i-1}) = (q_i, \alpha_i, \beta_i)] Pr[M_i(D) = (\omega_i, \epsilon_i)]
\end{eqnarray*}
Note that the analyst's choice of query $q_i$ and accuracy requirement depends only on the transcript $\trans_{i-1}$ and not the sensitive database, and thus incurs no privacy loss (from Theorem~\ref{theorem:post}). Thus, it is enough to show that
\[Pr[M_i(D) = (\omega_i, \epsilon_i)] \leq e^{\epsilon_i} Pr[M_i(D') = (\omega_i, \epsilon_i)]\]
\noindent{\underline{Case 1:}}
When $\omega_i\neq \bot$ and $M_i$ is LM, WCQ-SM, ICQ-SM, or TCQ-LTM, the mechanism satisfies $\upperbound_i$-DP and $\epsilon_i = \upperbound_i$. Therefore, $Pr[M_i(D) = (\omega_i, \epsilon_i)] \leq e^{\epsilon_i} Pr[M_i(D') = (\omega_i, \epsilon_i)]$.

\noindent{\underline{Case 2:}}
When $\omega_i \neq \bot$ and $M_i$ is ICQ-MPM, the mechanism satisfies $\upperbound_i$-DP across all outputs. However, when either mechanism outputs $(\omega_i, \epsilon_i)$, for $\epsilon_i < \upperbound_i$, we can show that $Pr[M_i(D) = (\omega_i, \epsilon_i)] \leq e^{\epsilon_i} Pr[M_i(D') = (\omega_i, \epsilon_i)]$. In the case of ICQ-MPM, if the algorithm returns in Line~11 after $i$ iterations of the loop, the noisy answer is generated by a DP algorithm with privacy loss $\epsilon_i = \frac{j}{m}\upperbound_i$.

\noindent{\underline{Case 3:}}
Finally, when $\omega_i = \bot$  (i.e., the query is declined), the decision to decline depends on $\upperbound_i$ of all mechanism applicable to the query (which is independent of the data) rather than $\epsilon_i$ (which could depend on the data in the case of ICQ-MPM). Therefore, $Pr[M_i(D) = (\omega_i, \epsilon_i)]  = Pr[M_i(D') = (\omega_i, \epsilon_i)]$ for all $D, D'$. The proof would fail if the decision to deny a query depends on $\epsilon_i$.
\end{proof}
}

\eat{
\subsection{Data Dependent Translation for \lcc}\label{app:lccproof}
We first provide full proof for Theorem~\ref{theorem:poking_lcm}:
given a \lcc query $q_{\phi,>c}$,
for any table $D\in \tabledomain$, LCM with Poking
(Algorithm~\ref{algo:poking_lcm}),
denoted by $LCMP_{q_{\phi,>c}}^{\alpha,\beta}$
achieves $(\alpha,\beta)$-\lcc tolerance,
and satisfies $\epsilon$-differential privacy cost,
where $\epsilon=\frac{(1+f)\ln(1/\beta)}{\alpha}$.
\begin{proof}
The probability to fail the tolerance requirement is
(i) when $q_{\phi}(D)<c-\alpha$,
\begin{eqnarray}
&& \Pr[LCMP^{\alpha,\beta}_{q_{\phi,>c}}(D)=\text{True} ~|~ q_{\phi}(D)<c-\alpha] \nonumber \\
&=& \Pr[q_{\phi}(D)-c +\eta -\alpha_{0} + \alpha >0 ~|~q_{\phi}(D)-c +\alpha <0]\nonumber \\
&& + \Pr[LCM^{\alpha,\beta/2}_{q_{\phi,>c}}(D)=\text{True} ~|~ q_{\phi}(D)<c-\alpha] \nonumber \\
&<& \Pr[\eta >\alpha_{0}] + \beta/2 = e^{-\alpha_{0}\epsilon_0}/2 +\beta/2 = \beta.
\end{eqnarray}
and (ii) when $q_{\phi}(D)<c+\alpha$,
\begin{eqnarray}
&& \Pr[LCMP^{\alpha,\beta}_{\phi,>c}(D)=\text{False} ~|~ q_{\phi}(D) >c+\alpha] \nonumber \\
&=& \Pr[q_{\phi}(D)-c +\eta +\alpha_{0} - \alpha<0 ~|~ q_{\phi}(D)-c -\alpha >0] \nonumber \\
&& + \Pr[LCM^{\alpha,\beta/2}_{q_{\phi,>c}}(D)=\text{False} ~|~ q_{\phi}(D)>c+\alpha] \nonumber \\
&<& \Pr[\eta_{pp} <-\alpha_{0}]  + \beta/2 = e^{-\alpha_{0}\epsilon_0}/2 + \beta/2 =\beta
\end{eqnarray}
Hence, LCMP achieves $\alpha$-\lcc with probability $1-\beta$.
The poking part spends a privacy budget of $\epsilon_{0}$
and the second part using LCM spends a budget of $\epsilon_{LCM}$ Theorem~\ref{theorem:pc_lcm}.
By composition of differential privacy (Theorem~\ref{theorem:seq}),
LCMP satisfies $(\epsilon_0+\epsilon_{LCM})$-differential privacy.
\end{proof}

For LCM with Multi-Poking, we show proof sketch for Theorem~\ref{theorem:mpoking_lcm}
that given a \lcc query $q_{\phi,>c}$,
for any table $D\in \tabledomain$, LCM with Multi-Poking (Algorithm~\ref{algo:mpoking_lcm}), denoted by $LCMMP^{\alpha,\beta}_{q_{\phi,>c}}$,
achieves $(\alpha,\beta)$-\lcc tolerance
with $\epsilon$-differential privacy,
where $\epsilon=\frac{\ln(m/(2\beta))}{\alpha}$.
\begin{proof}(sketch)
The NoiseDown Algorithm shown in Algorithm~\ref{algo:noisedown}
correlates the new noise $\eta_{i+1}$ with noise $\eta_i$ from the previous iteration
In this way, the composition of the first $i+1$ iterations is $\epsilon_{i+1}$,
and the noise added in the $i+1$th iteration is equivalent to
a noise generated with Laplace distribution with privacy budget $\epsilon_{i+1}$
and the first $i$ iterations also satisfy $\epsilon_i$-DP for $i=0,1,\ldots,m-1$
(Theorem~9 \cite{DBLP:journals/corr/KoufogiannisHP15a}).
This approach allows data cleaner to learn the query answer with a gradual relaxation of privacy cost.
At $i$th iteration, the probability to fail is $\beta/m$.
Hence, when outputting an answer at $i$th iteration,
the probability to fail the requirement is $i\beta/m <\beta$.
The proof for the tolerance is similar to the proof of Theorem~\ref{theorem:poking_lcm}.
\end{proof}
}

\section{Composition Theorems}\label{sec:seq}
The sequential composition theorem helps assess the privacy loss of multiple differentialy private  mechanisms.
\vspace{-.1em}
\begin{theorem}[Sequential Composition  \cite{Dwork:2014:AFD:2693052.2693053}] \label{theorem:seq}
	Let $M_1(\cdot)$ and $M_2(\cdot, \cdot)$ be algorithms with independent sources of randomness that ensure $\epsilon_1$- and $\epsilon_2$-differential privacy respectively.
	An algorithm that outputs both $M_1(D) = O_1$ and $M_2(O_1, D) = O_2$
	ensures $(\epsilon_1 + \epsilon_2)$-differential privacy.
\end{theorem}

%Given a sequence of data independent translation mechanisms $M_1,M_2,\ldots,M_i$ shown in Section~\ref{sec:translation}
%that satisfy differential privacy with cost of $\epsilon_1, epsilon_2,\ldots,\epsilon_i$ respectively,
%$B_i$ are $(\epsilon_1+\cdots+\epsilon_i)$.
%For data dependent translation mechanisms,
%\system  considers the worst privacy loss of $M_i$.
%For instance, if $M_i$ is multi-poking mechanism (Algorithms~\ref{algo:mpm}),
%the estimated privacy is obtained by incrementing the privacy loss $B_{i-1}$ by the worst privacy cost outputted from $\transFunc().\upperbound$
%while the actual privacy loss $B_i$ is obtained by increasing $B_{i-1}$ by the outputted privacy cost $\epsilon_i$ of $\runFunc()$, which lies between the upper bound and lower bound specified by $\transFunc()$.

Another important property of differential privacy
is that  postprocessing the outputs  does not degrade privacy.
\begin{theorem}[Post-Processing \cite{Dwork:2014:AFD:2693052.2693053}]\label{theorem:post}
	Let $M_1(\cdot)$ be an algorithm that satisfies $\epsilon$-differential privacy. If applying an algorithm $M_2$ the output of $M_1(\cdot)$, then the overall mechanism $M_2 \circ M_1(\cdot)$ also satisfies $\epsilon$-differential privacy.
\end{theorem}
All steps in the post-processing algorithm do not access the raw data, so they do not affect the privacy analysis.
Hence, any smart choices made by the data analyst after receiving the noisy answers from the privacy engine
are considered as post-processing steps and do not affect the privacy analysis.

\eat{
\subsection{Vary ICQ-MPM Poking Steps}\label{app:icq_mpm}

\begin{figure}[h]
	\centering
	\includegraphics[width=0.3\textwidth]{../../figure/experiments/cost_mpm}%
	\caption{Incraseing ICQ-MPM max poking step increases the estimated privacy cost from \system, but the real privacy cost in fact decreases because on average, per poking costs less. \am{Can cut this for space if needed.}}\label{fig:cost_mpm}
\end{figure}

Figure~\ref{fig:cost_mpm} studies the privacy cost as a factor of max poking steps allowed in ICQ-MPM using $QI4$. 
With the increasing of max poking steps, on the one side, the estimated cost from privacy translator increases, obviously because allowing more poking implies more privacy leakage. On the other side, the average cost per poking becomes smaller; therefore, given a fixed workload and ICQ counting threshold, the real cost per poking in fact decreases.
}

\ifpaper
%omit cleaner model
\else
\full{
\section{Cleaner Model}\label{sec:cleaner_model}
The cleaning engineer typically narrows down the choices through a sequence of actions, consisting of issuing queries and making choices based on the answer of issued query. We use a {\em strategy} to denote a class of actions that use the same set of queries but make different choices.
Figure~\ref{fig:blockstrategies} shows the strategies for the blocking task. These two strategies are different as they use different query types, though they share the same criteria of decision choices. In particular, the queries $q1,q5$ in the strategy shown in Figure~\ref{fig:blockstrategy1}
are \wcq, while the queries $q1',q5'$ in Figure~\ref{fig:blockstrategy2} are \tcq and \icq respectively. This case exemplifies how a cleaning engineer constructs a single path (disjunction) of predicates to form a blocking function, though a real function can be more complex~\cite{DBLP:journals/tkde/ElmagarmidIV07}. Similarly, Figure~\ref{fig:matchstrategies} illustrates two matching strategies using \wcq (Figure~\ref{fig:matchstrategy1}) and \icq/\tcq (Figure~\ref{fig:matchstrategy2}), respectively, where the matching function is formed as a conjunction of predicates.

\begin{figure}[t]
\centering
\begin{subfigure}[t]{0.5\textwidth}
        \includegraphics[width=\textwidth]{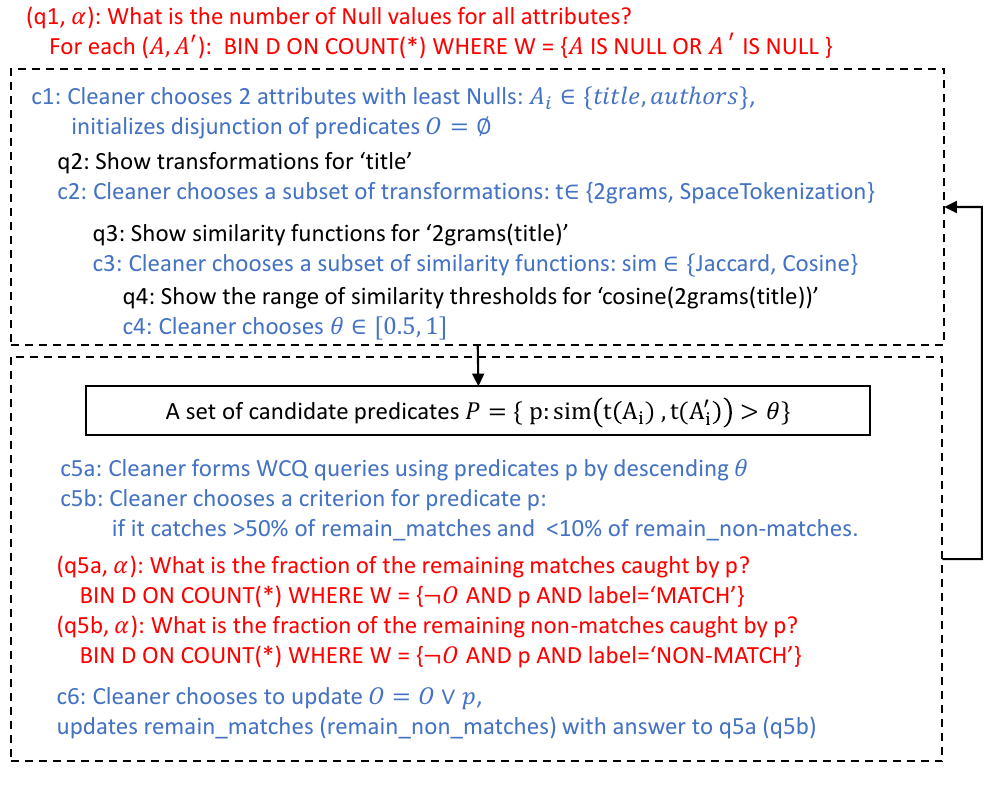}
        \caption{Strategy instance 1 for blocking}\label{fig:blockstrategy1}
\end{subfigure}
\begin{subfigure}[t]{0.5\textwidth}
        \includegraphics[width=\textwidth]{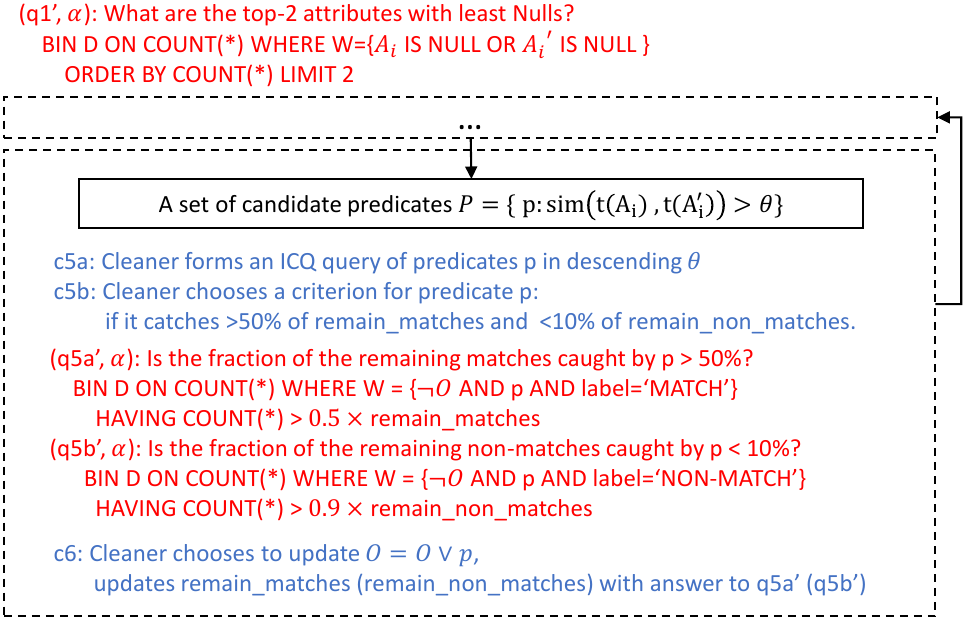}
        \caption{Strategy instance 2 for blocking}\label{fig:blockstrategy2}
\end{subfigure}
\caption{Two strategies for blocking}\label{fig:blockstrategies}
\end{figure}

\begin{figure}[h]
\centering
\begin{subfigure}[t]{0.5\textwidth}
        \includegraphics[width=\textwidth]{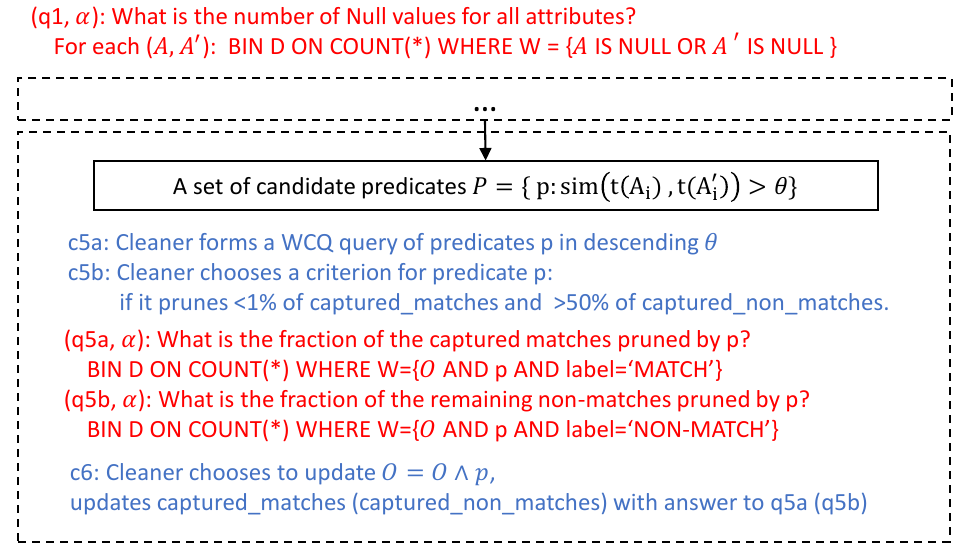}
        \caption{Strategy instance 1 for matching (MS1)}\label{fig:matchstrategy1}
\end{subfigure}
\begin{subfigure}[t]{0.5\textwidth}
        \includegraphics[width=\textwidth]{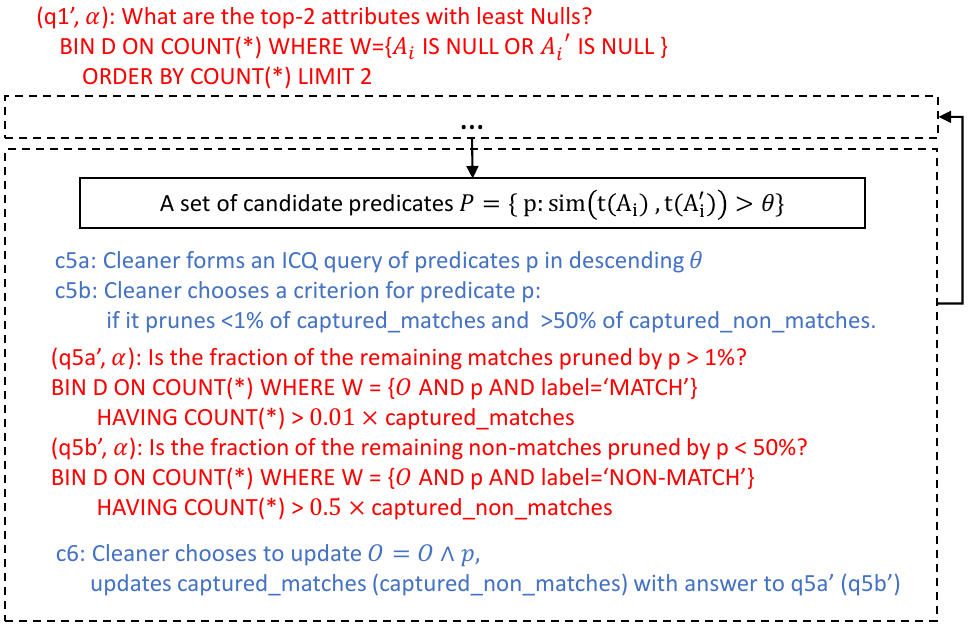}
        \caption{Strategy instance 2 for matching (MS2)}\label{fig:matchstrategy2}
\end{subfigure}
\caption{Two strategies for matching}\label{fig:matchstrategies}
\end{figure}

The cleaner model encodes the space of all the parameters involved in the cleaner's decisions $c1$-$c6$. Table~\ref{tab:cleanermodel1} summarizes the space of all the parameters for $c1$-$c6$ in blocking strategy 1. From $c1$ to $c4$, the program chooses 
(1) a subset of  attributes $x_1$ of size ranging from $2$ up to the total number of attributes $|\attrlist|$, 
(2) a subset of transformations $x_2$ from  $\trans = \{2grams, 3grams, SpaceTokenization\}$, 
(3) a subset of similarity functions $x_3$ from $\simf = \{Edit, SmithWater, Jaro, Cosine, Jaccard, Overlap, Diff\}$, and 
(4) $x_6$ thresholds from the range of $[x_4,x_5]$, where $x_4\in (0,0.5)$, $x_5\in(0.5,1)$, and $x_6\in \{2,3,4,5,6\}$. 
The cross product of these choices forms a set of predicates $P$, and $c5a$ picks an ordering $x_7$, one of the permutation of $P$. In $c5b$, the model sets the criterion for pruning or keeping a predicate $p$ from the top list of $P$. In particular, the model sets $x_8$ and $x_9$ as the minimum fraction of the remaining matches caught and the maximum fraction of the remaining non-matches caught by  $p\lor O$ respectively, where $x_8\in[0.2,0.5]$ and $x_9\in [0.1,0.2]$. These values are reset as $x_8=x_8/x_{10}$ and $x_9=x_9x_{10}$ where $x_{10}\in \{2,3\}$ if all predicates have been checked but $O=\emptyset$. In $c6$, the model considers three possible styles of cleaners on trusting the noisy answers:  \emph{neutral} style corresponds to trust the noisy answers; for \emph{optimistic} (\emph{pessimistic}) style, the cleaner trusts the values by adding (subtracting) $\alpha/5$ to (from) the noisy answers. 
If the criterion is met and the blocking cost over $D_t$ is less than a fixed cutoff threshold (e.g., a hardware constraint, we set $550$ for {\tt citations} dataset). An instance of all variables $\cleanermodel = (x_1,\ldots,x_{11})$ in Table~\ref{tab:cleanermodel1} forms a concrete cleaner. The model for other strategies is similarly constructed.

\begin{table}[t]
  	\center
	\caption{A Cleaner Model for Blocking Strategy 1}
	\label{tab:cleanermodel1}
	\begin{tabularx}{\columnwidth}{l|X} \hline
   	 \multicolumn{2}{c}{A cleaner model $\cleanermodel=\{x_1,\ldots,x_{11}\}$ for c1-c6 in Figure~\ref{fig:blockstrategy1}}  \\ \hline
 	$c1$ & Choose $x_1$, an ordered subset of attributes with least Nulls, where $x_1\in \{2, 3, |\attrlist|\}$ \\ \hline
	$c2$ & Choose $x_2$, an ordered subset of transformations from $\trans$, where $|x_2|\in \{1,2,3\}$  \\ \hline
	$c3$ & Choose $x_3$, an ordered subset of similarity functions from $\simf$, where $|x_3|\in \{2,3,4,5,6\}$ \\ \hline
	$c4$ & Choose  a lower bounds $x_4$ and a upper bound $x_5$ of threshold range, where $x_5\in (0.5,1)$, $x_4\in (0,0.5)$, and evenly choose $x_6$ thresholds in order of either ASC or DSC, and $x_6 \in \{2,3,4,5,6\}$\\ \hline
   $c5a$ & Predicate $p$ is sequentially selected based the order of $x_{7} \in \text{Permute}(x_1 \times x_2 \times x_3 \times x_6)$
        \\ \hline
       $c5b$ & Chooses a criterion for predicate $p$: if it catches $<x_8$ fraction of the remaining matches and $>x_9$ fraction of remaining non-matches, where $x_8 \in [0.2,0.5]$ and $x_9\in [0.1,0.2]$. Reset $x_8=x_8/x_{10}$ and $x_9=x_{10}x_9$, $x_{10} \in \{2,3\}$ if all queries have been asked but $O=\emptyset$.
       \\ \hline
	$c6$ & Choose method $x_{11}\in\{$neutral, optimistic, pessimistic$\}$ to take tolerance into account. If conditions are met, and the blocking cost is less than cutoff threshold, add $p$ to the output $O$ and remove it from $P$.
  \\ \hline
	\end{tabularx}
\end{table}
}\fi

%\section{Relationship to Accuracy-Constrained ERM}
%\label{app:erm}

\eat{The accuracy-first approach in accuracy-constrained ERM~\cite{accuracyfirst:nips17}
solves the translation problem
by first computing the empirical error in the output of a mechanism
and then differentially private testing if the empirical error $\hat{\alpha}$ is bounded by the accuracy constraint $\alpha$.
The problems considered by this approach typically have an empirical error with small sensitivity.
This is not true for many queries we considered.
For example, for queries with prefix workload $\workload$, the empirical error $\hat{\alpha}$ has the same high sensitivity as $\|\workload\|_1$.
Similarly, for an \icq, the sensitivity to the empirical error in the output can be as large as the workload size.
Hence, we choose to use the theoretical property of the noise
to translate accuracy requirement to privacy cost in \system.}

\ifpaper
\conf{
\section{Additional Evaluation}\label{sec:f1}
To further understand the relation between our accuracy requirement with commonly used error metric, we use F1 score to measure the similarity between the correct answer set and noisy answer set outputted from the mechanism. Figure~\ref{fig:f1-qi4-qt1} presents the F1 score of two queries: QI4 of an \icq, and QT1 of a \tcq. In QT1, when $\alpha$ from $0.02|D|$ to $0.64|D|$, the median F1 score decreases from $0.9$ to $0.15$, which has a steeper gradient than the changes in Figure~\ref{fig:cost_accuracy_all} and a closer trend with the theoretical $\alpha$. In QI4, the F1 score is even more consistent with $\alpha$ and the empirical error shown in Figure~\ref{fig:cost_accuracy_all}. This shows  that our $(\alpha,\beta)$ accuracy  requirement is still a good indicator in data exploration process.

\begin{figure}[h]
	\centering
        \includegraphics[width=0.45\textwidth]{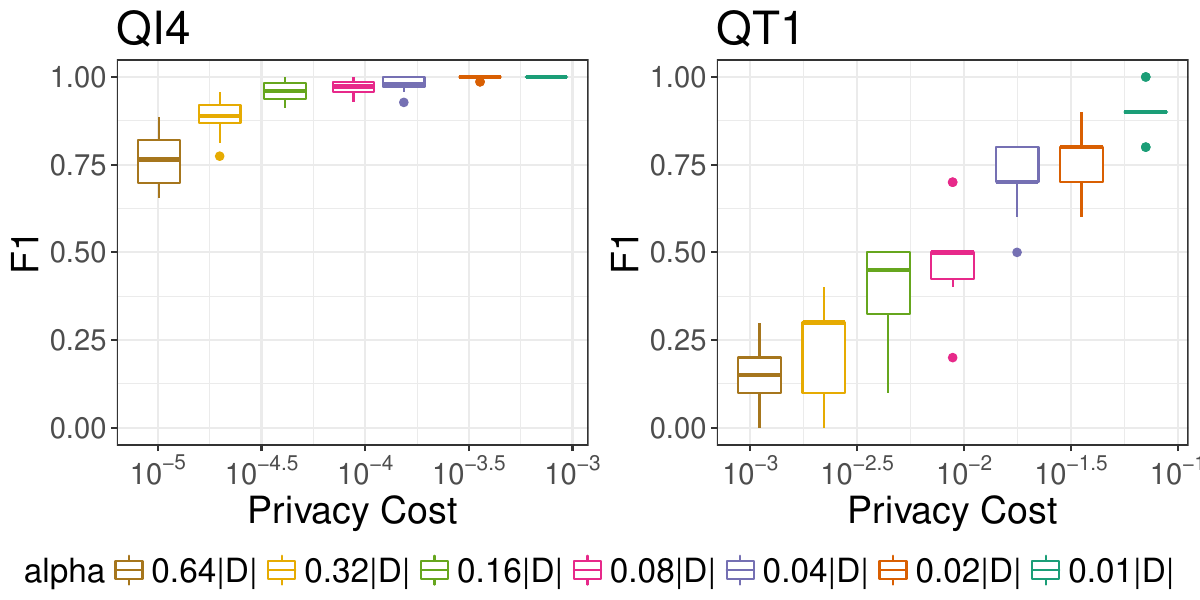}%
	\caption{Privacy cost and empirical accuracy ( F1 score) using optimal mechanism chosen by APEx (Optimistic Mode) on QI4 and QT1.}\label{fig:f1-qi4-qt1}
\end{figure}

}\else
%%f1 part is in main body
\fi

%!TEX root=./main.tex
\section{Related Work}\label{sec:related}

\stitle{DP in practice.} Differential privacy~\cite{DBLP:conf/tcc/DworkMNS06} has emerged as 
a popular standard in real-world analysis products~\cite{Erlingsson14Rappor, Greenberg16Apples, haney17:census, Johnson:2013:PDE:2487575.2487687, Li:2014:DWA:2732269.2732271, machanavajjhala08onthemap}. There are well known techniques for special tasks such as answering linear counting queries\cite{Hay:2016:PED:2882903.2882931, Li:2015:MMO:2846574.2846647, hdmm18}, top-k queries~\cite{Lee:2014:TFI:2623330.2623723} and frequent itemset mining~\cite{Li:2012:PFI:2350229.2350251,Zeng:2012:DPF:2428536.2428539}. However, all the proposed work focuses on a particular type of query and there is no clear winner between these proposed algorithms for a given query. \system considers the key state-of-the-art differentially private algorithms~\cite{Li:2015:MMO:2846574.2846647, Dwork06differentialprivacy, Li:2014:DWA:2732269.2732271} and makes decision on the algorithms on behalf of the data analyst.

General purpose frameworks like PINQ \cite{McSherry:2009:PIQ:1559845.1559850}, wPINQ \cite{Proserpio:2014:CDS:2732296.2732300} and $\epsilon$ktelo \cite{ektelo} allow users to write programs in higher level languages for various tasks. These systems automatically prove that every program written in this framework ensures differential privacy. FLEX \cite{DBLP:journals/corr/JohnsonNS17} allows users to answer a SQL query under a specified privacy budget. Unlike all these systems, \system is the first that allows analysts ask a sequence of queries in high level language and specify only an accuracy bound, rather than a privacy level. \system automatically computes the privacy level to match the accuracy level.

\eat{
Differential privacy has been widely adopted in organizations~\cite{haney17:census,machanavajjhala08onthemap,Vilhuber17Proceedings, Erlingsson14Rappor, Greenberg16Apples, Johnson:2013:PDE:2487575.2487687}, and is perfectly compatible with privacy laws (e.g., GDPR~\cite{GDPR}) which can be used to quantify the privacy loss that users only need to consent, or as a maximum budget to be spent without asking users consent.
}

\balance
\stitle{Accuracy constraints in DP.} Our problem is similar in spirit to Ligett et al.~\cite{accuracyfirst:nips17}, which also considers analysts who specify accuracy constraints.
%While their work is theoretical with a focus on private machine learning, our focus is to build a query answering system for exploring private data with accuracy guarantees. 
Rather than concentrating on private machine learning theoretically, our focus is to explore private data with accuracy guarantees. 
The main technical differences are highlighted below.

%First, unlike Ligett et al., \system does not need to take a set of privacy budgets as an input. For the exploration queries in \system, there are more than just one differentially private mechanisms for each query type, and \system automatically picks the optimal. Second, for the queries and mechanisms supported by \system, using Ligett et al's methods would be overkill due to (a) the extra privacy cost $\epsilon_0$ to test the empirical error, and (b) high sensitivity in many exploration queries.

First, the end-to-end problems in \system and the approach in Ligett et al.~\cite{accuracyfirst:nips17} are different. \system aims to translate a \emph{given query with accuracy bound}, $(q,\alpha,\beta)$, to a differentially private mechanism that achieves this accuracy bound with the minimal privacy cost, $(M, \epsilon)$. On the other hand, Ligett et al.~\cite{accuracyfirst:nips17} takes as input a \emph{given mechanism with accuracy bound and a set of privacy budgets}, $(M, \alpha,\beta, \{\epsilon_1<\epsilon_2<\cdots <\epsilon_T\})$, and outputs the minimal privacy cost (epsilon) for $M$ to achieve the accuracy bound. Unlike Ligett et al., \system does not need to take a set of privacy budgets as an input. For the exploration queries in \system, there are more than just one differentially private mechanisms for each query type, and none of these mechanisms dominate the others. Thus, in this sense, the problem solved by \system is more general than the one solved by Ligett et al.

The second key difference between the approaches is the following. \system currently only supports mechanisms for which the relationship between the accuracy bound and the privacy loss epsilon can be established analytically. On the other hand, Ligett et al. can handle arbitrarily complex mechanisms, and use empirical error of the mechanisms to pick the epsilon. In this way, the solution proposed by Ligett et al. applies a larger class of mechanisms. Nevertheless, for the queries and mechanisms supported by \system, using Ligett et al.'s methods would be overkill in two ways: (1) there is an extra privacy cost $\epsilon_0$ to test the empirical error, (2) since the exploration queries are variations of counts, the sensitivity of the error will be so high that the noise introduced to the empirical error could limit our ability to distinguish between different epsilon values. For example, for a simple counting query of size 1, the maximum privacy cost required in \system is $\ln(1/\beta)/\alpha$ while the privacy cost with Ligett et al. is more than $\epsilon_0=16(\ln(2T/\beta))/\alpha$. When \system applies a data-dependent approach, the privacy cost can be even smaller than $\ln(1/\beta)/\alpha$. For a prefix counting query of size $L$ and sensitivity of $L$, the best privacy cost achieved in \system is $O(\log L)$, but $\epsilon_0$ in [22] is $O(L)$ as the sensitivity of the error to the final query answer is $L$ (for both Laplace and strategy-based mechanisms).  Similarly, for iceberg counting queries (\icq), the sensitivity of error to the final query answer is large and results in a large $\epsilon_0$ for the differentially private testing. Thus, using the method in Ligett et al. would not help the currently supported mechanisms and queries in \system. In the future, we will add more complex mechanisms into \system (like DAWA~\cite{Li:2014:DWA:2732269.2732271}, MWEM~\cite{Hardt:2012:SPA:2999325.2999396} that add data dependent noise) and study whether the methods of Ligett et al. can be adapted to our setting.

%%%% EATING CLEANING RELATED WORK
\eat{
%Existing data exploration tools and research mainly focus
%improving the effectiveness and efficiency of data exploration.
To support data exploration, many useful tools such as
Trifacta~\cite{DBLP:conf/chi/KandelPHH11} and Tableau~\cite{polaris02}
have been developed and followed by a number of research directions.
These research work mainly focus on (i) enhancing the effectiveness of data exploration
in identifying interesting or relevant data items or queries~\cite{DBLP:conf/cidr/SellamK13,Fan:2011:ISQ:2004686.2005644}
or (ii) improving the performance of data exploration in answering exploratory queries efficiently
via redesigning database systems with adaptive storage or indexes~\cite{Alagiannis:2014:HHA:2588555.2610502,Halim:2012:SDC:2168651.2168652}
or via building query optimizations over the database engine~\cite{Cormode:2012:SMD:2344400.2344401,Agarwal:2014:KYW:2588555.2593667,
Agarwal:2013:BQB:2465351.2465355,aqpplus16,revisitaqp17, Tauheed:2012:SPL:2350229.2350267, Kalinin:2014:IDE:2588555.2593666}.
In particular, approximate query processing techniques~\cite{Cormode:2012:SMD:2344400.2344401,Agarwal:2014:KYW:2588555.2593667,
Agarwal:2013:BQB:2465351.2465355,aqpplus16,revisitaqp17} trade-off query accuracy for better performance,
where the analyst can explore the data with tolerable noisy answers.
However, these approaches do not support data exploration over sensitive data,
where individual records in the data needs privacy protection (e.g. health records in a hospital).
These settings either force the analyst to give up the opportunity of exploring the data by treating the dataset as a black box,
or force the data owner to give up the privacy even if the records are encrypted.
}

\eat{Differential privacy~\cite{DBLP:conf/tcc/DworkMNS06} has emerged as the state-of-the-art privacy guarantee
that permits noisy but accurate answers to aggregate queries, while provably bounding the information leakage about any single record in the database.
This model allows the data owner to easily collaborate with external data analysts on exploring data.
However, existing approaches using differential privacy for data exploration
focus on one specific domain~\cite{Johnson:2013:PDE:2487575.2487687, visdpt16} and hence are not extensible to general domains.
There are general systems for differentially private query answering over relations such as
PINQ \cite{McSherry:2009:PIQ:1559845.1559850}, wPINQ \cite{Proserpio:2014:CDS:2732296.2732300} and
elastic sensitivity \cite{DBLP:journals/corr/JohnsonNS17},
but it is also hard for data owners or the data analysts who are non-privacy experts to use them.
However, \system is the first system where the data analyst specifies accuracy constraints on queries that are translated into privacy losses by the system.
There is concurrent work by Ligett et al \cite{accuracyfirst:nips17}
that considers analysts who specify accuracy constraints for machine learning tasks.
\system uses theoretical properties of the mechanisms to perform the error to accuracy translation,
while Ligett et al develop a mechanism to choose the privacy level based on the actual error.}

\stitle{Data cleaning on private data.} In our case study, we use \system as a means to help tune a cleaning workflow (namely entity resolution) on private data.  Amongst prior work on cleaning private data \cite{Chiang:2018:IPS:3208074.3190577, private_clean}, the most relevant is PrivateClean \cite{private_clean} as it uses differential privacy too. However, similarities to \system stop with that. First, PrivateClean assumes a different setting, where no active data cleaner is involved. The data owner perturbs the dirty data without cleaning it, while the data analyst who wants to obtain some aggregate queries will clean the perturbed dirty data. Moreover, all the privacy perturbation techniques in PrivateClean are based on record-level perturbation, which (1) only work well for attributes with small domain sizes, and (2) has asymptotically poorer utility for aggregated queries. At last, the randomized response technique used for sampling string attributes in PrivateClean does not satisfy differential privacy -- it leaks the active domain of the attribute.

%These cleaning steps can be considered as post-processing steps,
%and hence incur no further privacy loss.
\eat{However, to clean the private dirty data,
the data analyst should have the prior knowledge on
the suitable set of transformations required for the given data.
In our setting, the set of transformations for cleaning the given data
are the output of the exploration process by the cleaning engineer,
and hence should not be known in advance.}
\eat{On the other hand, though InfoClean~\cite{Chiang:2018:IPS:3208074.3190577}
shows better cleaning quality than PrivateClean,
this approach uses a bounded information disclosure across attributes as a privacy measure,
which is different from differential privacy used by PrivateClean or \system.
There is also no guarantee on the overall privacy loss if
two cleaning algorithms with different objective functions are applied the same data,
and is susceptible to reconstruction attacks.}

\eat{
Tools for data cleaning and data integration
\cite{DBLP:conf/chi/KandelPHH11,DBLP:conf/cidr/StonebrakerBIBCZPX13,
DBLP:conf/vldb/RamanH01, DBLP:conf/cidr/AbedjanMGISPO15, DBLP:conf/sigmod/Galhardas00} have been developed for enterprises to unify, clean, and transform their data. These tools are inapplicable for cleaning private data.

In our setting, the set of transformations for cleaning the given data
are the output of the exploration process by the cleaning engineer,
and hence should not be known in advance.
Moreover, all the privacy perturbation techniques in PrivateClean
are based on record-level perturbation,
which (a) only work well for attributes with small domain sizes, and (b)
has asymptotically poorer utility for aggregated queries.
Moreover, the randomized response technique used for sampling string attributes in PrivateClean does not satisfy differential privacy -- it leaks the active domain of the attribute. In experiments not shown, we reimplement PrivateClean so that it satisfies DP. Its quality was consistently poor (0.5 recall for Blocking and $<0.8$ F-1 for Matching on {\tt restaurants}) for all privacy levels considered in this paper, while \system is able to achieve high quality for reasonable privacy levels.
}

\stitle{Relationship to Privacy Laws.} \system requires the data owner to only specify the overall privacy budget $B$. This setup may support and realize privacy laws in the future (e.g. GDPR~\cite{GDPR}) to manage the privacy loss of users where $B$ can be seen as the maximum budget to be spent without asking users' consent.

\section{Queries Supported by \system}\label{app:otherqueries}
Besides supporting linear counting queries, our algorithms in \system can be easily extended to support SUM() queries. Queries for MEDIAN() (and percentile) can be supported by first querying the CDF (using a \wcq), and finding the median from that. GROUPBY can be expressed as a sequence of two queries: first query for values of an attribute with COUNT(*) > 0 (using \icq), and then query the noisy counts for these bins (using \wcq). Similarly, if the aggregated function $f()$ in HAVING differs from the aggregated function $g()$, \system can express it as a sequence of two queries: first query for bins having $f()>c$ (using \icq), and then apply $g()$ on these bins (using \wcq).

However, no differentially private algorithms can accurately report MAX() and MIN(). Non-linear queries such as AVG() and STD() can have very sensitive error bounds to noise when computed on a small number of rows (like the measures discussed in Section~\ref{sec:accuracy}). Hence, supporting non-linear queries would need new translation mechanisms. Moreover, \system can support queries with foreign key joins (that do not amplify the contribution of one record), but designing accurate differentially private algorithms for queries with general joins is still an active area of research.

\end{document}